%% file: mainJournal-arxiv.tex
\documentclass[11pt]{article}

\usepackage[a4paper, total={6in, 8in}]{geometry}

\input{imports.tex}

\usepackage{wrapfig}

\begin{document}

\input{macros.tex}

\newcommand{\figSizeBig}{1.53}
\newcommand{\figSize}{0.89}
\newcommand{\figSizeSmall}{0.715}
\newcommand{\figSizeJournal}{0.818} 




\newcommand{\VCalgSize}{scriptsize} 



\title{\emph{Renaissance}: A Self-Stabilizing Distributed SDN Control Plane using In-band Communications}


\author{Marco Canini$^1$ 
	Iosif Salem$^2$ 
	Liron Schiff$^3$ 
	Elad M.\ Schiller$^2$ 
	Stefan Schmid$^4$\\
	{\small $^1$ Computer Science, Universit\'{e} catholique de Louvain, Belgium} \\
	{\small  \texttt{marco.canini@acm.org}}\\ 
	{\small $^2$ Computer Science and Engineering, Chalmers University of Technology, Sweden} \\
	{\small \texttt{\{iosif, elad\}@chalmers.se}}\\
	{\small $^3$ Akamai, Israel} \\
	{\small \texttt{liron.schiff@guardicore.com}}\\
	{\small $^4$ Faculty of Computer Science, University of Vienna, Austria \& TU Berlin, Germany}  \\
	{\small \texttt{stefan\_schmid@univie.ac.at}}
	}

%
%
%
%
%




\date{}

\maketitle

\begin{abstract}
\remove{By introducing programmability, automated verification, and innovative debugging tools, Software-Defined Networks (SDNs) are poised to meet the increasingly stringent dependability requirements of today's communication networks. However, the design of fault-tolerant SDNs remains an open challenge.
This paper considers the design of dependable SDNs through the lenses of \emph{self-stabilization}---a very strong notion of fault-tolerance. In particular, we develop algorithms for an in-band and distributed control plane for SDNs, called \emph{Renaissance}, which tolerate a wide range of (concurrent) controller, link, and communication failures. Our self-stabilizing algorithms ensure that after the occurrence of an arbitrary combination of failures, (i) every non-faulty SDN controller can reach any switch (or another controller) in the network within a bounded communication delay (in the presence of a bounded number of concurrent failures) and (ii) every switch is managed by at least one controller (as long as at least one controller is not faulty).
We evaluate $\system$ through a rigorous worst-case analysis as well as a prototype implementation (based on OVS and Floodlight), and we report on our experiments using Mininet. }

By introducing programmability, automated verification, and innovative debugging tools, Software-Defined Networks (SDNs) are poised to meet the increasingly stringent dependability requirements of today's communication networks. However, the design of fault-tolerant SDNs remains an open challenge. This paper considers the design of dependable SDNs through the lenses of self-stabilization---a very strong notion of fault-tolerance. In particular, we develop algorithms for an in-band and distributed control plane for SDNs, called Renaissance, which tolerate a wide range of failures. Our self-stabilizing algorithms ensure that after the occurrence of arbitrary failures, (i) every non-faulty SDN controller can reach any switch (or another controller) within a bounded communication delay (in the presence of a bounded number of failures) and (ii) every switch is managed by a controller. We evaluate Renaissance through a rigorous worst-case analysis as well as a prototype implementation (based on OVS and Floodlight, and Mininet).


\end{abstract}


\section{Introduction}
\label{sec:intro}
\noindent \textbf{Context and Motivation.} 
Software-Defined Network (SDN) technologies have emerged as a promising alternative to the vendor-specific, complex, and hence error-prone, operation of traditional communication networks. In particular, by outsourcing and consolidating the control over the data plane elements to a logically centralized software, SDNs support a programmatic verification and enable new debugging tools. Furthermore, the decoupling of the control plane from the data plane, allows the former to evolve independently of the constraints of the latter, enabling faster innovations. 

However, while the literature articulates well the benefits of the separation between control and data plane and the need for distributing the control plane (e.g., for performance and fault-tolerance), the question of how connectivity between these two planes is maintained (i.e., the communication channels from controllers to switches and between controllers) has not received much attention. Providing such connectivity is critical for ensuring the availability and robustness of SDNs.

Guaranteeing that each switch is managed, at any time, by at least one controller is challenging especially if control is \emph{in-band}, i.e., if control and data traffic is forwarded along the same links and devices and hence arrives at the same ports. In-band control is desirable as it avoids the need to build, operate, and ensure the reliability of a separate out-of-band management network. Moreover, in-band management can in principle improve the resiliency of a network, by leveraging a higher path diversity (beyond connectivity to the management port). 

The goal of this paper is the design of a highly fault-tolerant distributed and in-band control plane for SDNs. In particular, we aim to develop a self-stabilizing software-defined network: An SDN that recovers from controller, switch, and link failures, as well as a wide range of communication failures (such as packet omissions, duplications, or reorderings). As such, our work is inspired by Radia Perlman's pioneering work~\cite{Perlman1999}: Perlman's work envisioned a self-stabilizing Internet and enabled today's link state routing protocols to be robust, scalable, and easy to manage. Perlman also showed how to modify the ARPANET routing broadcast scheme, so that it becomes self-stabilizing~\cite{DBLP:journals/cn/Perlman83}, and provided a self-stabilizing spanning tree algorithm for interconnecting bridges~\cite{DBLP:conf/sigcomm/Perlman85}. Yet, while the Internet core is ``conceptually self-stabilizing'', Perlman's vision remains an open challenge, especially when it comes to recent developments in computer networks, such as SDNs, for which we propose self-stabilizing algorithms.

\noindent \textbf{Fault Model.} We consider (i) fail-stop failures of controllers, which failure detectors can observe, (ii) link failures, and (iii) communication failures, such as packet omission, duplication, and reordering. In particular, our fault model includes up to $\kappa$ link failures, for some parameter~$\kappa \in \mathbb{Z}^+$. In addition, to the failures captured in our model, we also aim to recover from \emph{transient faults}, i.e., any temporary violation of assumptions according to which the system and network were designed to behave, e.g., the corruption of the packet forwarding rules changes to the availability of links, switches, and controllers. We assume that (an arbitrary combination of) these transient faults can corrupt the system state in unpredictable manners. In particular, when modeling the system, we assume that these violations bring the system to an arbitrary state (while keeping the program code intact). Starting from an arbitrary state, the correctness proof of self-stabilizing systems~\cite{D2K,DBLP:journals/cacm/Dijkstra74} has to demonstrate the return to correct behavior within a bounded period, which brings the system to a \textit{legitimate state}. 

\noindent \textbf{The Problem.} This paper answers the following question: How can all non-faulty controllers maintain bounded (in-band) communication delays to any switch as well as to any other controller? We interpret the requirements for provable (in-band) bounded communication delays to imply (i) the absence of out-of-band communications or any kind of external support, and yet (ii) the possibility of fail-stop failures of controllers and link failures, as well as (iii) the need for guaranteed bounded recovery time after the occurrence of arbitrary transient faults. These faults are transient violations of the assumptions according to which the system was designed to behave.

\ems{Current implementations assume that outdated rules can expire via timeouts. Using such timeouts, one must guarantee that the network becomes connected eventually (even when starting from arbitrary timeout values and corrupted packet forwarding rules). This non-trivial challenge motivates our use of the asynchronous model when solving the studied problem via a mechanism for in-band network bootstrapping that connects every controller to every other node in the network.} 

\ems{Since we aim at recovering after the last occurrence of an arbitrary transient fault, the construction of a self-stabilizing bootstrapping mechanism makes the task even more challenging. Our solution combines a novel algorithm for in-band bootstrapping with well-known approaches for rapid recovery from link-failures, such as conditional forwarding rules~\cite{DBLP:conf/sigcomm/BorokhovichSS14}. Our analysis uses new proof techniques for showing that the system as a whole can recover rapidly from link and node failures as well as after the occurrence of the last arbitrary transient fault.} 

\noindent \textbf{Our Contributions.} We present an important module for dependable networked systems: a self-stabilizing software-defined network. In particular, we provide a (distributed) self-stabilizing algorithm for distributed SDN control planes that, relying solely on in-band communications, recover (from a wide spectrum of controller, link, and communication failures as well as transient faults) by re-establishing connectivity in a robust manner. Concretely, we present a system, henceforth called $\system$\footnote{The word renaissance means `rebirth' (French) and it symbolizes the ability of the proposed system to recover after the occurrence of transient faults that corrupt its state.}, which, to the best of our knowledge, is the first to provide:

\begin{enumerate}
	\item \emph{A robust efficient and distributed control plane:}
	We maintain short, $\bigO(D)$-length
	control plane paths in the presence of controller 
	and link (at most 
	$\kappa$ many) failures, as well as, communication failures, 
	where $D\leq N$ is the (largest) network diameter 
	(when considering any possible network topology changes over time) and $N$ is the number of nodes in the network.
	More specifically, suppose that throughout the recovery period the 
	network topology was $(\kappa+1)$-edge-connected and included at least 
	one (non-failed) controller. 
	We prove that starting from a legitimate state, i.e., after recovery, 
	our self-stabilizing solution can:
	
	\begin{itemize}
		\item \textit{Deal with fail-stop failures of controllers:}
		These failures require the removal of stale information (that is related to unreachable 
		controllers) from the switch configurations. 
		Cleaning up stale information avoids inconsistencies 
		and having to store large amounts
		of history data. 
		
		\item \textit{Deal with link failures:}
		Starting from a legitimate system state, the controllers 
		maintain an $\bigO(D)$-length path to all nodes 
		(including switches and other controllers), as long as at 
		most $\kappa$ links fail. That is, after the recovery period the communication delays are bounded.
		
	\end{itemize}
	
	\item \textit{Recovery from transient faults:} 
	We show that
	our control plane can even recover 
	after the occurrence of transient faults. 
	That is, starting from an \emph{arbitrary} state, 
	the system recovers within time $\bigO(D^2 N)$ to
	a legitimate state.
	In a legitimate state, 
	the number of packet forwarding rules per switch is at 
	most $|P_C|$ times the 
	optimal, where $|P_C|$ is the number of controllers. 
	The proposed algorithm is \emph{memory adaptive}~\cite{DBLP:conf/wdag/AnagnostouEH92}, i.e., 
	after the recovery from transient faults, each node's use of local memory 
	depends on the actual number, $n_C$, of controllers in the system, rather than the upper bound, $N_C$, on the number of controllers in the system. 

	\item The proposed algorithm is memory adaptive. That is, after its recovery from transient faults, each node's use of local memory depends on the actual number of controllers in the system, $n_C$, rather than the upper bound on the number of controllers in the system, $N_C$. We present a non-memory adaptive variation on the proposed algorithm that recovers within a period of $\Theta(D)$ after the occurrence of transient faults. This is indeed faster than the $\bigO(D^2 N)$ recovery time of the proposed algorithm. However, the cost of memory use \emph{after stabilization} can be $N_C/n_C$ times higher than the proposed algorithm. Moreover, the fact that the recovery time of the proposed memory adaptive solution is longer is relevant only in the presence of rare faults that can corrupt the system state arbitrarily, because for the case of benign failures, we demonstrate recovery within $\Theta(D)$.
		
\end{enumerate}

While we are not the first to consider the design of self-stabilizing systems which maintain redundant paths also beyond transient faults, the challenge and novelty of our approach comes from the specific restrictions imposed by SDNs (and in particular the switches). In this setting not all nodes can compute and communicate, and in particular, SDN switches can merely forward packets according to the rules that are decided by other nodes, the controllers. This not only changes the model, but also requires different proof techniques, e.g., regarding the number of resets and illegitimate rule deletions.

In order to validate and evaluate our model and algorithms, we implemented a prototype of $\system$ in Floodlight using Open vSwitch (OVS), complementing our worst-case analysis. Our experiments in Mininet demonstrate the feasibility of our approach, indicating that in-band control can be bootstrapped and maintained efficiently and automatically, also in the presence of failures. To ensure reproducibility and to facilitate research on improved and alternative algorithms, we have released the source code and evaluation data to the community at~\cite{RenaissanceWeb}.

We also discuss relevant extensions to the proposed solution (Section~\ref{sec:ext}), such as a  combing both in-band and out-of-band communications, as well as coordinating the actions of the different controllers using a reconfigurable replicated state machine.
 
\noindent \textbf{Organization.} We give an overview of our system and the components it interfaces in Section~\ref{sec:sys} and introduce our formal model in Section~\ref{sec:models}. Our algorithm is presented in Section~\ref{sec:selfStab}, analyzed in Section~\ref{sec:proof}, and validated in Section~\ref{sec:eval}. We then discuss related work (Section~\ref{sec:relwork}) before drawing the conclusions from our study (Section~\ref{sec:conclusion}).

\section{The System in a Nutshell}
\label{sec:sys}
Our self-stabilizing SDN control plane can be seen as one critical piece of a larger architecture for providing fault-tolerant communications. Indeed, a self-stabilizing SDN control plane can be used together with existing self-stabilizing protocols on other layers of the OSI stack, e.g., self-stabilizing link layer and self-stabilizing transmission control protocols~\cite{DBLP:conf/sss/DolevHSS12,DBLP:journals/ipl/DolevDPT11}, which provide logical FIFO communication channels. To put things into perspective, we provide a short overview of the overall network architecture we envision. Our proposal includes new self-stabilizing components that leverage existing self-stabilizing protocols towards an overall network architecture that is more robust than existing SDNs. We consider an architecture (Figure~\ref{fig:layers}) that comprises mechanisms for local topology discovery and a logic for packet forwarding rule generation. We contribute to this architecture a self-stabilizing abstract switch as well as a self-stabilizing SDN control platform.

\begin{figure*}[t!]
	\centering
	\includegraphics[clip=true,scale=\figureSizeSwitch{1.0}{1.4}]{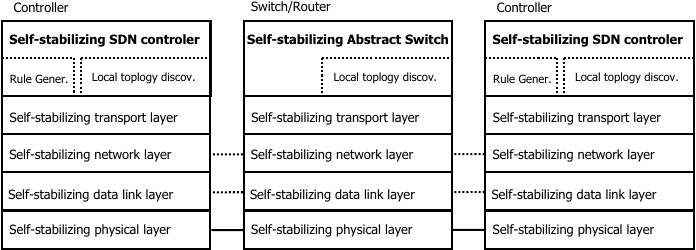}
	\caption{\footnotesize{The system architecture, which is based on self-stabilizing versions of existing network layers. The external building blocks for rule generation and local topology discovery appear in the dotted boxes. The proposed contribution of self-stabilizing SDN controller and self-stabilizing abstract switch appear in bold.}}
	\label{fig:layers}
\end{figure*}

The network includes a set $P_C = \{ p_1, \ldots, p_{n_C} \}$ of $n_C$ {\em (remote) controllers}, and a set $P_S = \{ p_{n_C+1}, \ldots, p_{n_C+n_S} \}$ of the $n_S$ {\em (packet forwarding) switches}, where $i$ is the unique identifier of node $p_i \in P= P_C \cup P_S$. We denote by $N_c(i) \subseteq P$ (communication neighborhood) the set of nodes which are directly connecting node $p_i \in P$ and node $p_j$, i.e., $p_j \in N_c(i)$. At any given time, and for any given node $p_i \in P$, the set $N_o(i)$ (operational neighborhood) refers to $p_i$'s directly connected nodes for which ports are currently available for packet forwarding. The local topology information in $N_o(i)$ is liable to change rapidly and without notice. We denote the operational and connected communication topology as $G_o=(P, E_o)$, and respectively, as $G_c=(P, E_c)$, where $E_{x} = \{(p_i, p_j) \in P \times P : p_j \in N_x(i) \}$ for $x\in \{o,c\}$.

Each switch $p_i \in P_S$ stores a set of rules that the controllers install in order to define which packets have to be forwarded to which ports. In the out-of-band control scenario, a controller communicates the forwarding rules via a dedicated management port to the \emph{control module} of the switch. In contrast, in an in-band setting, the control traffic is interleaved with the data plane traffic, which is the traffic between \emph{hosts} (as opposed to controller-to-controller and controller-to-switch traffic): switches can be connected to hosts through data ports and may have additional rules installed in order to correctly forward their traffic. We do not assume anything about the hosts' network service, except for that their traffic may traverse any network link.

In an in-band setting, control and data plane traffic arrive through the same ports at the switch, which implies a need for being able to \emph{demultiplex} control and data plane traffic: switches need to know whether to forward (data) traffic out of another port or (control) traffic to the control module. In other words, control plane packets need to be logically distinguished from data plane traffic by some tag (or another deterministic discriminator). 

Figure~\ref{fig:self-stab-switch} illustrates the switch model considered in this paper. Our self-stabilizing control plane considers a proposal for \textit{abstract switches} that do not require the extensive functionality that existing SDN switches provide. An abstract switch can be managed either via the management port or in-band. It stores forwarding (match-action) rules. These rules are used to forward data plane packets to ports leading to neighboring switches, or to forward control packets to the local control module (e.g., instructing the control module to change existing rules). Rules can also drop all the matched packets. The match part of a rule can either be an exact match or optionally include wildcards.

Maintaining the forwarding rules with in-band control is the key challenge addressed in this paper: for example, these rules must ensure (in a self-stabilizing manner) that control and data packets are demultiplexed correctly (e.g., using tagging). Moreover, it must be ensured that we do not end up with a set of misconfigured forwarding rules that drop \emph{all} arriving (data plane and control plane) packets: in this case, a controller will never be able to manage the switch anymore in the future. 

In the following, we will assume a local topology discovery mechanism that each node uses to report to the controllers the availability of their direct neighbors. Also, we assume access to self-stabilizing protocols for the link layer (and the transport layer)~\cite{DBLP:conf/sss/DolevHSS12,DBLP:journals/ipl/DolevDPT11} that provide reliable, bidirectional FIFO-communication channels over unreliable media that is prone to packet omission, reordering, and duplication.

\Subsection{Switches and rules}
\label{sec:SwitchesAndRules}
Each switch $p_i \in P_S$ stores a set of forwarding rules which are installed by the controllers (servers) and define which packets have to be forwarded to which ports. In an out-of-band network, a controller communicates the forwarding rules via a dedicated management port to the \emph{control module} of the switch. In contrast, in an in-band setting, the control traffic is interleaved with the dataplane traffic, and is communicated (possibly along multiple hops, in case of a remote controller) to a regular switch port. This implies that in-band control requires the switch to demultiplex control and data plane traffic. In other words, the dataplane of a switch cannot only be used to connect the switch ports internally, but also to connect to the control module.

\begin{figure}
	\centering
	\includegraphics[clip=true,scale=\figureSizeSwitch{0.29}{0.4}]{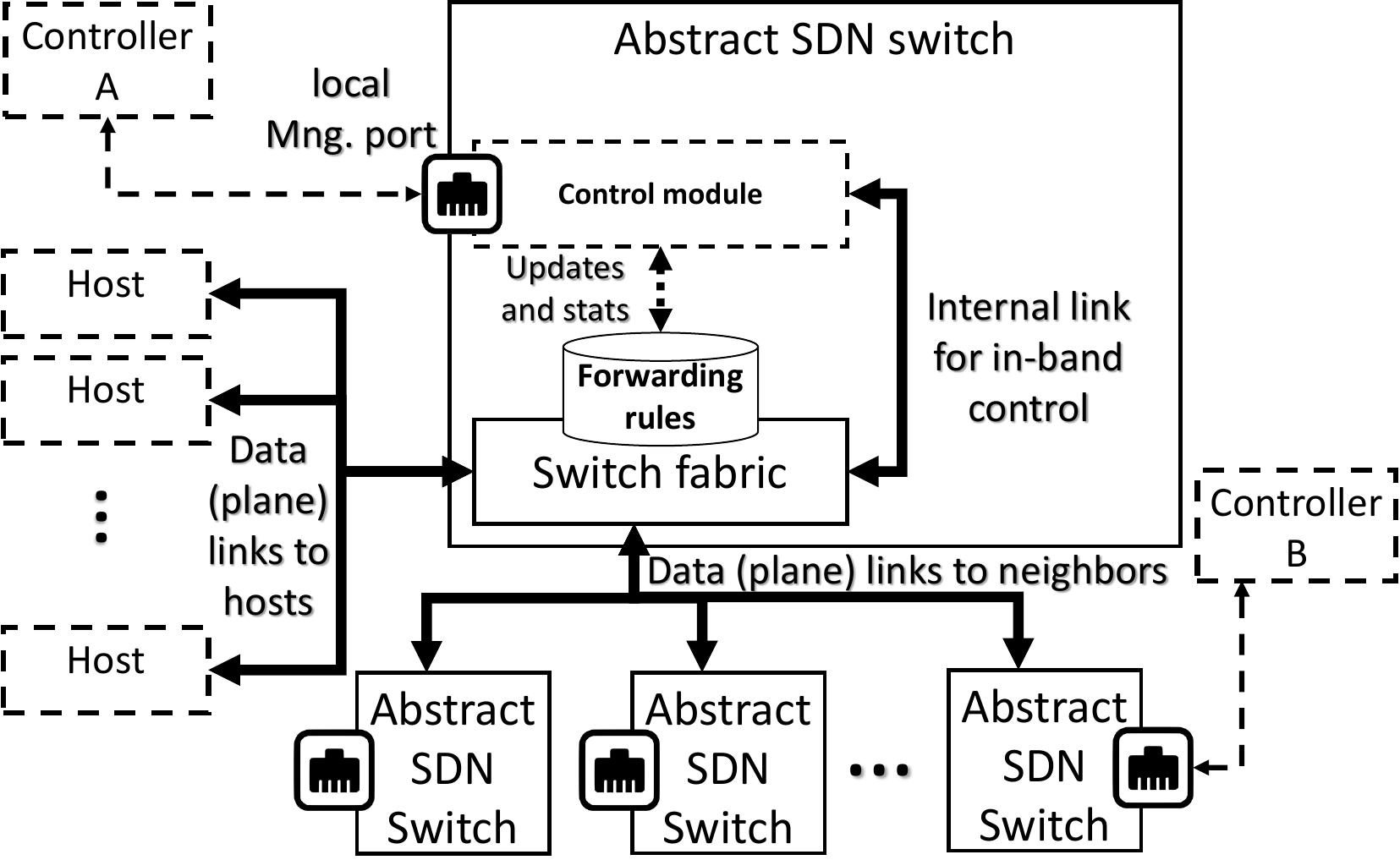}
	\caption{\label{fig:self-stab-switch}\footnotesize{Abstract SDN switch illustration.}}
\end{figure}

In this paper, we make the natural assumption that switches have a bounded amount of memory. Moreover, we assume that rules come in the form of match-action pairs, where the match can optionally include wildcards and the action part mainly defines a forwarding operation (cf. Figure~\ref{fig:self-stab-switch}).

More formally, suppose that $p_i \in P_S$ is a switch that receives a packet with $p_{src} \in P_C$ and $p_{dest} \in P$, as the packet source and destination, respectively. We refer to a \textit{rule} (for packet forwarding at the switch) by a tuple $\langle k, i, {src}, {dest}, prt$, $j, metadata \rangle$. The fields of a rule refer to $p_k$ as the controller that created this rule, $prt \in \{0, \ldots, n_{prt}\}:n_{prt}\geq \kappa+1$ is a priority that $p_k$ assigns to this rule, $p_j \in N_c(i)$ is a port on which the packet can be sent whenever $p_j \in N_o(i)$, and $metadata$ is an (optional) opaque data value. Our self-stabilizing abstract switch considers only rules that are installed on the switches indefinitely, i.e., until a controller \emph{explicitly} requests to delete them, rather than setting up rules with expiration \emph{timeouts}. 

We say that the rule $r$ $=$ $\langle k$, $i$, ${src}$, ${dest}$, $prt$, $j$, $metadata \rangle$ is \textit{applicable} for a packet that reaches switch $p_i$ and has source $p_{src}$ and destination $p_{dest}$, when $r$ is the rule with the highest $prt$ (priority) that matches the packet's source and destination fields, and $p_j \in N_o(i)$, i.e., the link $(p_i, p_j)$ is operational. We say that the set of rules of switch $p_i$, $rules(i)$, is \textit{unambiguous}, if for every received packet there is at most one applicable rule. Thus, a packet can be forwarded if there exists only one applicable rule in the switch's memory. We assume an interface function $myRules()$ which outputs the unambiguous rules that a controller $p_k\in P_C$ needs to install to a switch $p_j\in P_S$, based on $p_k$'s knowledge of the network's topology. We require rules to be unambiguous and offer resilience against at most $\kappa$ link failures (details appear in Section~\ref{sec:kappaFlowsAboveZeroConstructing}).

\subsubsection{The abstract switch}
\label{sec:abstractSwitch}
The main task of switches is to forward traffic according to the rules installed by the controllers. In addition, switches provide basic functionalities for interacting with the controllers.

While OpenFlow, the de facto standard specification for the switch interface, as well as other suggestions (Forwarding Metamorphosis~\cite{DBLP:conf/sigcomm/BosshartGKVMIMH13}, P4~\cite{DBLP:journals/ccr/BosshartDGIMRSTVVW14}, and SNAP~\cite{DBLP:conf/sigcomm/ArashlooK0RW16}) provide innovative abstractions with respect to data plane functionality and means to implement efficient network services, there is less work regarding the control plane abstraction, especially with respect to fault tolerance. 

We consider a slightly simpler switch model that does not include all the functionality one may find in an existing SDN switch. In particular, the proposed abstract SDN switch only supports the \emph{equal roles} approach (where multiple ``equal'' controllers manage the switch); the \emph{master-slave setup} usually used by switches~\cite{openflow-spec} is not relevant toward the design of our self-stabilizing distributed SDN control plane. We elaborate more on the interface in the following.

\subsubsection*{Configuration queries (via a direct neighbor)}
\label{sec:scq}
As long as the system rules and operational links support (bidirectional) packet forwarding between controller $p_i$ and switch $p_j$, the abstract switch allows $p_i$ to access $p_j$'s configuration remotely, i.e., via the interface functions $manager(j)$ (query and update), $rules(j)$ (query and update) as well as $N_c(j)$ (query-only), where $manager(j) \subseteq P_C$ is $p_j$'s set of assigned managers and $rules(j)$ is $p_j$'s rule set. Also, a switch $p_j$, upon arrival of a query of a controller $p_i$, responds to $p_i$ with the tuple $\langle j, N_c(j), manager(j), rules(j) \rangle$.

The abstract switch also allows controller $p_i$ to query node $p_j$ via $p_j$'s direct neighbor, $p_k$ as long as $p_i$ knows $p_k$'s local topology. In case $p_j$ is a switch, $p_i$ can also modify $p_j$'s configuration (via $p_j$'s abstract switch) to include a flow to $p_i$ (via $p_k$) and then to add itself as a manager of $p_j$. (The term \emph{flow} refer here to rules installs on a path in the network in a way that allows packet exchange between the path ends.) We refer to this as the \textit{query (and modify)-by-neighbor} functionality. 

\subsubsection*{The switch memory management}
\label{sec:memory}
\ems{We assume that} the number of rules and controllers (that manage switches) that each switch can store is bounded by $maxRules$ and $maxManagers$, respectively. We require that the abstract switch has a way to deal with clogged memory\ems{, \ie, when the flow table is full, cf.~\cite{openflow-spec}, Section B.17.7. Specifically, the abstract switch needs to implement an eviction policy that gives the lowest priority to rules that were least recently updated.} Similarly, we assume that whenever the number of managers that a switch stores exceeds $maxManagers$, the last to be stored (or \ems{accessed)} manager is removed so that a new manager can be added. \ems{We note that these requirements can be implemented using well-known techniques, for details see~\cite{DBLP:journals/corr/abs-1712-07697}, Section 2.1.1.} 


\Subsection{Building blocks}
\label{sec:buildingBlocks}
Our architecture relies on a fault-tolerant mechanism for topology discovery. We use such a mechanism as an external building block. Moreover, we require a notion of resilient flows. We next discuss both these aspects.

\Subsubsection{Topology discovery}
\label{sec:topologyDiscovery}
We assume a mechanism for local neighborhood discovery. We consider a system that uses an (ever running) failure detection mechanism, such as the self-stabilizing $\Theta$ failure detector~\cite[Section 6]{DBLP:conf/netys/BlanchardDBD14}: it discovers the switch neighborhood by identifying the failed/non-failed status of its attached links and neighbors. We assume that this mechanism reports the set of nodes which are directly connecting node $p_i \in P$ and node $p_j$, i.e., $p_j \in N_c(i)$.

\Subsubsection{Fault-resilient flows}
\label{sec:kappaFlowsAboveZeroConstructing}
We consider fault-resilient flows that are reminiscent of the flows in~\cite{DBLP:conf/hotnets/LiuYSS11}. The definition of $\kappa$-fault-resilient flows considers the network topology $G_c$ and assumes that $G_c$ is not subject to changes. The idea is that the network can forward the data packets along the shortest routes, and use alternative routes in the presence of link failures, based on conditional forwarding rules~\cite{DBLP:conf/sigcomm/BorokhovichSS14}; these failover rules provide a backup for every edge and an enhancement of this redundancy for the case in which at most $\kappa$ links fail, as we describe next. 

Let $(p_{r_1}, \ldots, p_{r_n}) \in P^n$ be a directed path in the communication network $G_c$, where $n \in \{2, \ldots, |P|\}$. Given an operational network $G_o$, we say that $(p_{r_1}, \ldots, p_{r_n})$ is a \textit{flow} (over a simple path) in $G_o$, when the rules stored in $p_{r_1}, \ldots, p_{r_n}$ relay packets from source $p_{r_1}$ to destination $p_{r_n}$ using the switches in the sequence $p_{r_2}, \ldots, p_{r_{n-1}}$ for packet forwarding (relay nodes). Let $G_o(k)$ be an operational network that is obtained from $G_c$ by an arbitrary removal of $k$ links. We say there is a $\kappa$-fault-resilient flow from $p_i$ to $p_j$ in $G_c$ when for any $k \leq \kappa$ there is a flow (over a simple path) from $p_i$ to $p_j$ in any $G_o(k)$. We note that when considering a communication graph, $G_c$, with a general topology, the construction of $\kappa$-fault-resilient flows is possible when $\kappa < \lambda(G_c)$, where $\lambda(G_c)$ is the edge-connectivity of $G_c$ (i.e., the minimum number of edges whose removal can disconnect $G_c$).

\section{Models}
\label{sec:models}
This section presents a formal model of the studied system (Figure~\ref{fig:layers}), which serves as the framework for our correctness analysis of the proposed self-stabilizing algorithms (Section~\ref{sec:proof}).

We model the control plane as a message passing system that has no notion of clocks (nor explicit timeout mechanisms), however, it has access to link failure detectors (in a way that is similar to the Paxos model~\cite{DBLP:conf/netys/BlanchardDBD14,DBLP:journals/tocs/Lamport98}). We borrow from~\cite[Section 6]{DBLP:conf/netys/BlanchardDBD14} a technique for local link monitoring (Section~\ref{sec:topologyDiscovery}), which assumes that every abstract switch can complete \emph{at least} one round-trip communication with any of its direct neighbors while it completes \emph{at most} $\Theta$ round-trips with any other directly connected neighbor. In other words, in our analytical model, but not in our emulation-base evaluation, we assume that nodes have a mechanism to locally detect temporary link failures (e.g., a link may also be unavailable due to congestion); a link which is unavailable for a longer time period will be flagged as permanent failure by a failure detector, which we borrow from~\cite[Section 6]{DBLP:conf/netys/BlanchardDBD14}. Apart from this monitoring of link status, we consider the control plane as an asynchronous system. Note that once the system installs a $\kappa$-fault-resilient flow between controller $p_i \in P_c$ and node $p_j \in P \setminus \{ p_i\}$, the network provides a communication channel between $p_i$ and $p_j$ that has a bounded delay (because we assume that there are never more than $\kappa$ link failures). Moreover, these bounded delays are offered by the data plane while the control plane is still asynchronous as described above (since, for example, we assume no bound on the time it takes a controller to perform a local computation). 

Self-stabilizing algorithms usually consist of a \emph{do forever} loop that contains communication operations and validations that the system is in a consistent state as part of the transition decision. An iteration (of the \emph{do forever} loop) is said to be \textit{complete} if it starts in the loop's first line and ends at the last (regardless of whether it enters branches). As long as every non-failed node eventually completes its do forever loop, the proposed algorithm is oblivious to the rate in which this completion occurs. Moreover, the exact time considerations can be added later for the sake of fine-tuning performances.

\Subsection{The communication channel model}
\label{sec:transportModel}
We are given reliable end-to-end FIFO channels over capacitated links, as implemented, e.g., by~\cite{DBLP:conf/sss/DolevHSS12,DBLP:journals/ipl/DolevDPT11}, which guarantee reliable message transfer regardless of packet omission, duplication, and reordering. After the recovery period of the channel algorithm~\cite{DBLP:conf/sss/DolevHSS12,DBLP:journals/ipl/DolevDPT11}, it holds that, at any time, there is exactly one token $pkt \in \{ act, ack \}$ in the channel that is either in transit from the sender $p_i \in P$ to the receiver $p_j \in P$, i.e., $channel_{i,j}=\{act\} \land channel_{j,i}=\emptyset$, or the token $pkt$ is in transit from $p_j$ to $p_i$, i.e., $channel_{i,j}=\emptyset \land channel_{j,i}=\{ack\}$. During the recovery period (after the last occurrence of a transient fault), it can be the case that the sender sends a message $m_0$ for which it receives a (false) acknowledgment $ack_0$ without having $m_0$ go through a complete round-trip. However, that can occur at most $\Delta_{comm}$ times, where $\Delta_{comm}\leq 3$ for the case of~\cite{DBLP:conf/sss/DolevHSS12,DBLP:journals/ipl/DolevDPT11}. That is, once the sender sends message $m_1$ and receives its acknowledgment $ack_1$, the channel algorithm~\cite{DBLP:conf/sss/DolevHSS12,DBLP:journals/ipl/DolevDPT11} guarantees that $m_1$ has completed a round-trip. 

When node $p_i$ sends a packet, $pkt \in \{ act, ack \}$, to node $p_j$, the operation $send$ inserts a copy of $pkt$ to the FIFO queue that represents the above communication channel from $p_i$ to $p_j$, while respecting the above token circulation constraint. When $p_j$ receives $pkt$ from $p_i$, node $p_j$ delivers $pkt$ from the channel's queue and transfers $pkt$'s acknowledgment to the channel from $p_j$ to $p_i$ immediately after.

\Subsection{The execution model}
\label{sec:interModel}
For our analysis, we consider the standard \emph{interleaving model}~\cite{D2K}, in which there is a single (atomic) step at any given time. An input event can be either a packet reception or a periodic timer triggering $p_i$ to resend while executing the do forever loop. In our settings, the timer rate is completely unknown and the only assumption that we make is that every non-failing node executes its do forever loop infinitely often. 

We model a node (switch or controller) using a state machine that executes its program by taking a sequence of {\em (atomic) steps}, where a step of a controller starts with local computations and ends with a single communication operation: either $send$ or $receive$ of a packet. A step of the (control module of an) abstract switch starts with a single message reception, continues with internal processing and ends with a single message send.

The \emph{state} of node $p_i$, denoted by $s_i$, consists of the values of all the variables of the node including its communication channels. The term {\em (system) state} is used for a tuple of the form $(s_1, s_2, \cdots, s_n, G_o)$, where each $s_i$ is the state of node $p_i$ (including messages in transit to $p_i$) and $G_o$ is the operational network that is determined by the environment. We define an {\em execution (or run)} $R={c_0,a_0,c_1,a_1,\ldots}$ as an alternating sequence of system states $c_x$ and steps $a_x$, such that each state $c_{x+1}$, except the initial system state $c_0$, is obtained from the preceding state $c_x$ by applying step $a_x$.

For the sake of simple presentation of the correctness proof, we assume that the abstract switch deals with one controller at a time, e.g., when requesting a configuration update or a query. Moreover, we assume that within a single atomic step, the abstract switch can receive the controller request, perform the update, and send a reply to the controller.

\Subsection{The network model}
We consider a system in which $maxRules$ is large enough to store all the rules that all controllers need to install to any given switch, and that $maxManagers \geq N_C$. We assume that $|P_C| = n_C$ and $|P_S| = n_S$ are known only by their upper bounds, i.e., $N_C \geq |P_C|$, and respectively, $N_S \geq |P_S|$. We use these bounds only for estimating the memory requirements per node, in terms of $maxRules$ and $maxManagers$, i.e., the maximum number of rules, and respectively, managers at any switch. 

Suppose that a $\kappa$-fault-resilient flow from $p_i$ to $p_j$ is installed in the network. The term \textit{primary path} refers to the path along which the network forwards packets from $p_i$ to $p_j$ \textit{in the absence of failures}. We assume that $myRules()$ returns rules that encode $\kappa$-fault-resilient flows for a given network topology. The primary paths encoded by $myRules()$ are also the shortest paths in $G_c$ (with the highest rule priority). A rule in $myRules()$ corresponding to $k$ link failures ($k$-fault-resilient flow) has the $(k+1)$-highest rule priority.

\Subsubsection{Communication fairness}
\label{sec:communicationFairness}
Due to the presence of faults in the system, we do not consider any bound on the communication delay, which could be, for example, the result of the absence of properly installed flows between the sender and the receiver. Nevertheless, when a flow is properly installed, the channel is not disconnected and thus we assume that sending a packet infinitely often implies its reception infinitely often. We refer to the latter assumption as the \textit{communication fairness} property. We make the same assumptions both for the link and transport layers.

\Subsubsection{Message round-trips and iterations of self-stabilizing algorithms}
\label{sec:messageRoundtrips}
This work proposes a solution for bootstrapping in-band communication in SDNs. The correctness proof depends on the nodes' ability to exchange messages during this bootstrapping. The proof uses the notion of a message round-trip, which includes sending a message to a node and receiving a reply from that node. Note that this process spans over many system states.

We give a detailed definition of round-trips as follows. Let $p_i \in P_C$ be a controller and $p_j \in P \setminus \{p_i\}$ be a network node. Suppose that immediately after state $c$ node $p_i$ sends a message $m$ to $p_j$, for which $p_i$ awaits a response. At state $c'$, that follows state $c$, node $p_j$ receives message $m$ and sends a response message $r_m$ to $p_i$. Then, at state $c''$, that follows state $c'$, node $p_i$ receives $p_j$'s response, $r_m$. In this case, we say that $p_i$ has completed with $p_j$ a round-trip of message $m$. 

We define an iteration of a self-stabilizing algorithm in our model. Let $P_i$ be the set of nodes with whom $p_i$ completes a message round trip infinitely often in execution $R$. Suppose that immediately after the state $c_{begin}$, controller $p_i$ takes a step that includes the execution of the first line of the do forever loop, and immediately after system state $c_{end}$, it holds that: (i) $p_i$ has completed the iteration it has started immediately after $c_{begin}$ (regardless of whether it enters branches) and (ii) every message $m$ that $p_i$ has sent to any node $p_j \in P_i$ during the iteration (that has started immediately after $c_{begin}$) has completed its round trip. In this case, we say that $p_i$'s iteration (with round-trips) starts at $c_{begin}$ and ends at $c_{end}$. 

\Subsection{The fault model}
\label{sec:faultModel}
We characterize faults by their duration, that is, they are either transient or permanent. We consider the occurrence frequency of transient faults to be either rare or not rare. We illustrate our fault model in Figure~\ref{fig:self-stab-SDN}. 

\Subsubsection{Failures that are not rare}
Transient packet failures, such as omissions, duplications, and reordering, may occur often. Recall that we assume communication fairness and the use of a self-stabilizing link layer (and transport layer)~\cite{DBLP:conf/sss/DolevHSS12,DBLP:journals/ipl/DolevDPT11}. This protocol assures that the system's unreliable media, which are prone to packet omission, reordering, and duplication, can be used for providing reliable, bidirectional FIFO-communication channels without omissions, duplications or reordering. Note that the assumption that the communication is fair may still imply that there are periods in which a link is temporarily unavailable. We assume that at any time there are no more than such $\kappa$ link failures.

\Subsubsection{Failures that may occur rarely}
We model rare faults to occur only before the system starts running. That is, during the system run, $G_c$ does not change and it is $(\kappa+1)$-edge connected.

A permanent link failure or addition results in the removal, and respectively, the inclusion of that link from the network. The fail-stop failure of node $p_j$ is a transient fault that results in the removal of $(p_i, p_j)$ from the network and $p_j$ from $N_c(i)$, for every $p_i\in N_c(j)$. Naturally, node addition is combined with a number of new link additions that include the new node.

Other than the above faults, we also consider any violation of the assumptions according to which the system is assumed to operate (as long as the code stays intact). We refer to them as \emph{(rare) transient faults}. They can model, for example, the event in which more than $\kappa$ links fail concurrently. A transient fault can also corrupt the state of the nodes or the messages in the communication channels.

\Subsubsection{Benign vs. transient faults}
We define the set of \emph{benign faults} to include any fault that is not both rare and transient. The correctness proof of the proposed algorithm demonstrates the system's ability to recover after the occurrence of either benign or transient faults, which are not necessarily rare.  Our experiments, however, consider all benign faults and no rare transient faults due to the computation limitations that exist when considering all possible ways to corrupt the system state (Section~\ref{sec:limitations}).

\begin{figure*}[t!]
	\centering
	\begin{smaller}
	\begin{tabular}{lll}
		\cline{2-3}
		\multicolumn{1}{l|}{}                    & \multicolumn{2}{l|}{~~~~~~~~~~~~~~~~~~~~~~~~~~~~~~~~~~~~~~~~~~~~\textbf{Frequency}}                                                                               \\ \hline
		\multicolumn{1}{|l|}{\textbf{Duration}}  & \multicolumn{1}{l|}{\textit{Rare}}                          & \multicolumn{1}{l|}{\textit{Not rare}}                     \\ \hline
		\multicolumn{1}{|l|}{}                   & \multicolumn{1}{l|}{Any violation of the assumptions according}         & \multicolumn{1}{l|}{Packet failures: omissions,}        \\
		\multicolumn{1}{|l|}{\textit{}}          & \multicolumn{1}{l|}{to which the system is assumed to}           & \multicolumn{1}{l|}{duplications, reordering} \\
		\multicolumn{1}{|l|}{\textit{}}          & \multicolumn{1}{l|}{operate (as long as the code stays intact).}    & \multicolumn{1}{l|}{(assuming communication}      \\
		\multicolumn{1}{|l|}{\textit{Transient}} & \multicolumn{1}{l|}{This can result in any state corruption.} & \multicolumn{1}{l|}{fairness holds).}                                   \\ \cline{2-3} 
		\multicolumn{1}{|l|}{}                   & \multicolumn{1}{l|}{}                                       & \multicolumn{1}{l|}{Link failures (assuming} \\
		\multicolumn{1}{|l|}{}                   & \multicolumn{1}{l|}{}                                       & \multicolumn{1}{l|}{at most $\kappa$ links failures).}   \\ \hline
		\multicolumn{1}{|l|}{\textit{Permanent}} & \multicolumn{1}{l|}{Node and link failures.}                 & \multicolumn{1}{l|}{}                                   \\ \hline
		\vspace*{0.25em}
	\end{tabular}
\end{smaller}
	\includegraphics[clip=true,scale=\figureSizeSwitch{0.75}{0.73}]{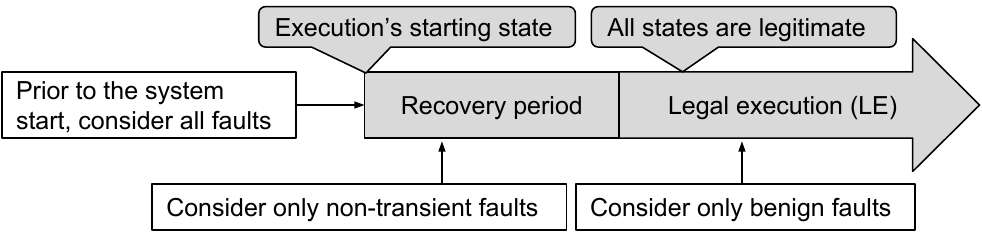}\\
	\caption{\label{fig:self-stab-SDN}\footnotesize{The table above details our fault model and the chart illustrates when each fault set is relevant. The chart's gray boxes represent the system execution, and the white boxes specify the failures considered to be possible at different execution parts and recovery guarantees of the proposed self-stabilizing algorithm. The set of benign faults includes both transient link failures as well as permanent link and node failures.}}
\end{figure*}

\Subsection{Self-Stabilization}
We define the system's task by a set of executions called {\em legal executions} ($LE$) in which the task's requirements hold. That is, each controller $p_i$ constructs a $\kappa$-fault-resilient flow to every node $p_j \in P$ (either a switch or a controller). We say that a system state $c$ is {\em legitimate}, when every execution $R$ that starts from $c$ is in $LE$. A system is {\em self-stabilizing}~\cite{D2K} with relation to task $LE$, when every (unbounded) system execution reaches a legitimate state with relation to $LE$ (cf. Figure~\ref{fig:self-stab-SDN}). The criteria of self-stabilization in the presence of faults~\cite[Section 6.4]{D2K} requires the system to recover within a bounded period after the occurrence of a single benign failure during legal executions (in addition to the design criteria of self-stabilization that requires recovery within a bounded time after the occurrence of the last transient fault). We demonstrate self-stabilization in Section~\ref{sec:recover} and self-stabilization in the presence of faults in Section~\ref{sec:post}.

Self-stabilizing systems require the use of bounded memory, because real-world systems have only access to bounded memory. Moreover, the number of messages sent during an execution does not have an immediate relevance in the context of self-stabilization. The reason is that self-stabilizing algorithms can never terminate and stop sending messages, because if they did it would not be possible for the system to recover from transient faults (cf. \cite[Chapter 2.3]{D2K}). That is, suppose that the algorithm includes a predicate, such that when the predicate is true the algorithm forever stops sending messages. Then, a single transient fault can cause this predicate to be true in the starting state of an execution, from which the system can never recover. The latter holds, because the algorithm will never send any message and yet in the starting system state any variable that is not considered by the predicate can be corrupted.

\Subsubsection{Execution fairness}
\label{sec:executionFairness}
We say that a system execution is \textit{fair} when every step that is applicable infinitely often is executed infinitely often and fair communication is kept (both at the link and the transport layer). Note that only failing nodes ever stop taking steps and thus a violation of the fairness (communication or execution) assumptions implies the presence of transient faults, which we assume to happen only before the starting system state of any execution.   
%
We clarify that fair execution and communication are weaker assumptions than partial synchrony~\cite{dwork1988consensus} because they imply unknown upper bounds on relative processor speeds and message delay. 

\Subsubsection{Asynchronous frames}
\label{sec:cycleFrame}
The first \textit{(asynchronous) frame} in a fair execution $R$ is the shortest prefix $R'$ of $R=R' \circ R''$, such that each controller starts and ends at least one complete iteration (with round-trips) during $R'$ (see Section~\ref{sec:messageRoundtrips}), where $\circ$ denotes an operation that concatenates two executions. The second frame in execution $R$ is the first frame in execution $R''$, and so on.

\Subsubsection{Complexity measures}
\label{sec:timeComplexity}
The \textit{stabilization time} (or recovery period from transient faults) of a self-stabilizing system is the number of asynchronous frames it takes a fair execution to reach a legitimate system state when starting from an arbitrary one. The recovery period from benign faults is also measured by the number of asynchronous frames it takes the system return to a legal execution after the occurrence of a single benign failure.

We also consider the design criterion of \emph{memory adaptiveness} by Anagnostou et al.~\cite{DBLP:conf/wdag/AnagnostouEH92}. This criterion requires that, after the recovery period, the use of memory by each node is a function of the actual network dimensions. In our system, a memory adaptive algorithm has space requirements that depend on $n_C$, which is the actual number controllers rather than their upper bound, $N_C$. Moreover, when considering \ems{non-adaptive solutions,} one can achieve a shorter recovery period from transient faults (Section~\ref{sec:conclusion}).

For the sake a simple presentation, our theoretical analysis assumes that all local computations are done within a negligible time that is independent of, for example, the number of messages sent and received during each frame. We do however consider all network dimensions that are related to the recovery costs (including the number of messages sent and received during each frame) during the evaluation of the proposed prototype (Section~\ref{sec:eval}). 

\section{Renaissance: A Self-Stabilizing SDN Control Plane}
\label{sec:selfStab} 
We present a self-stabilizing SDN control plane, called $\system$, that enables each controller to discover the network, remove any stale information in the configuration of the discovered unmanaged switches (e.g., rules of failed controllers), and construct a $\kappa$-fault-resilient flow to any other node (switch or controller) that it discovers in the network. For the sake of presentation clarity, we start with a high-level description of the proposed solution in Algorithm~\ref{alg:SHORTselfStabCode} before we present the solution details in Algorithm~\ref{alg:selfStabCode}.

\LinesNumbered

\begin{algorithm*}[t!]
	\begin{\VCalgSize}
	\BlankLine
	\noindent \textbf{Local state:} 
	$replyDB \subseteq \{ m(j) : p_j \in P \}$ has the most recently received query replies\label{ln:SHORTlocal}\; 
	$currTag$ and $prevTag$ are $p_i$'s current and previous synchronization round, respectively\label{ln:SHORTtags}\;
	
	\noindent \textbf{Interface:} $myRules(G,j,tag)$: returns the rules of $p_i$ on switch $p_j$ given a topology $G$ on round $tag$\label{ln:SHORTmyRules}\;
	
	\BlankLine
	
	\noindent \textbf{do forever} \Begin{\label{ln:SHORTdoForeverStart}
		
		Remove from $replyDB$ any reply from unreachable (in terms of graph connectivity) senders or not from round $prevTag$ or $currTag$. Also, remove from $replyDB$ any response from $p_i$ and then add a record that includes the directly connected neighbors, $N_c(i)$\label{ln:SHORTstale}\;
		
		\If{$replyDB$ includes a reply (with tag $currTag$) from every node that is reachable (in terms of graph connectivity) according to the accumulated local topology, $G$, in $replyDB$\label{ln:SHORTifAllNetDiscovered}}{Store $currTag$'s value in $prevTag$ and get a new and unique tag for $currTag$. By that, $p_i$ starts a new synchronization round\label{ln:SHORTifAllNetDiscoveredStart}\label{ln:SHORTifAllNetDiscoveredEnd};}
		
		\ForEach{switch $p_j \in P_S$ and $p_j$'s most recently received reply}{\label{ln:SHORTforEachSwitch}
			\If{this is the start of a new synchronization round\label{ln:SHORTifSynch}}{Remove from $p_j$'s configuration any manager $p_k$ or rule of $p_k$ that was not discovered to be reachable during round $prevTag$\label{ln:SHORTrm};}
			Add $p_i$ in $p_j$'s managers (if it is not already included) and replace $p_i$'s rules in $p_j$ with $myRules(G, j, currTag)$\;
		}
		
		\lForEach{$p_j \in P$ that is reachable from $p_i$ according to the most recently received replies in $replyDB$\label{ln:SHORTsendStart}}{\textbf{send} \textbf{to} $p_j$ (with tag $currTag$) an update message (if $p_j \in P_S$ is a switch) and query $p_j$'s configuration\label{ln:SHORTsendUpdToSwitch}}
	}
	
	\BlankLine
	
	\noindent \textbf{upon query reply $m$ from $p_j$\label{ln:SHORTqueryReturn}} \Begin{
		\lIf{there is no space in $replyDB$ for storing $m$}{perform a C-reset by including in $replyDB$ only the direct neighborhood, $N_c(i)$\label{ln:SHORTcheckResponsesSize}} 
		\textbf{if} $m$'s tag equals to $currTag$ \textbf{then} include $m$ in $replyDB$ after removing the previous response from $p_j$\label{ln:SHORTrcv}\;
	}
	
	\BlankLine
	
	\noindent \textbf{upon arrival of a query (with a $syncTag$) from $p_j$}\label{ln:SHORTarrivalCtrl} \Begin{\textbf{send  to} $p_j$ a response that includes the local topology, $N_c(i)$, and $syncTag$~\label{ln:SHORTqryCon}}
	
	\caption{\label{alg:SHORTselfStabCode}\footnotesize{Self-stabilizing SDN, high-level code description for controller $p_i$. Algorithm~\ref{alg:selfStabCode} is a detailed version of this algorithm.}}
	\end{\VCalgSize}
\end{algorithm*}

\Subsection{High-level description of the proposed algorithm}
Algorithm~\ref{alg:SHORTselfStabCode} creates 
an iterative process of topology discovery that, first, lets each controller identify 
the set of nodes that it is directly connected to; from there, it finds the 
nodes that are directly connected to them; and so on. 
This network discovery process is combined with another process 
for bootstrapping communication between any controller 
and any node in the network, i.e., connecting 
each controller to its direct neighbors, and then 
to their direct neighbors, and so on, until it is connected to the entire 
reachable network. 

Each controller associates independently each iteration with a unique tag~\cite{DBLP:journals/jcss/AlonADDPT15} that synchronizes a round in which the controller performs configuration updates and queries. 
Controller $p_i$ also maintains the variables $currTag$ and $prevTag$ (line~\ref{ln:SHORTtags}) of the round synchronization procedure, which starts when $p_i$ queries all reachable nodes and ends when it receives replies from all of these nodes (cf. lines~\ref{ln:SHORTifAllNetDiscovered}--\ref{ln:SHORTifAllNetDiscoveredEnd}, as well as, Section~\ref{sec:models}).
Upon receiving a query response, $p_i$ runs lines~\ref{ln:SHORTqueryReturn}--\ref{ln:SHORTrcv} and replies to other controllers' queries in lines~\ref{ln:SHORTarrivalCtrl}--\ref{ln:SHORTqryCon}.

A controller $p_i\in P_C$ keeps a local state of query replies (cf. Section~\ref{sec:SwitchesAndRules}) from other nodes (line~\ref{ln:SHORTlocal}). 
These replies allow $p_i$ to accumulate information about the network topology according to which the switch configurations are updated in each round.
The following three basic functionalities of Algorithm~\ref{alg:SHORTselfStabCode} are provided by the do-forever loop in lines~\ref{ln:SHORTdoForeverStart}--\ref{ln:SHORTsendUpdToSwitch}, which we detail below.

\Subsubsection{Establishing communication between any controller and any other node}
A controller $p_i \in P_C$ can communicate and manage a switch $p_j \in P_S$ only after $p_i$ has installed rules at all the switches on a path between $p_i$ and $p_j$. 
This, of course, depends on whether there are no permanent link failures on the path. 
In order to discover these link failures, we use local mechanisms for failure detection at each node for querying about the status of every link (cf. Section~\ref{sec:topologyDiscovery}). 
These mechanisms consider any permanent link failure as a transient fault and we assume that Algorithm~\ref{alg:SHORTselfStabCode} starts running only after the last occurrence of any transient fault (cf. Figure~\ref{fig:self-stab-SDN}). 
Thus, as soon as there is a flow installed between $p_i$ and $p_j$ and there are no permanent failures on the primary path (Section~\ref{sec:models}), $p_i$ and $p_j$ can exchange messages that arrive \textit{eventually} since it only depends on the temporary availability of the link which supports the communication fairness assumption (Section~\ref{sec:communicationFairness}).

The above iterative process of network topology discovery and the process of rule installation consider $\kappa$-fault-resilient flows (cf. Section~\ref{sec:kappaFlowsAboveZeroConstructing} and $myRules()$ function in Section~\ref{sec:models}).
These flows are computed through the interface $myRules(G, j, tag)$ (line~\ref{ln:SHORTmyRules}), where $G$ is the input topology, $p_j$ is the switch to store these rules, and $tag$ is the tag of the synchronization round.
Once the entire network topology is discovered, Algorithm~\ref{alg:SHORTselfStabCode} guarantees the installation of a $\kappa$-fault-resilient flow between $p_i$ and $p_j$. Thus, once the system is in a legitimate state, the availability of $\kappa$-fault-resilient flows implies that the system is resilient to the occurrence of at most $\kappa$ temporary link failures (and recoveries)
and $p_i$ can communicate with any node in the network within a bounded time.

\Subsubsection{Discovering the network topology and dealing with unreachable nodes}
Algorithm~\ref{alg:SHORTselfStabCode} lets the controllers connect to each other via $\kappa$-fault-resilient flows. 
Moreover, Algorithm~\ref{alg:SHORTselfStabCode} can detect situations in which controller $p_k \notin P_C$ is not reachable from controller $p_i$ (line~\ref{ln:SHORTstale}). 
The reason is that $p_i$ is guaranteed to (i) discover the entire network eventually, and (ii) communicate  with any node in the network. 
This means that $p_i$ eventually gets a response from every node in the network. 
Once that happens, the set of nodes that respond to $p_i$ equals to the set of nodes that were discovered by $p_i$ (line~\ref{ln:SHORTifAllNetDiscovered}) and thus $p_i$ can restart the process of discovering the network (line~\ref{ln:SHORTifAllNetDiscoveredEnd}). 

The start of a new round (in which $p_i$ rediscovers the network) allows $p_i$ to also remove information at the switches that is related to any unreachable controller $p_k \in P_C$, only when it has succeeded in discovering the network and bootstrapped communication.
We note that, during new rounds (line~\ref{ln:SHORTifSynch}), $p_i$ removes information related to $p_k$ from any switch $p_j$ (line~\ref{ln:SHORTrm}); whether this information is a rule or $p_k$'s membership in $p_j$'s management set. 
This stale information clean-up eventually brings the system to a legitimate state, 
as we will prove in Section~\ref{sec:proof}. 

Recall that we regard the long-term failure 
of links (or of more than $\kappa$ links) as transient faults. 
After the occurrence of the last transient fault, the network returns to 
fulfill our assumptions about the topology $G_c$, i.e., $G_c$ 
is $(\kappa+1)$-edge connected. 
Then, Algorithm~\ref{alg:SHORTselfStabCode} brings the system 
back to a legitimate state (Section~\ref{sec:proof}). 
The do-forever loop of Algorithm~\ref{alg:SHORTselfStabCode} completes by sending rule and manager updates to every switch that has a reply in $replyDB$, as well as querying every reachable node, with the current synchronization round's tag (lines~\ref{ln:SHORTsendStart}--\ref{ln:SHORTsendUpdToSwitch}).

\LinesNumbered

\begin{figure*}[t]
\begin{\VCalgSize}
	\begin{framed}
	
	\noindent \textbf{Symbols and operators:} `$\bullet$' stands for `any sequence of values', $()$ is the empty sequence, $\circ$  (binary) is the sequence concatenation operator and $\bigcirc$ (unary) concatenates a set's items in an arbitrary order.
	
	\noindent \textbf{Constants:} $N_c(i) \subseteq P$, $p_i$'s directly connected nodes. 
	$maxRules$ and $maxManagers$, maximum number of rules and managers, respectively.
	$maxReplies$: maximum size of the set $replyDB$.
	
	\noindent \textbf{Interfaces:} 
	Recall the interface function $myRules(G, j, tag)$, which creates $p_i$'s rules at switch $p_j$ according to $G$ with tag $tag$ (Section~\ref{sec:kappaFlowsAboveZeroConstructing}).  The interface between controller $p_i\in P_C$ and the abstract switch $p_j$ appears in the table below.
\begin{flushleft}
	\begin{scriptsize}
		\hspace*{-0.2cm}\begin{tabular}{|c|c|c|}
			\hline
			\textbf{Command type} & \textbf{Command} & \textbf{Switch $p_j$'s control module action}\\ \hline
			new round &$\langle \text{`$newRound$'}, t_{metaRule} \rangle$& updates current synchronization tag of the switch\\ \hline
			\multirow{4}{*}{update command}& $\langle `\text{$delMngr$'}, k \rangle$ & deletes $p_k$ from $manager(j)$\\ \cline{2-3} 
			& $\langle \text{`$addMngr$'},k \rangle$ & adds $p_k$ in $manager(j)$\\ \cline{2-3}
			& $\langle \text{`$delAllRules$'}, k \rangle$ & deletes all rules of $p_k$\\ \cline{2-3} 
			& $\langle \text{`$updateRule$'}, newRules \rangle$ & replaces all rules of $p_i$ with $newRules$\\ \hline
			query command & $\langle \text{`$query$'}, t_{query} \rangle$ & sends query response $m(j)$ to $p_i$\\ \hline
		\end{tabular}
	\end{scriptsize}
\end{flushleft}
		\noindent \textbf{Local state:} A controller's local state is the set $replyDB$ which stores the most recently received query replies. 
	A query reply $m = \langle ID, N_c, Mng, rules \rangle$ includes the respondent's ID, $m.ID \in P$, its communication neighborhood, $m.N_c \subseteq P$, its set of managers, $m.Mng \subseteq P_C$, and its set of installed rules, $m.rules$.
	A rule $r = \langle cID, sID, src, dest, prt, fwd, tag \rangle \in m.rules$ includes the switch's ID, $r.sID$, the ID of the controller which installed the rule, $r.cID$, the source and destination fields, $r.src$, and respectively, $r.dest$, the rule's priority $r.prt$, the ID of the neighbor to which the packet should be forwarded, $r.fwd$, and the rule's tag, $r.tag$, where 
	$r.sID, r.fwd, r.dest \in P$, $r.cID, r.src\in P_C$, $r.prt\in \{0,\ldots, n_{prt}\}$, and $r.tag\in tagDomain$. A command record $x$ includes the switch's ID, $x.sID$, and the command, $x.cmd$;
	$currTag$ and $prevTag$ are $p_i$'s current, and respectively, previous synchronization round tags;

\end{framed}
	\caption[Variables]{\label{fig:recSAvars}\label{fig:interface}\footnotesize{A list of symbols, operators, constants, interfaces and variables in Algorithm~\ref{alg:selfStabCode}.}}
\end{\VCalgSize}
\end{figure*} 
 
\begin{algorithm*}[t!]
	\begin{\VCalgSize}
		
				\noindent \textbf{Local state}: 
		$replyDB$$\subseteq$$\{ m(j)$$:=$$\langle j,$$N_c(j),$$manager(j),$$rules(j) \rangle \}_{p_j \in P }$\label{ln:local}\;
		 
		$currTag$ and $prevTag$ are $p_i$'s current and previous tags respectively\label{ln:tags}\;
		
		\noindent \textbf{Macros:}
		$res(x)=\{m \in replyDB : \forall_{r \in m.rules}\, r.tag = x\}\cup \{\langle i, N_c(i), \emptyset, \emptyset \rangle\}$\label{ln:resX}\; 
		$G(S)$$:=$$(\{p_k$$:$$\exists_{m \in S}$$:$$(m.ID$$=$$k$$\lor$$p_k \in m.N_c)\}, \{ (j,k)$$:$$\exists_{m \in S} : (m.ID = j \land p_k \in  m.N_c\})$\label{ln:G}\;
		$\dis := res(currTag) \cup \{ m \in res(prevTag): \nexists_{m'\in res(currTag)} m'.ID = m.ID\}$\label{ln:dis}\;		
		$p_j\rightarrow_G p_k:=$ true if there is a path from $p_j$ to $p_k$ in $G$\label{ln:arrow}\;
		
					
		\noindent \textbf{do forever} \Begin{\label{ln:doForeverStart} \tcc{Use replies from reachable senders with $prevTag$ or $currTag$} 
						
			$replyDB \gets \{m \in replyDB : 
			m.ID = k \neq i \land 
			(\exists_{x \in \{currTag, prevTag \}} 
			m \in res(x) \land 
			p_i \rightarrow_{G(res(x))} p_k 
			\} 
			\cup \{ \langle i, N_c(i), \emptyset, \emptyset \rangle \}$\label{ln:stale}\;
			%
			
			\textbf{let} $(newRound,msg)$$:=$$(false,\emptyset)$\label{ln:newRoundfalse}\tcc*{$newRound$ and $msg$ get defaults}
			
			\tcc{a new round with a new tag; remove replies with tag $currTag$}
			\If{$\forall_{p_{\ell}\in G(res(currTag))} (p_i$$\rightarrow_{G(res(currTag))}$$p_{\ell}$$\implies$$\exists_{m \in res(currTag)} m.ID = \ell)$\label{ln:ifAllNetDiscovered}}{
				$(newRound,prevTag) \gets (true, currTag)$;
				$currTag \gets nextTag()$\label{ln:ifAllNetDiscoveredStart}\;
				$replyDB \gets replyDB \setminus res(currTag)$\label{ln:ifAllNetDiscoveredEnd}\;
			}
			
			\tcc{The reference tag is $currTag$ only when the topology changes}
			
			\lIf{$G(\dis)=G(res(prevTag))$}{\textbf{let} $\mathit{referTag} := prevTag$ \textbf{else let} $\mathit{referTag} := currTag$\label{ln:GDisGresprevTag}}
			
			\ForEach{$p_j \in P_S : \exists_{m \in res(referTag)}\, m.ID = j$\label{ln:forEachSwitch}}{
				
				\tcc{On new rounds, remove unreachable or rule-less managers}
				\textbf{let} $M:=\{p_k \in m.Mng: (\exists_{r \in m.rules}\, r.cID= k) \land$ $(\neg newRound \lor p_i \rightarrow_{G(res(prevTag))} p_k) \}\cup \{p_i\}$\label{ln:rm}\;
				
				
				${msg}$$\gets$$msg$$\cup$$\{(p_j,$$\langle \text{`}delMngr\text{'},k \rangle)$$:$$p_k$$\in$$(m.Mng$$\setminus$$M)\}$$\cup$$\{(p_j,$$\langle \text{`}addMngr\text{'}, i \rangle)\}$\label{ln:msgM}\; 
				
				\tcc{Remove any $p_j$'s rule related to an unreachable node, $p_k$}
				
				${msg} \gets {msg} \cup \{(p_j,\langle \text{`}delAllRules\text{'}, k \rangle): (\exists_{r\in m.rules}\, r.cID =k) \land p_k \notin M\}$\label{ln:delAllRules}\;
				

				\tcc{$p_i$ refreshes its rules at switch $p_j$ with $referTag$}
				
				${msg}$$\gets$${msg}$$\cup$$\{(p_j,$$\langle \text{`}updateRule\text{'},$$myRules(G(res(\mathit{referTag})),$$j,$$currTag) \rangle)\}$\label{ln:msgL}\; 
				
			}

			\tcc{Send prepared messages to all reachable nodes aggregately }
			
			\lForEach{$p_j : p_i \rightarrow_{G(\dis)} p_j$}{\textbf{send} $(\langle \text{`}newRound\text{'}, currTag \rangle)\circ \bigcirc \{x.cmd : x\in msg \land x.sID = j \} \circ (\langle \text{`}query\text{'}, currTag \rangle)$ \textbf{to} $p_j$\label{ln:sendUpdToSwitch}}}
		

		\noindent \textbf{upon query reply $m$ from $p_j$\label{ln:queryReturn}} \Begin{
			\tcc{make space for $m$ (C-reset) and tests $m$'s tag is $prevTag$}
			\lIf{$|replyDB\cup\{m\}|>maxReplies$\label{ln:checkResponsesSize}}{$replyDB \gets \{\langle i, N_c(i), \emptyset, \emptyset \rangle\}$\label{ln:resetResponses}}
			\If{$(\exists_{r \in m.rules} r.tag = currTag)$}{$replyDB \gets (replyDB \setminus \{m'\in replyDB : m'.ID = m.ID\}) \cup \{m\}$}\label{ln:rcv}
		}
		
		
		\noindent \textbf{upon arrival of $(\bullet \circ (\langle \text{`}query\text{'}, tag \rangle))$ from $p_j$} {\textbf{do send} $\langle i, N_c(i), \bot, \{\langle j, i, \bot,\bot,\bot,\bot, tag \rangle\} \rangle$ \textbf{to} $p_j$;\label{ln:qryCon}}
		
		
	\end{\VCalgSize}
			\caption{\label{alg:selfStabCode}\footnotesize{Self-stabilizing algorithm for SDN control plane, controller $p_i$'s code (Algorithm~\ref{alg:SHORTselfStabCode}'s detailed version with definitions at Figure~\ref{fig:interface}).}}
\end{algorithm*}

\remove{
\begin{algorithm*}[t!]
	\begin{\VCalgSize}

		\noindent \textbf{Variables:} 
		$replyDB \subseteq \{ m(j) : p_j \in P \}$ most recently received replies $m(j)$, $p_j \in P$, 
		where $m(j)$ := $\langle j$, $N_c(j)$, $manager(j)$, $rules(j) \rangle$, $N_c(j)$ is $p_j$'s neighbors, $manager(j) \subseteq$ $P_C$ has $p_j$'s managers, and $rules(j) \subseteq \{ \langle k, j, {src}, {dest}, prt, z, tag \rangle : (p_k, p_j, p_z, p_{dest} \in P) \land p_{src}\in P_C \land prt\in \{0,\ldots, n_{prt}\} \land tag\in tagDomain \}$ is $p_j$'s rule set\label{ln:local}; 
		$currTag$ and $prevTag$ are $p_i$'s current, and resp., previous tags\label{ln:tags}\;
		
		\noindent \textbf{Macros:} $res(x)=\{\langle \bullet, rules(j) \rangle \in replyDB : \forall_{r \in rules(j)}\, r = \langle \bullet, x \rangle \}\cup \{\langle i, N_c(i), \emptyset$, $\emptyset \rangle\}$\label{ln:resX}; $G(S):=(\{p_k :\exists_{\langle j, N_c(j), \bullet \rangle \in S} : (k=j \lor p_k \in N_c(j))\}$, $\{ (j,k): \exists_{\langle j, N_c(j), \bullet \rangle \in S} :$ $ (p_k \in  N_c(j))\})$\label{ln:G}; $\dis := res(currTag) \cup \{ \langle k, \bullet, prevTag \rangle \in res(prevTag): \langle k, \bullet,$ $ currTag \rangle \notin res(currTag)\}$\label{ln:dis}; $p_j\rightarrow_G p_k:=$ true if $G$ has a path from $p_j$ to $p_k$\label{ln:arrow}\;

		\noindent \textbf{do forever} \Begin{\label{ln:doForeverStart} \tcc{Use replies from reachable senders with $prevTag$ or $currTag$} 
			
			$replyDB \gets \{\langle k, \bullet, rules \rangle \in replyDB : k\neq i \land (\exists_{x \in \{currTag, prevTag \}}$ $\langle k, \bullet, rules \rangle \in res(x) \land p_i \rightarrow_{G(res(x))} p_k \land  
			\langle i, \bullet, x \rangle \in rules) \} \cup \{ \langle i, N_c(i), \emptyset, \emptyset \rangle \}$\label{ln:stale}\;
			
			\textbf{let} $(newRound,msg)$$:=$$(false,\emptyset)$\label{ln:newRoundfalse}\tcc*{$newRound$ and $msg$ get defaults}
			
			\tcc{a new tag for a new round; remove replies with that tag}
			\If{$(\forall p_{\ell} : p_i \rightarrow_{G(res(currTag))} p_{\ell} \implies \langle \ell, \bullet \rangle \in res(currTag)$\label{ln:ifAllNetDiscovered}}{
				$(newRound,prevTag) \gets (true, currTag)$;
				$currTag \gets nextTag()$\label{ln:ifAllNetDiscoveredStart}\;
				$replyDB \gets replyDB \setminus \{\langle j,\bullet\rangle \in res(currTag) : p_j\in P \}$\label{ln:ifAllNetDiscoveredEnd}\;
			}
			
			\tcc{The reference tag is $currTag$ only when the topology changes}
			
			\lIf{$G(\dis)=G(res(prevTag))$}{\textbf{let} $\mathit{referTag} := prevTag$ \textbf{else let} $\mathit{referTag} := currTag$\label{ln:GDisGresprevTag}}
			
			\ForEach{$p_j \in P_S : \langle j,Ngb,Mng,Rul \rangle \in res(referTag)$\label{ln:forEachSwitch}}{
				
				\tcc{On new rounds, remove unreachable or rule-less managers}
				\textbf{let} $M$$:=$$\{p_i\} \cup \{p_k$$\in$$Mng$$:$$(\exists_{r \in Rul}\, r$$=$$\langle k, \bullet \rangle) \land (p_i$$\rightarrow_{G(res(prevTag))}$$p_k\lor \neg newRound) \}$\label{ln:rm}\;
				
				${msg}$$\gets$$msg$$\cup\{(p_j,\langle \text{`}delMngr\text{'},k \rangle)$$:$$p_k \in (Mng \setminus M)\} \cup \{(p_j,\langle \text{`}addMngr\text{'}, i \rangle)\}$\label{ln:msgM}\; 
				
				\tcc{Remove any $p_j$'s rule related to an unreachable node, $p_k$}
				
				${msg} \gets {msg} \cup \{(p_j,\langle \text{`}delAllRules\text{'}, k \rangle): (\exists_{r\in Rul}\, r =\langle k, \bullet \rangle) \land p_k \notin M\}$\label{ln:delAllRules}\;
				
				\tcc{$p_i$ refreshes its rules at switch $p_j$ with $referTag$}
				
				${msg}$$\gets$${msg}$$\cup$$\{(p_j,$$\langle \text{`}updateRule\text{'},$$myRules(G(res(\mathit{referTag})),$$j,$$currTag) \rangle) \}$\label{ln:msgL}\; 
			}
						
			\tcc{Send the prepared messages to all reachable nodes } 
			
			\lForEach{$p_j : p_i \rightarrow_{G(\dis)} p_j$}{\textbf{send} $(\langle \text{`}newRound\text{'}, currTag \rangle)\circ [\bigcirc_{m: (p_j,m) \in msg} (m)]\circ (\langle \text{`}query\text{'}, currTag \rangle)$ \textbf{to} $p_j$\label{ln:sendUpdToSwitch}}}
		
		\noindent \textbf{upon query reply $m:=\langle j,\bullet, rls \rangle$ from $p_j$\label{ln:queryReturn}} \Begin{
			\tcc{make space for $m$ (C-reset) and tests $m$'s tag is $prevTag$}
			\lIf{$|replyDB\cup\{m\}|>maxReplies$\label{ln:checkResponsesSize}}{$replyDB \gets \{\langle i, N_c(i), \emptyset, \emptyset \rangle\}$\label{ln:resetResponses}}
			\lIf{$(\exists_{r \in rls}\, r = \langle \bullet, currTag \rangle)$}{$replyDB \gets (replyDB \setminus \{ \langle j, \bullet \rangle\}) \cup \{m\}$}\label{ln:rcv}
		}
		
		\noindent \textbf{upon arrival of $(\bullet \circ (\langle \text{`}query\text{'}, tag \rangle))$ from $p_j$} {\textbf{do send} $\langle i, N_c(i), \bot, \{\langle j, i, \bot,\bot,\bot,\bot, tag \rangle\} \rangle$ \textbf{to} $p_j$\label{ln:qryCon}}
		
			\caption{\label{alg:selfStabCode}\footnotesize{Self-stabilizing algorithm for SDN control plane, controller $p_i$'s code (Algorithm~\ref{alg:SHORTselfStabCode}'s detailed version).}}
	\end{\VCalgSize}
\end{algorithm*}
} 

\Subsection{Refining the model: variables, building blocks, and interfaces}
\label{s:refiningModel}

After the provision of a high-level description of the proposed solution in Algorithm~\ref{alg:SHORTselfStabCode}, we provide the solution details in Algorithm~\ref{alg:selfStabCode}, which requires more notation, interfaces, and building blocks.

\paragraph{Local Variables}
Each controller's state includes $replyDB$ (line~\ref{ln:local}), which is the set of the most recent query replies, and the tags $currTag$ and $prevTag$, which are $p_i$'s current, and respectively, previous synchronization round tags. Each response $m(j) \in replyDB$ can arrive from either a switch or another controller and it has the form $\langle j, N_c(j), manager(j), rules(j) \rangle$, for $p_j \in P$. The code denotes by $N_c(j)$ the neighborhood of $p_j$, by $manager(j) \subseteq P_C$ the controllers of $p_j$, and by $rules(j) \subseteq \{ \langle k, j$, ${src}$, ${dest}$, $prt$, $z, tag \rangle :$ $(p_k, p_j, p_z, p_{dest} \in P) \land (p_{src} \in P_C) \land prt \in \{0,\ldots, n_{prt}\} \land tag \in tagDomain \}$ the rule set of $p_j$. Throughout Algorithm~\ref{alg:selfStabCode} and for ease of presentation we refer to the elements of responses and rules using the struct notation, which is used by the C programming language. We refer to the fields of $m = \langle ID, N_c, Mng, rules\rangle$ stated above, by $m.ID = j$, $m.N_c = N_c(j)$, $m.Mng = manager(j)$, and $m.rules = rules(j)$. We assume that the size of $replyDB$ is bounded by $maxReplies\geq 2(N_C + N_S)$, hence the local state has bounded size (the factor of 2 is due to responses from the rounds  $prevTag$ and $currTag$).

\paragraph{An internal building block: round synchronization}
\label{s:sync}
An SDN controller 
accesses the abstract switch in synchronized rounds. 
Each round has a unique tag that distinguishes the given 
round from its predecessors. We assume access to a 
self-stabilizing algorithm that generates unique \textit{tags} 
of bounded size from a finite domain of tags, $tagDomain$.
The algorithm provides a function called $nextTag()$ that, 
during a legal execution, returns a unique tag.
That is, immediately before calling $nextTag()$ there is no 
tag anywhere in the system that has the returned value from that call. 
Given two tags, $t_1$ and $t_2$, 
we require that $t_1=t_2$ holds if, and only if, they have 
identical values. 
We use these tags for synchronizing the rounds in which the controllers 
perform configuration updates and queries. 
Namely, in the beginning of a round, controller $p_i \in P_C$ 
generates a new tag and stores that tag in the variable $currTag \gets nextTag()$. 
Controller $p_i$ then attempts to install at every reachable 
switch $p_j\in P_S$ a special meta-rule $\langle i, j, \bot, \bot, n_{prt},  \bot, t_{metaRule} \rangle$, which includes, in addition to $p_i$'s identity, the tag 
$t_{metaRule}=currTag$ and has the lowest priority (before making any configuration update on that switch). 
It then sends a query to all (possibly) reachable nodes in the network and 
combines that query with the tag $t_{query}=currTag$. 
The response to that query from other controllers 
$p_j \in P_C$ includes the query tag, $t_{query}$. 
The response to the query from the switch $p_k \in P_S$ 
includes the tag $t_{metaRule}$ of the most recently installed 
meta-rule that $p_k$ has in its configuration. 
The controller $p_i$ ends its current round once it has received a 
response from every (possibly) reachable node in the network and 
that response has the tag of $currTag$. 

We note the existence of self-stabilizing algorithms, 
such as the one by Alon et al.~\cite{DBLP:journals/jcss/AlonADDPT15}, 
that in fair executions (that are legal with respect to the self-stabilizing 
end-to-end communication protocol) provide unique tags within a 
number of synchronization rounds that is bounded (by a constant whenever 
the execution is legal with respect to the self-stabilizing 
end-to-end communication protocol).
We refer to that known bound by $\Delta_{synch}$ and note that 
during a legal execution of the round synchronization algorithm, 
it holds that controller $p_i$ receives only a response message 
$m$ that matches $currTag$, i.e., it discards any message with a different tag. 
Moreover, since during legal executions $nextTag()$ returns 
only unique tags, $m$ and its acknowledgment are guaranteed 
to form a complete round-trip. 
Note that we do not require $nextTag()$ to support concurrent 
calls since every controller manages its own synchronization rounds; 
one round at a time.
We note the existence of other relevant synchronizers, such as the $\alpha$-synchronizer by Awerbuch et el.~\cite{DBLP:journals/tdsc/AwerbuchKMPV07,D2K}, which have simpler tags than~\cite{DBLP:journals/jcss/AlonADDPT15}. However, we prefer the elegant interface defined in~\cite{DBLP:journals/jcss/AlonADDPT15}.

\paragraph{Interfaces}
Controller $p_i$ can send requests or \emph{queries}
to any other node $p_j$ 
(which could be either another controller or a switch).
We detail the switch interface below and illustrate it in Figure~\ref{fig:interface}.


The controllers send command batches, 
which are sequences of commands. 
The  special 
metadata command $\langle \text{`}newRound\text{'}, t_{metaRule} 
\rangle$ is always the first command and updates the special meta-rule to store $t_{metaRule}$. 
We use it for starting 
a new round (where $t_{metaRule}=t$
is the round's tag).
This starting command could be followed by a number of commands, 
such as  $\langle \text{`}delMngr\text{'}, k \rangle$ for the 
removal of controller $p_k$ from the management of 
switch $p_j$, $\langle \text{`}addMngr\text{'},k \rangle$ for the addition 
of controller $p_k$ from the management of switch 
$p_j$, and $\langle \text{`}delAllRules\text{'}, k \rangle$ for the 
deletion of all of $p_k$'s rules from the configuration 
of switch $p_j$, where $p_k \in P_C\setminus \{p_i\}$. 
The rules' update is done via $\langle \text{`}updateRule\text{'}, 
newRules \rangle$ and it is the second last command. 
This update replaces all of $p_i$'s rules at switch $p_j$ 
(except for the special meta-rule) with the rules in $newRules$.
These commands are to be followed by the round's 
query $\langle \text{`}query\text{'},t_{query} \rangle$, where 
$t_{query}=t$ is the query's tag. The switch $p_j$ 
replies to a query by sending $m=\langle j, N_c(j)$, 
$manager(j)$, $rules(j) \rangle$ to $p_i$, such that the rule 
set includes also the special meta-rule $\langle i, \bullet, t 
\rangle \in rules(j)$. Whenever $p_j \in P_C$ is another controller, 
response to a query is simply $\langle i, N_c(i), \bot, 
\{\langle j, i, \bot,\bot,\bot,\bot, t_{query} \rangle\} \rangle$ 
(line~\ref{ln:qryCon}). Note that controller $p_j$ simply ignores 
all other types of commands.  
We use the interface function $myRules(G, j, tag)$ (Section~\ref{sec:kappaFlowsAboveZeroConstructing}) for 
creating the packet forwarding rules that controller $p_i$ installs at switch $p_j$ when $p_i$'s current view on the network topology is $G$ in round $tag$ (Figure~\ref{fig:recSAvars}). 

\Subsection{Algorithm details}
\label{sec:description}
Algorithm~\ref{alg:selfStabCode} presents the proposed solution with a greater degree of details than Algorithm~\ref{alg:SHORTselfStabCode}. 
Algorithm~\ref{alg:selfStabCode} is centered around 
a \emph{do forever} loop, which starts by removing stale 
information from $replyDB$ (line~\ref{ln:stale}). 
This removal action includes refreshing information related 
to controller $p_i$, which deletes information about any node 
that is not reachable from $p_i$. The reachability test 
uses the currently known information about the network topology, 
$G$ and the relation $\rightarrow_G$ (line~\ref{ln:arrow})
that tells whether node $p_j$ is reachable 
from controller $p_i$ in $G$, given the information in $replyDB$. 

Algorithm~\ref{alg:selfStabCode} accesses 
the switch configurations in synchronization rounds. 
Lines~\ref{ln:newRoundfalse}--\ref{ln:ifAllNetDiscoveredEnd} manage the start (and end) 
of synchronization rounds. When a new round starts, 
i.e., the condition of the if-statement of line~\ref{ln:ifAllNetDiscovered} 
holds, controller $p_i$ marks the start of a new round ($newRound_i=true$), 
updates the values of the tags $prevTag_i$ and $currTag_i$ and clears 
any record with tag $currTag$ of the replies stored in $replyDB_i$ (line~\ref{ln:ifAllNetDiscoveredStart} and~\ref{ln:ifAllNetDiscoveredEnd}).

Algorithm~\ref{alg:selfStabCode} refreshes (and reconstructs) the information about remote nodes (controllers and switches including the ones that are directly attached to it) by sending queries (line~\ref{ln:sendUpdToSwitch}) and updating the set of stored replies (line~\ref{ln:rcv}). Notice that controller $p_i$ also responds to query requests coming from other controllers (line~\ref{ln:qryCon}). Algorithm~\ref{alg:selfStabCode} uses these replies for completing the information about the switches that are directly connected to a remote controller (and thus the other fields in the response messages are the empty sets). 

The heart of Algorithm~\ref{alg:selfStabCode} includes the 
updates of every switch $p_j \in P_S$ (line~\ref{ln:forEachSwitch} to~\ref{ln:delAllRules}). 
For every switch $p_j$ (line~\ref{ln:forEachSwitch}), 
controller $p_i$ considers $p_j$'s stored response 
$\langle j, Ngb_i, Mng_i, Rul_i\rangle$ for which it prepares a set of commands 
to be stored in the set $msg_i$ (lines~\ref{ln:newRoundfalse},~\ref{ln:msgM},~\ref{ln:delAllRules},~\ref{ln:msgL} and~\ref{ln:sendUpdToSwitch}).    To that end, 
$p_i$ first calculates the set of managers that $p_j$ should 
have in the following manner. If this iteration of the do 
forever loop (lines~\ref{ln:doForeverStart} to~\ref{ln:sendUpdToSwitch}) 
is the first one for the round $currTag_i$, the value of $newRound_i$ is 
$true$ (line~\ref{ln:ifAllNetDiscoveredStart});  this leads $p_i$ to 
remove any controller $p_k$ that is not reachable according 
to ${G(res(prevTag))}$ (lines~\ref{ln:rm} to~\ref{ln:delAllRules}). 
Whenever the iteration is not the first one, 
$p_i$ merely asserts that it is a manager of $p_j$. 

Controller $p_i$ removes any rules of an unreachable 
controller $p_k$ (line~\ref{ln:delAllRules}) and 
updates all of its rules at switch $p_j$ (line~\ref{ln:msgL}) 
using the interface function $myRules()$ (line~\ref{ln:msgL}) 
and the reference tag, $\mathit{referTag}$ 
(line~\ref{ln:G} and line~\ref{ln:GDisGresprevTag}). 
The proposed algorithm 
selects $\mathit{referTag}$'s value to be $prevTag$ during 
legal executions. During recovery periods, the discovered 
topology can differ from that one that is stored with the tag $prevTag$. 
In that case, the algorithm selects $currTag$ as the reference tag. 
After preparing these commands to all the switches, controller $p_i$ 
prepares query commands to all reachable nodes (including both 
controllers and switches) and then sends all prepared commands 
to their designated destinations. Note that each of these configuration 
updates are done via a single message that aggregates all 
commands for a given destination (line~\ref{ln:sendUpdToSwitch}). 

We note that when a query response arrives at $p_i$, before the update of the response set (line~\ref{ln:rcv}), $p_i$ checks that there is sufficient storage space for the arriving response (line~\ref{ln:resetResponses}). If space is lacking, $p_i$ performs what we call a `C-reset'. Note that $p_i$ stores replies only for the current synchronization round, $currTag$.

\section{Correctness Proof}
\label{sec:proof}
We prove the correctness of Algorithm~\ref{alg:selfStabCode} by 
showing that when the system starts in an arbitrary state, 
it reaches a legitimate state (Definition~\ref{def:legitimateState}) 
within $$\bigO{(((\Delta_{comm} + \Delta_{synch}))D)  [((\Delta_{comm} + \Delta_{synch})D) \cdot N_S + N_C])}$$ frames (Theorem \ref{lem:staleResetIllegalDelete}). Moreover, we show that when 
starting from a legitimate state, the system satisfies the task 
requirements and it is also resilient to a bounded number of 
failures (lemmas~\ref{thm:kappaLink} and~\ref{lem:nodeAdditionDeletion}). 


We refer to the values of variable $X$ at 
node $p_i$ (controller or switch) as $X_i$, i.e., 
the variable name with a subscript that indicates the node index. 
Similarly, we refer to the return values of function $f$ 
at controller $p_k$ as $f_k$.

\begin{definition}[Legitimate System State]
	\label{def:legitimateState}
	State $c \in R$ is legitimate with respect to Algorithm~\ref{alg:selfStabCode} when, for every controller $p_i \in P_C$ and node $p_k \in P\setminus \{p_i\}$, the following conditions hold.
	\begin{enumerate}
		\item \label{def:legitimateState:dis} 
		$\langle k$, $N_c(k)$, $manager(k)$, $rules(k) \rangle$ $\in$ $replyDB_i$ if, 
		and only if, $N_c(k)$, $manager(k)$, and $rules(k)$ are $p_k$'s 
		neighborhood, managers, and respectively, set of packet 
		forwarding rules (line~\ref{ln:local}) as well as $p_i \rightarrow_G p_k$ (line~\ref{ln:arrow}). 
		Moreover, for the case of controller $p_k \in P_C$, the task does 
		not require $p_k$ to have any managers or rules, 
		i.e., $manager(k) = \emptyset$ and $rules(k) = \emptyset$. 
		\item \label{def:legitimateState:manage} Any controller 
		is the  manager of every switch and only these controllers 
		can be the mangers 
		of any switch, i.e., $p_i \in P_C \land p_k \in P_S \iff p_i \in manager(k)$.
		\item \label{def:legitimateState:flow} The rules installed in the 
		switches encode $\kappa$-fault-resilient flows between controller 
		$p_i$ and node $p_k$ in the network $G_c$ 
		(Section~\ref{sec:kappaFlowsAboveZeroConstructing}).
		\item \label{def:legitimateState:end2end} The end-to-end protocol (Section~\ref{sec:transportModel}) as well as the round 
		synchronization protocol (Section~\ref{sec:topologyDiscovery}) 
		between $p_i$ and $p_{k}$ are in a legitimate state.
	\end{enumerate}
\end{definition}

\Subsection{Overview}

The proof of Theorem \ref{lem:staleResetIllegalDelete} starts by establishing bounds on the number of rules that each switch needs to store (Lemma~\ref{thm:boundedSwitchMemory}). 
The proof arguments are based on the bounded network size and the memory management scheme of the abstract switch (Section~\ref{sec:memory}), which guarantees that, during a legal execution, all non-failing controllers are able to store their rules (Lemma~\ref{thm:boundedSwitchMemory}). 
The bounded network size also helps to bound, during a legal execution, the amount of memory that each controller needs to have (Lemma~\ref{thm:boundedControllerMemory}). 
This proof also bounds the number of C-resets that a controller might take (line~\ref{ln:checkResponsesSize}) during the  period in which the system recovers from transient faults. 
This is line~\ref{ln:SHORTcheckResponsesSize} in Algorithm~\ref{alg:SHORTselfStabCode}.
Note that this bound on the number of C-resets is important because C-resets delete all the information that a controller has about the network state. 

C-resets are not the only disturbing actions that might occur during the recovery period. 
The system cannot reach a legitimate state before it removes stale information from the configuration of every switch. 
Note that failing controllers cannot remove stale information that is associated with them and therefore non-failing controllers have to remove this information for them. 
Due to transient faults, it could be the case that one controller can remove information that is associated with another non-failing controller. 
We refer to these `mistakes' as illegitimate deletion of rules or managers (Section~\ref{sec:illegiDel}). 
Note that illegitimate deletions occur when the (stale) information that a controller has about the network topology differ from the actual network topology, $G_c$. 
Moreover, due to stale information in the communication channels, any given controller might aggregate (possibly stale) information about the network more than once and thus instruct more than once the switch to delete illegitimately the rules of other controllers. 

Theorem~\ref{thm:boundedResetsIllegitimateDeletions} bounds the number of these illegitimate deletions. 
It does so by counting the number of possible steps in which a controller might have stale information about the network and that stale information leads the controller to perform an illegal deletion. 
The proof arguments start by considering a starting state in which controller $p_i \in P_c$ is just about to take a step that instructs the switches to perform illegitimate deletions. 
The proof then argues that between any two such steps, controller $p_i$ has to aggregate information about the network in such a way that $p_i$ \ems{(mistakenly) decides that it has completed the task of topology discovery.} 
%
%
But, this can only happen after receiving a reply from every node in the preserved topology (Claim~\ref{thm:whenDone}). 
By induction on the distance $k$ between controller $p_i \in P_c$ and node $p_j \in P\setminus \{ p_i \}$, the proof shows that the information that $p_i$ has about $p_j$ is correct within $k\cdot(\Delta_{comm}+\Delta_{synch}+1)+1$ times in which $p_i$ instruct the switches to perform an illegitimate deletion, because there is a bounded number of stale information in the communication channel between $p_i$ and $p_j$ (Lemma~\ref{thm:whenAfterDone}). 
Thus, the total number of illegitimate deletions is at most $D\cdot(\Delta_{comm}+\Delta_{synch}+1)+1$.

The proof demonstrates recovery from transient faults by considering a period in which there are no C-resets and no illegitimate deletions (Section~\ref{sec:recover}). 
In such a period, all the controllers construct $\kappa$-fault-resilient flows to any other node in the network (Lemma~\ref{thm:propagation}). 
This part of the proof is again by induction on the distance $k$ between controller $p_i \in P_c$ and node $p_j \in P\setminus \{ p_i \}$. 
The induction shows that, within $((\Delta_{comm}+\Delta_{synch})+2)k$ frames, $p_i$ discovers correctly its $k$-distance neighborhood and establishes a communication channel between $p_i$ and $p_j$. 
This means that within $((\Delta_{comm}+\Delta_{synch})+2)D$ frames in which there are no C-resets and no illegitimate deletions, the system reaches a legitimate state (Lemma~\ref{thm:noDelMnger}).


The above allows Theorem~\ref{lem:staleResetIllegalDelete} to show that \ems{after at most  $$\bigO(((\Delta_{comm}+\Delta_{synch}))D)[((\Delta_{comm}+\Delta_{synch})D)\cdot N_S + N_C])$$} frames in $R$, there is a period of \ems{$\bigO((\Delta_{comm}+\Delta_{synch}))D)$} frames in which there are no C-resets and no illegitimate deletions and thus the system reaches a legitimate state. 
Lemma~\ref{thm:kappaLink} shows that, when starting from a legitimate state an then letting a single link in the network to be added or remove from $G_c$, the system recovers within $\bigO(D)$ frames. 
The arguments here consider that number of frames it takes for each controller to notice the change and to update all the switches. 
By similar arguments, Lemma~\ref{lem:nodeAdditionDeletion} shows that after the addition or removal of at most $N_C -1$ controllers, the system reaches a legitimate system state within $\bigO(D)$ frames.  

\Subsection{Analysis of memory and message size requirements}
\label{sec:basicFacts}
Lemmas~\ref{thm:boundedSwitchMemory} 
and~\ref{thm:boundedControllerMemory} bound the needed 
memory at every node during a legal execution. 
Recall that we assume that the switches implement a mechanism 
for dealing with clogged memory (Section~\ref{sec:memory}), 
such that once controller $p_i \in P_C$ refreshes its rules on a 
given switch, that switch never removes $p_i$'s rules.

Lemma~\ref{def:legitimateState} considers 
an event that can delay recovery, i.e., the removal of a rule at a 
switch due to lack of space.
Lemma~\ref{def:legitimateState} bounds the needed memory for every switch, and thus relates to events that can delay recovery, i.e., the removal of a rule at a 
switch due to lack of space.

\begin{lemma}[Bounded Switch Memory]
	\label{thm:boundedSwitchMemory}
	(i) Suppose that $R$ is a legal execution of Algorithm~\ref{alg:selfStabCode}. 
	A switch needs to let no more than $maxManagers \geq  N_C$ controllers 
	to manage it and (2) no more than $maxRules  \geq N_C\cdot (N_C + N_S - 1)\cdot n_{prt}$ 
	packet forwarding rules. 
\end{lemma}

\begin{proof}
	Let $p_j\in P_S$ be a switch. 
	
	\noindent \textbf{Number of managers.}
	Recall that we assume that $maxManagers\geq N_C \geq 
	|P_C|$, i.e., the bound is large enough 
	to store all managers (once all stale information is removed in 
	a FIFO manner that is explained in Section~\ref{sec:memory}). 
	During a legal execution $R$ of Algorithm~\ref{alg:selfStabCode}, 
	every controller accesses every switch repeatedly (line~\ref{ln:sendUpdToSwitch}). 
	This way, every $p_i\in P_C$, is always among the $N_C$ most recently installed controllers at $p_j\in P_S$.

	\noindent \textbf{Number of rules.}
	Recall that a rule is a tuple of the form $\langle k$, 
	$i$, ${src}$, ${dest}$, $prt$, $j, tag \rangle$, 
	where $p_k \in P_C$ is the controller that created this rule, 
	$p_i \in P_S$ is the switch that stores this rule, 
	$p_{src} \in P_C$ and $p_{dest} \in P$ are the source, 
	and respectively, the destination of the packet, $prt$ is the 
	packet's priority, $p_j \in P$ is the relay node (i.e., the rule's action field) 
	and $tag$ is the synchronization round tag.
	
	To show that there are no more than $N_C\cdot (N_C + N_S - 1)\cdot n_{prt}$ rules that a switch needs to store, recall that each of the $N_C$ controllers $p_{src} \in P_C$ constructs $\kappa$-fault-resilient flows to every node $p_{dest} \in P \setminus \{p_{src}\}$ in the network. 
	Thus, switch $p_i \in P_S$ might be a hop on the $\kappa$-fault-resilient flow between $p_{src}$ and $p_{dest}$. 
	That is, there are at most $ N_C\cdot (N_C + N_S - 1)$ such flows that pass via $p_i$, because for each of the $N_C$ possible flow sources $p_{src}$, there are exactly $(N_C + N_S - 1)$ destinations $p_{dest}$. 
	Each such flow stores at most $n_{prt} \geq \kappa+1$ rules at $p_i$, i.e., one for each priority. 
	Note that, during a legal execution, each switch $p_i \in P_S$ stores at most one tag per $p_{src} \in P_C$ (line~\ref{ln:sendUpdToSwitch}).
\end{proof} 

Lemma~\ref{thm:boundedControllerMemory} considers an event C-reset, 
which can delay recovery.

\begin{lemma}[Bounded Controller Memory]
	\label{thm:boundedControllerMemory}
	(1)~Let $a_x\in R$ be the first step in which controller $p_i$ runs lines~\ref{ln:queryReturn}--\ref{ln:rcv} (upon query reply).
	For every state in $R$ that follows step $a_x$, 
	node $p_i$ stores no more than $maxReplies$ 
	replies in the set $replyDB_i$. 
	(2)~Suppose that $R$ is a 
	legal execution. 
	Controller $p_i\in P_C$ needs to store, 
	in the set $replyDB_i$, no more than $maxReplies \geq 2 \cdot (N_C +N_S) $ items. 
	(3)~Suppose that $R$ is any execution, which may start in an arbitrary state. 
	Controller $p_i$ performs a C-reset at most once in $R$, i.e., takes 
	a step $a_{x'}\in R$ that includes the execution of line~\ref{ln:checkResponsesSize}, in which the if-statement condition is true.
\end{lemma}

\begin{proof}
	\noindent \textbf{Part (1).~~}
	We note that $p_i$ modifies $replyDB_i$ only 
	in line~\ref{ln:stale} and line~\ref{ln:ifAllNetDiscoveredEnd} 
	in the do-forever loop (lines~\ref{ln:doForeverStart}--\ref{ln:sendUpdToSwitch}), 
	and in lines~\ref{ln:resetResponses} and~\ref{ln:rcv} in the query reply procedure (lines~\ref{ln:queryReturn}--\ref{ln:rcv}).
	In line~\ref{ln:stale} and line~\ref{ln:ifAllNetDiscoveredEnd}, 
	the size of $replyDB_i$ either decreases (possible only at the first step 
	that $p_i$ executes line~\ref{ln:stale} or line~\ref{ln:ifAllNetDiscoveredEnd}) 
	or stays the same.
	Thus, the rest of this proof focuses only at lines~\ref{ln:resetResponses} 
	and~\ref{ln:rcv}, where the set $replyDB_i$ increases due 
	to the addition of an incoming reply (line~\ref{ln:rcv}).
	
	Let $a_{x'}$ be the first step in $R$, in which controller 
	$p_i$ executes lines~\ref{ln:queryReturn}--\ref{ln:rcv} 
	due to a message $m_j$ that $p_i$ receives from node $p_j$. 
	By line~\ref{ln:checkResponsesSize}, if 
	$|replyDB_i \cup \{m_j\}| > maxReplies$ holds, 
	then $p_i$ performs a C-reset, i.e., sets 
	$replyDB_i \gets \{\langle i, N_c(i), \emptyset, \emptyset \rangle\}$, 
	which implies that $|replyDB_i| = 1$ after the execution of 
	line~\ref{ln:resetResponses}. 
	Hence, after the execution of line~\ref{ln:rcv} in step $a_{x'}$, 
	$|replyDB_i| < maxReplies$ holds for the state $c_{{x'}+1}$, 
	which follows $a_{x'}$ immediately.
	Similarly, since the size of $replyDB_i$ increases only when $p_i$ executes line~\ref{ln:rcv}, for every step $a_{x''}$ and the system state $c_{{x''}+1}$ 
	that appears in $R$ after $c_{{x'}+1}$, it is true that $|replyDB_i| 
	\leq maxReplies$ holds in $c_{{x''}+1}$, due to line~\ref{ln:resetResponses}.
	Thus, for every system state that follows the first step $a_{x'} \in R$,
	it holds that $|replyDB_i| \leq maxReplies$.
	
	\noindent \textbf{Part (2).~~}
	Line~\ref{ln:stale} removes from 
	$replyDB_i$ any response 
	that its synchronization round tag is not 
	in the set $\{prevTag_i,currTag_i\}$ and line~\ref{ln:rcv} 
	does not add to $replyDB_i$ a response that its synchronization 
	round tag is not $currTag_i$. Moreover, line~\ref{ln:ifAllNetDiscoveredEnd} 
	makes sure that when finishing one synchronization 
	round and then transitioning to the next one, $replyDB_i$ 
	includes replies only with synchronization round tags that are 
	$prevTag_i$. Therefore, there are no more than two 
	synchronization round  tags that could be simultaneously present in 
	$replyDB_i$. 
	Moreover, line~\ref{ln:stale} also removes any response from an 
	unreachable node, because item~\ref{def:legitimateState:dis} of Definition~\ref{def:legitimateState} holds in any system state of 
	a legal execution. This further limits the set $replyDB_i$ to 
	includes response from at most $N_C + N_S$ nodes. Therefore, 
	$|replyDB_i| \leq 2 \cdot (N_C + N_S)$.
	
	\noindent \textbf{Part (3).~~}
	Suppose that $p_i$ does perform a C-reset during $R$. Once that 
	happens, parts (i) and (ii) of this proof imply that this can never happen again.
\end{proof}

Lemma~\ref{thm:message} demonstrates that the proposed algorithm requires bounded message size.  

\begin{lemma}
	\label{thm:message}
	The message size before and after the recovery period is in $\bigO( maxRules \log N)$, and respectively, $\bigO( \Delta N\log N )$ bits, where $N=N_C+N_S$ and $\Delta$ is the maximum node degree.
\end{lemma}
\begin{proof}
	The size of the messages sent differs during and after the recovery period.
	Algorithm~\ref{alg:selfStabCode} involves messages sent from a controller to any other node and their subsequent replies to the controller. A message from a controller to a switch is a set of commands $msg$ initialized to the empty set in line~\ref{ln:newRoundfalse}. Commands are appended in $msg$ in lines~\ref{ln:msgM}, \ref{ln:delAllRules}, and \ref{ln:msgL}, before a controller appends two more commands to $msg$ (line~\ref{ln:sendUpdToSwitch}) and sends it to a switch. We denote with $msg_{\ref{ln:msgM}}$, $msg_{\ref{ln:delAllRules}},$ $msg_{\ref{ln:msgL}}$ the sets of commands appended to $msg$ in the respective lines. Thus, $|msg| = |msg_{\ref{ln:msgM}}| + |msg_{\ref{ln:delAllRules}}| + |msg_{\ref{ln:msgL}}| + \bigO(\log c_{tag})$ bits, where $|msg_{x}|$ refers to the message size due to line x and $c_{tag}$, is the maximum size of a tag. Note that when using tags based on the ones in~\cite{DBLP:journals/jcss/AlonADDPT15}, $\bigO(\log(N))$ bits are needed, whereas using the ones by Awerbuch et el.~\cite{DBLP:journals/tdsc/AwerbuchKMPV07,D2K} requires $\bigO(1)$ bits.
	
	We now calculate the size of each $msg_x$, for each line $x$ mentioned above, following the analysis of the current section. Recall from Section~\ref{sec:SwitchesAndRules} that the size of a single rule is in $\bigO(\log N_C + \log N_S + \log n_{prt} + \log c_{tag})$ bits, where $n_{prt} \geq \Delta + 1$ suffices for expressing all rules. A command in $msg_{\ref{ln:msgM}}$, $msg_{\ref{ln:delAllRules}}$, and $msg_{\ref{ln:msgL}}$ has size in $\bigO(\log N_C + \log N_S)$, $\bigO(\log N_C + \log N_S)$, and respectively, in $\bigO((N_C + N_S -1)n_{prt}(\log N_C + \log N_S + \log n_{prt} + \log c_{tag}))$ bits. During recovery the following hold for the product of cardinality with command size for each set: $|msg_{\ref{ln:msgM}}| \in \bigO(maxManagers\cdot (\log N_C+ \log N_S))$,  $|msg_{\ref{ln:delAllRules}}| \in \bigO(maxRules\cdot (\log N_C + \log N_S))$, $|msg_{\ref{ln:msgL}}| \in \bigO((N_C + N_S -1)n_{prt}(\log N_C + \log N_S + \log n_{prt} + \log c_{tag})))$. Similarly, during a legal execution the following hold: $|msg_{\ref{ln:msgM}}| \in \bigO(\log N_C + \log N_S))$,  $|msg_{\ref{ln:delAllRules}}| = 0$  $|msg_{\ref{ln:msgL}}| \in \bigO((N_C + N_S -1)n_{prt}(\log N_C + \log N_S + \log n_{prt} + \log c_{tag})))$. Summing up, during recovery $|msg| \in \bigO((maxRules+maxManagers)(\log N_C+\log N_S) + (N_C + N_S -1)n_{prt}(\log N_C + \log N_S + \log n_{prt} + \log c_{tag})))$ and during a legal execution $|msg| \in \bigO((\log N_C+\log N_S) + (N_C + N_S -1)n_{prt}(\log N_C + \log N_S + \log n_{prt} + \log c_{tag})))$.
	
	We now turn to calculate the message size for a query response. Since the query response of a switch has a larger size than the one of a controller (by definition), we present only the case of switches. During recovery, a switch query response has size in $\bigO(\log N_S + \Delta(\log N_S + \log N_C) +  maxManagers\log N_C  + maxRules (\log N_C + \log N_S + \log n_{prt} + \log c_{tag}))$ bits, while a legal execution the response size is in $\bigO(\log N_S + \Delta(\log N_S + \log N_C) +  N_C\log N_C + (N_C + N_S -1)n_{prt}(\log N_C + \log N_S + \log n_{prt} + \log c_{tag}))$ bits, where $\Delta$ is the maximum degree.
\end{proof}

The proof of Lemma~\ref{thm:message} reveals that the proposed solution is communication adaptive~\cite{DBLP:journals/tpds/DolevS03}, because after stabilization the messages size is reduced.

\Subsection{Bounding the Number of Illegitimate Deletions}
\label{sec:illegiDel}

We consider another kind of event that might delay recovery (Definition~\ref{def:illdel}) and prove that it can occur a bounded number of times. 
Recall that $\Delta_{comm}$ is the number of frames in which the end-to-end protocol stabilizes (Section~\ref{sec:transportModel}) and $\Delta_{synch}$ the number of frames in which the round synchronization mechanism stabilizes (Section~\ref{s:refiningModel}).

\begin{definition}[Illegitimate deletions]
	\label{def:illDel}
	A switch $p_j$ performs an illegitimate deletion when it removes a non-failing controller $p_\ell \in P_C$ from its manager set (or its rules), due to a command that it received from another controller $p_k\in P_C$.
	\label{def:illdel}
\end{definition}

\begin{theorem}[Bounded number of illegitimate deletions]
	\label{thm:boundedResetsIllegitimateDeletions}
	Let $a_{x_k} \in R$ be the $k$-th step in which controller 
	$p_i \in P_C$ executes lines~\ref{ln:ifAllNetDiscoveredStart}--\ref{ln:ifAllNetDiscoveredEnd} during execution $R$.
	Suppose that $R$ includes at least 
	$((\Delta_{comm}+\Delta_{synch})D+1)$ such 
	$a_{x_k}$ steps, where $D$ is the network diameter. 
	Let $R'$ be a prefix of $R=R'\circ R''$ that includes the 
	steps $a_1, \ldots, a_{x_{(\Delta_{comm}+\Delta_{synch})D+1}} 
	\in R'$ and $R''$ be the matching suffix. 
	Controller $p_i$ does not take steps $a_{s'_k} \in R''$ that send a message $m_k$ to $p_j \in P_S$, such that $p_j$ performs an illegitimate deletion (Definition~\ref{def:illDel}) upon receiving $m_k$.
\end{theorem}

\begin{proof}
	This proof uses Claim~\ref{thm:whenDone} and 
	Lemma~\ref{thm:whenAfterDone}. 
	Theorem~\ref{thm:boundedResetsIllegitimateDeletions} follows by 
	the case of $k\geq D$ for Lemma~\ref{thm:whenAfterDone} 
	and then applying Part (ii) of Claim~\ref{thm:whenDone}.
	
	\begin{claim}
		\label{thm:whenDone}
		(i) The condition in the if-statement of line~\ref{ln:ifAllNetDiscovered} 
		holds if, and only if, $V_{reported} = V_{reporting}$, where 
		$V_{reported} =\{p_k : \exists_{\langle j, N_c(j), \bullet, rls 
			\rangle \in replyDB_i}\, ((k=j \lor p_k \in N_c(j)) \land 
		\exists \langle i, j_k, \bullet, currTag_i \rangle \in rls) \}\cup 
		\{\langle i, N_c(i), \emptyset, \emptyset \rangle\}$ and
		$V_{reporting} =\{p_j : \langle j, \bullet, rls \rangle \in 
		replyDB_i \land (\exists \langle i, j_k, \bullet, 
		currTag_i \rangle \in rls)\}$. 
		(ii)~Suppose that every 
		node $p_j$ in $G_c$ has sent a response $\langle j, \bullet \rangle$ 
		to $p_i$. Suppose that $p_i$ stores these replies in 
		$replyDB_i$ together with $p_i$'s report about its directly 
		connected neighborhood, $\langle i, N_c(i), \emptyset, 
		\emptyset \rangle$, cf. lines~\ref{ln:resX} and~\ref{ln:stale}.
		In this case, the condition in the if-statement of 
		line~\ref{ln:ifAllNetDiscovered} holds.  
	\end{claim}

	\begin{proof}[Proof of Claim~\ref{thm:whenDone}]
		
		\noindent \textbf{The proof of Part (i).~~}
		The condition in the if-statement of line~\ref{ln:ifAllNetDiscovered} is
		$(\forall p_{\ell}$$: p_i \rightarrow_{G(res_i(currTag_i))} p_{\ell}$$\implies \langle \ell, \bullet \rangle \in res_i(currTag_i)$. 
		When $V_{reported} = V_{reporting}$ holds, the following two claims also hold by the definition of these sets (and vice versa):
		(a) $p_i$'s response is in $replyDB_i$, and 
		(b) for every node $p_j$ that was queried with tag $currTag_i$, such that before the query either $p_j$ had a response in $replyDB_i$ or a direct neighbor of $p_j$ had a response in $replyDB_i$, there exists a response from $p_j$ in $replyDB_i$ with rules that have the tag $currTag_i$.
		Hence, the condition in the if-statement of line~\ref{ln:ifAllNetDiscovered} is true.
		
		\noindent \textbf{The proof of Part (ii).~~}
		This is just a particular case in which $P=V_{reported} = V_{reporting}$.
	\end{proof}
	
	\begin{lemma}
		\label{thm:whenAfterDone}
		Let $p_{j_k} \in P$ be a node that is at distance $k$ from 
		$p_i$ in $G_c$, such that $p_{j_0}, p_{j_1}, \ldots, p_{j_k}$
		is any shortest path from $p_i$ to $p_{j_k}$ and $p_{j_0}=p_i$. 
		Let $c_{x_y} \in R$ be the system state that immediately 
		follows step $a_{x_y} \in \{a_{x_1}, \ldots$, $a_{x_{k\cdot(\Delta_{comm}+\Delta_{synch})+1}}\} \subset R'$. 
		\begin{enumerate}
			\item Let $\ell > k\cdot\Delta_{comm}+1$.  The system state 
			$c_{x_\ell}$ is legal with respect to the end-to-end protocol of the 
			channel between $p_i$ and $p_{j_k}$, and it holds that 
			$m=\langle j_k, \bullet \rangle$ is a message arriving from 
			$p_{j_k}$ through the channel to $p_i$, which is an acknowledgment 
			for $p_i$'s message to $p_{j_k}$. 
			
			\item Let $\ell > k\cdot(\Delta_{comm}+\Delta_{synch})+1$. 
			The system state $c_{x_\ell}$ is legal with respect to the round 
			synchronization protocol between $p_i$ and $p_{j_k}$. That is, 
			for any message $m=\langle j_k, \bullet, rls \rangle$ that arrives 
			from the channel from $p_{j_k}$ to $p_i$, it holds that $m \in 
			replyDB_i \land \exists_{r\in rls}\, r=\langle i, j_k, \bullet, 
			currTag_i \rangle $. Moreover, message $m$ is an 
			acknowledgement of a message $m'$ that $p_i$ has sent to 
			$p_{j_k}$ and together $m'$ and $m$ form a completed round-trip. 
		\end{enumerate}
	\end{lemma}
	\begin{proof}[Proof of Lemma~\ref{thm:whenAfterDone}]
		We note that the first step, $a_{x_1}$ could occur due to the 
		fact that the system starts in an arbitrary state in which the condition 
		of the if-statement of line~\ref{ln:ifAllNetDiscovered} holds, hence the addition of 1 in $k\cdot(\Delta_{comm}+
		\Delta_{synch})$. 
		The proof is by induction on $k > 0$. That is, we consider the steps in $a_{x_y} \in \{a_{x_1}, \ldots, a_{x_{k\cdot(\Delta_{comm}+	\Delta_{synch})+1}}\}$.
		
		\noindent \textbf{The base case of $k=1$.~~}
		Claim~\ref{thm:whenDone} says that the condition in the if-statement of line~\ref{ln:ifAllNetDiscovered} holds if, and only if, $V_{reported} = V_{reporting}$, where $\{\langle i, N_c(i)$, $\emptyset, \emptyset \rangle\} \subseteq V_{reported}$ (line~\ref{ln:resX}).  
		Therefore, for any $\ell>1$, we have that $a_{x_\ell} \in \{a_{x_2}, \ldots, a_{x_{k\cdot(\Delta_{comm}+\Delta_{synch}+1)+1}}\}$ implies that $\{\langle i, N_c(i), \emptyset, \emptyset \rangle\} \subseteq V_{reporting}$ holds immediately before $a_{x_\ell}$. 

		\begin{claim}
			\label{thm:whenAfterClear}
			Between $a_{x_{k-1}}$ and $a_{x_{k}}$, a message $\langle j_k, \bullet, rls \rangle: \exists_{r \in rls}\, r = \langle i, j_k, \bullet$, $currTag_i \rangle$ arrives from the channel from $p_{j_k} \in N_c(i)$ to $p_i$, which $p_i$ stores in $replyDB_i$, where $k\geq 1$.
		\end{claim}
		\begin{proof}[Proof of Claim~\ref{thm:whenAfterClear}]
			During the step $a_{x_{k-1}}$, controller $p_i$ removes any response $\langle j_k, \bullet, rls \rangle: \exists_{r \in rls}\, r = \langle i, j_k, \bullet, currTag_i \rangle$ (line~\ref{ln:ifAllNetDiscoveredEnd}) and the only way in which $\langle j_k, \bullet, rls \rangle: \exists_{r \in rls}\, r = \langle i, j_k, \bullet, currTag_i \rangle$ holds immediately before $a_{x_{k}}$ is the following. 
			Between $a_{x_{k-1}}$ and $a_{x_{k}}$, a message arrives through the channel from $p_{j_k} \in N_c(j_{k-1}):j_0=i$ to $p_i$, which $p_i$ stores in $replyDB_i$ (line~\ref{ln:rcv}). 
			This is true because no other line in the code that accesses $replyDB_i$ adds that message to $replyDB_i$ (cf. lines~\ref{ln:stale},~\ref{ln:ifAllNetDiscoveredEnd}, and~\ref{ln:rcv}).
		\end{proof}
		
		\noindent \textit{The proof of Part (1).~~}
		It can be the case the $p_i$ sends a message for which it receives a (false) acknowledgement from $p_{j_1}$, i.e., without having that message go through a complete round-trip. However, by $\Delta_{comm}$'s definition (Section~\ref{sec:transportModel}), that can occur at most $\Delta_{comm}$ times. 
		
		\noindent \textit{The proof of Part (2).~~}
		It can be the case that $p_i$ receives message $m$ from $p_{j_1}$ for which the following condition does not hold in $c_{j_1}$: $m =\langle \bullet, rls \rangle \in replyDB_i \land \exists_{r \in rls}\, r=\langle i, j_k, \bullet, currTag_i \rangle$. However, by $\Delta_{synch}$'s definition (Section~\ref{sec:kappaFlowsAboveZeroConstructing}), that can occur at most $\Delta_{synch}$ times. The rest of the proof is implied by the properties of the round synchronization algorithm (Section~\ref{sec:kappaFlowsAboveZeroConstructing}).
		
		\noindent \textbf{The induction step.~~}
		Suppose that, within more than $(\Delta_{comm} k+1)$ and $((\Delta_{comm}+\Delta_{synch}) k+1)$ synchronization rounds from $R$'s starting state, the system reaches a state in which conditions (1), and respectively, (2) hold with respect to some $k\geq 1$. We show that in $c_{x_{\Delta_{comm}(k+1)+1}}$ and  $c_{x_{(\Delta_{comm}+\Delta_{synch})(k+1)+1}}$, conditions (1), and respectively, (2) hold with respect to $k+1$. 
		
		\noindent \textit{The proof of Part (1).~~}
		Claim~\ref{thm:whenDone} says that the condition in the if-statement of line~\ref{ln:ifAllNetDiscovered} holds if, and only if, $V_{reported} = V_{reporting}$. 
		By the induction hypothesis, condition (2) holds with respect to $k$ in $c_{x_{(\Delta_{comm}+\Delta_{synch})k+1}}$ and therefore $A(k+1) \cup \{\langle i, N_c(i), \emptyset, \emptyset \rangle\} \subseteq V_{reported}$, where $A(k)=\{ \langle j_{k'}, N_c(j_{k'}), \bullet, rls \rangle : 1 < k' \leq k \land \exists_{r \in rls}\, r = \langle i, j_{k'}, \bullet, currTag_i \rangle\}$. 
		Therefore, that fact that the step $a_{x_{(\Delta_{comm}+\Delta_{synch})(k+1)+2}} \in a_{x_2} \ldots a_{x_{k\cdot(\Delta_{comm}+\Delta_{synch}+1)+1}}$ implies that $A(k+1) \cup \{\langle i, N_c(i), \emptyset, \emptyset \rangle\} \subseteq V_{reporting}$ holds in the system state that appears in $R$ immediately before the step $a_{x_{(\Delta_{comm}+\Delta_{synch})(k+1)+2}}$. 
		Claim~\ref{thm:whenAfterClear} implies the rest of the proof.   
		
		\noindent \textit{The proof of Part (2).~~}
		The proof here follows by similar arguments to the ones that appear in the proof of item (2) of the base case.
\end{proof}\end{proof}

Part (iii) of Lemma~\ref{thm:boundedControllerMemory} and Theorem~\ref{thm:boundedResetsIllegitimateDeletions} imply Corollary~\ref{thm:boundDelReset}. 

\begin{corollary}
	\label{thm:boundDelReset}
	Any execution $R$ of Algorithm~\ref{alg:selfStabCode} includes no more than $N_C$ C-resets (Lemma~\ref{thm:boundedControllerMemory}) and $((\Delta_{comm}+\Delta_{synch}) D+1)\cdot N_S$ illegitimate deletions (Theorem~\ref{thm:boundedResetsIllegitimateDeletions}).
\end{corollary} 

\Subsection{Recovery from transient faults}
\label{sec:recover}
In this section we prove that Algorithm~\ref{alg:selfStabCode} is self-stabilizing.
Lemma~\ref{thm:propagation} shows that (under some conditions, such as reset freedom) controller $p_i$ eventually discovers the local topology of a switch $p_{j_k}$ that is at distance $k$ from $p_i$ in the graph $G_c$. 
This means that $p_i$ has all the information that its needs for constructing (at least) a $0$-fault-resilient flow to $p_{j_k}$ and discover any switch $p_{j_{k+1}}\in N_c(p_{j_k})$ that is at distance $k+1$ from $p_i$.    
Then, Lemma~\ref{thm:noDelMnger} shows that, within a bounded number of frames, no stale information exists in the system.
Theorem~\ref{lem:staleResetIllegalDelete} combines Corollary~\ref{thm:boundDelReset} and Lemma~\ref{thm:noDelMnger} to show that, within a bounded number of frames, the system reaches a legitimate state from which only a legal execution may continue.

We start by giving some necessary definitions.
Let $G_i$ be the value of $G(referTag_i)$ (line~\ref{ln:GDisGresprevTag}) that controller $p_i \in P_C$ computes in a step $a_x\in R$. 
We say that there is a path between $p_i \in P$ and $p_j \in P$, 
when there exist $p_{j_0}, p_{j_1}, \ldots, p_{j_k} \in P$, 
such that (1) $p_{j_0}=p_i$, (2) $p_{j_k}=p_j$, (3) $p_{j_1}, \ldots, p_{j_{k-1}} \in P_S$, 
and (4) the rules installed by a controller $p_\ell\in P_C$ at the switches in $p_{j_1}, \ldots, p_{j_{k-1}}$ 
(and also $p_i$ or $p_j$ if they are also switches) forward packets from $p_i$ to $p_j$ as well as from $p_j$ to $p_i$ (when the respective links are operational).
We say that two nodes $p_i \in P$ and $p_j \in P$ \textit{can exchange packets}, when there is a path between $p_i$ and $p_j$.
Moreover, we say that the rules installed in the switches $p_s \in P_S$ \textit{facilitate $\kappa$-fault-resilient flows between $p_i$ and $p_j $}, if at the event of at most $\kappa$ link failures there exists a path between $p_i$ and $p_j$.
Let $p_x$ and $p_y$ be two nodes in $P$ and recall that we assume that every node $p_z\in P$ has a fixed ordering of its neighbors, i.e., $N_c(z) = \{p_{i_1}, \ldots, p_{i_{|N_c(z)|}}\}$. 
We define the \textit{first shortest path} between $p_x$ and $p_y$ to be the shortest path between $p_x$ and $p_y$ that includes the nodes with minimum indices according to the neighborhood orderings (among all the shortest paths between these two nodes).

\begin{lemma}
	\label{thm:propagation}
	Let $p_i \in P_C$ be a controller and $p_{j_k} \in P$ be a node in $P$ that is at distance $k$ from $p_i$ in $G_c$, such that $p_{j_0}, p_{j_1}, \ldots, p_{j_k}$ is the first shortest path from $p_i$ to $p_{j_k}$ and $p_{j_0}=p_i$ in $G_c$.
	Suppose that C-resets (Lemma~\ref{thm:boundedControllerMemory}) and illegitimate deletions (Theorem~\ref{thm:boundedResetsIllegitimateDeletions}) do not occur in $R$.
	For every $k \geq 0$, and any system state that follows the first $((\Delta_{comm}+\Delta_{synch})+2)k$ frames from the beginning of $R$, the following hold.
	\begin{enumerate}
		\item $\langle j_k, N_c(j_k), manager_i(j_k)$, $rules_i(j_k) \rangle \in res_i(prevTag_i)$, where $N_c(j_k)$, $manager_i(j_k)$, and $rules_i(j_k)$ are $p_{j_k}$'s neighborhood, managers, and respectively, rules that $p_i$ has received from $p_{j_k}$. Moreover, for the case of controller $p_{j_k} \in P_C$, it holds that $manager(j_k) = \emptyset \land rules(j_k) = \emptyset$.
		\item $p_i \in manager_{j_k}(j_k)$. 
		\item the rules in $rules_{j_0}({j_0}), rules_{j_1}({j_1}), \ldots, rules_{j_k}({j_k})$ facilitate packet exchange between $p_i$ and $p_{j_k}$ along $p_{j_0}, p_{j_1}, \ldots, p_{j_k}$ (when the respective links are operational). 
		\item The end-to-end protocol as well as the round synchronization protocol between $p_i$ and $p_{j_k}$ are in a legitimate state.
	\end{enumerate}
\end{lemma}

\begin{proof}
	The proof is by induction on $k$.
	
	\noindent \textbf{The base case.~~}
	Claims~\ref{thm:base1holds}, \ref{thm:base2n3holds}, and~\ref{thm:base4holds} imply that the lemma statement holds for $k=1$. 
	
	\begin{claim}
		\label{thm:base1holds}
		Within one frame from $R$'s beginning, the system reaches a state in which condition (1) is fulfilled with respect to $p_i$ and any node that is in $p_i$'s distance-$1$ neighbors in $G_c$.
	\end{claim}
	\begin{proof}[Proof of Claim~\ref{thm:base1holds}]
		During the first frame (with round-trips) of $R$, controller $p_i$ starts and completes at least one iteration in which it sends a query (line~\ref{ln:sendUpdToSwitch}) to every node $p_{j_1} \in P$ that is in $p_i$'s distance-$1$ neighborhood in $G_c$ (this includes both switches, as we explain in Section~\ref{sec:scq}, as well as other controllers, which respond according to line~\ref{ln:qryCon}). 
		Moreover, during that first frame, $p_{j_1}$ receives that query and replies to $p_i$ (lines~\ref{ln:queryReturn}-\ref{ln:rcv}) within one step (Section~\ref{sec:interModel}).
		Thus, the first part of condition (1) is fulfilled, because controller $p_i$ then adds (or updates) the latest (query) replies that it received from these neighbors to $replyDB_i$. The second part of condition (1) is implied by the first part of condition (1) and by line~\ref{ln:qryCon}.
	\end{proof}
	
	\begin{claim}
		\label{thm:base2n3holds}
		Within two frames from the beginning of $R$, the system reaches a state in which conditions (2) and (3) are fulfilled with respect to $p_i$ and any node that is in $p_i$'s distance-$1$ neighbors in $G_c$.
	\end{claim}
	\begin{proof}[Proof of Claim~\ref{thm:base2n3holds}]
		This proof uses Claim~\ref{thm:base1holds} to prove this claim by first showing that within one frame from the beginning an execution in which condition (1) holds, the system reaches a state in which conditions (2) and (3) are fulfilled with respect to $p_i$ and any node $p_j \in N_c(i)$. 
		This indeed implies that conditions (2) and (3) are fulfilled within two frames of $R$ for $p_i$'s direct neighbors. 
		
		Let $R^\ast$ be a suffix of $R$ such that in $R^\ast$'s stating system state, it holds that condition (1) is fulfilled with respect to $p_i$ and any node that is in $p_i$'s distance-$1$ neighbors in $G_c$. 
		During the first frame (with round-trips) of $R^\ast$, 
		controller $p_i$ starts and completes at least one iteration (with round-trips) in which it is able to include $p_i$ in $p_j$'s manager set, $manager_j(j)$ (line~\ref{ln:rm} to~\ref{ln:delAllRules}) and to install rules at $p_j \in N_c(i)$ (line~\ref{ln:msgL}). We know that this installation is possible, because  $p_i$ is a direct neighbor of $p_j \in N_c(i)$ (Section~\ref{sec:scq}). Once these rules are installed, the packet exchange between $p_i$ and $p_j \in N_c(i)$ is feasible. 
		This implies that conditions (2) and (3) are fulfilled within one frame of $R^\ast$ (and two frames of $R$) for $p_i$'s direct neighbors. 
	\end{proof}
	
	\begin{claim}
		\label{thm:base4holds}
		Within $((\Delta_{comm}+\Delta_{synch})+2)$ frames from the beginning of $R$, the system reaches a state in which condition (4) is fulfilled with respect to $p_i$ and any node that is in $p_i$'s distance-$1$ neighbors in $G_c$.
	\end{claim}
	
	\begin{proof}[Proof of Claim~\ref{thm:base4holds}]
		Since conditions (2) and (3) hold within two frames with respect to $k=1$, controller $p_i$ and $p_{j_{1}}$ can maintain an end-to-end communication channel between them because the network part between $p_i$ and $p_{j_{1}}$ includes all the needed flows. By $\Delta_{comm}$'s definition (Section~\ref{sec:transportModel}), within $\Delta_{comm}$ frames, the system reaches a legitimate state with respect to the end-to-end protocol between $p_i$ and $p_{j_{1}}$. Similarly, by $\Delta_{synch}$'s definition (Section~\ref{sec:kappaFlowsAboveZeroConstructing}), within $\Delta_{synch}$ frames, the system reaches a legitimate state with respect to the round synchronization protocol between $p_i$ and $p_{j_{1}}$. Thus, condition (4) holds within $((\Delta_{comm}+\Delta_{synch})+2)$ frames from $R$'s beginning. 
	\end{proof}
	
	\noindent \textbf{The induction step.~~}
	Suppose that, within $((\Delta_{comm}+\Delta_{synch})+2)k$ frames from $R$'s starting state, the system reaches a state $c_x \in R$ in which conditions (1), (2), (3) and (4) hold with respect to $k$. We show that within $(\Delta_{comm}+\Delta_{synch})+2$ frames from $c_x$, the system reaches a state in which the lemma's statements hold with respect to $k+1$ as well. 
	
	\noindent \textbf{Showing that, within one frame from $c_x$, processor $p_i$ knows all of its distance-$(k+1)$ neighbors.~~}
	This part of the proof starts by showing that within one frame from $c_x$, execution $R$ reaches a state, such that $p_i \rightarrow_{G_i} p_j$ holds for every distance-$(k+1)$ neighbor of $p_i$ in $G_c$. 
	The system state $c_x$ encodes (packet forwarding) rules that allow $p_i$ to exchange packets with its distance-$k$ neighbors in $G_c$ (since by the induction hypothesis, conditions (3) and (4) hold with respect to $k$ in $c_x$).
	Moreover, $p_i$ stores in $res(prevTag_i)$ replies from $p_i$'s distance-$k$ neighbors in $G_c$ (since by the induction hypothesis, condition (1) holds for $k$ in $c_x$). 
	The latter implies that $p_i$ knows, as part of $G_i$ in $c_x$, all of its distance-$(k+1)$ neighbors, $\{p_k : \exists \langle j, N_c(j), \bullet \rangle \in res_i(prevTag_i) \land (k=j \lor k \in N_c(j,prevTag_i))\}$, since every reply of a distance-$k$ neighbor, $p_{j^\ast}$, in $G_c$ (which $res_i(prevTag_i)$ stores in $c_x$) includes $p_{j^\ast}$'s neighborhood.
	
	\noindent \textbf{Condition (1) holds with respect to $k+1$ within $((\Delta_{comm}+\Delta_{synch})+2)k+1$ frames.~~}
	Using the above we show that, within one frame from $c_x$, controller $p_i \in P_C$ queries all of its distance-$(k+1)$ neighbors (line~\ref{ln:sendUpdToSwitch}), receives their replies, and stores them in $replyDB_i$ (lines~\ref{ln:queryReturn}--\ref{ln:rcv}), i.e., $\langle j_{k+1}$, $N_c(j_{k+1}), manager_i(j_{k+1})$, $rules_i(j_{k+1}) \rangle \in res_i(currTag_i)$ for every distance-$(k+1)$ neighbor $p_{j_{k+1}}$ of $p_i$ in $G_i$. 
	Recall that $c_x$ encodes rules that let $p_i$ to forward packets with its distance-$k$ neighbors in $G_c$ (condition (3) holds for $k$ in $c_x$). 
	By the query-by-neighbor functionality (Section~\ref{sec:scq}), every such distance-$k$ neighbor reports on its direct neighbors (that include $p_i$'s distance-$(k+1)$ neighbors), which implies that it forwards the query message to $p_i$'s distance-$(k+1)$ neighbor as well as the reply back to $p_i$.
	Therefore, within ${((\Delta_{comm}+\Delta_{synch})+2)k+1}$ frames, the system reaches a state, $c_{x'}$, in which condition (1) holds with respect to $k+1$. 
	
	\noindent \textbf{Conditions (2) to (3) hold with respect to $k+1$ within $((\Delta_{comm}+\Delta_{synch})+2)k+2$ frames.~~}
	The next step of the proof is to show that within one frame from $c_{x'}$, the system reaches the state $c_{x''}$ in which conditions (2) and (3) hold with respect to $k+1$ (in addition to the fact that condition (1) holds). 
	By the functionality for querying (and modifying)-by-neighbor (Section~\ref{sec:scq}) and for every switch $p_j$ that is a distance-$(k+1)$ neighbor of $p_i$ in $G_c$, it holds that between $c_{x'}$ and $c_{x''}$: 
	(a) $p_i$ adds itself to the manager set $manager(j)$ of $p_j$ (line~\ref{ln:rm} to~\ref{ln:delAllRules}), and 
	(b) $p_i$ installs its rules in $p_j$'s configuration (line~\ref{ln:msgL}). (We note that for the case $p_j$ is another controller, there is no need to show that conditions (2) and (3) hold.)
	
	\noindent \textbf{Condition (4) holds for $k+1$ within $((\Delta_{comm}+\Delta_{synch})+2)(k+1)$ frames.~~} 
	The proof is by similar arguments to the ones that appear in the proof of Claim~\ref{thm:base4holds}.
	
	Thus, conditions (1), (2), (3), and (4) hold for $k+1$ within $((\Delta_{comm}+\Delta_{synch})+2)(k+1)$ frames in $R$ and the proof is complete.
\end{proof} 

Lemma~\ref{thm:noDelMnger} bounds the number of frames before the system reaches a legitimate system state.

\begin{lemma}
	\label{thm:noDelMnger}
	Let $R=R'\circ R''$ be an execution of Algorithm~\ref{alg:selfStabCode} that includes a prefix, $R'$, of $(\Delta_{comm}+\Delta_{synch})+2)D+1$ frames that has no occurrence of C-resets or illegitimate deletions. 
	(1) Any system state in $R''$ is legitimate (Definition~\ref{def:legitimateState}). 
	(2) Let $a_x \in R''$ be a step that includes the execution of the do-forever loop that starts in line~\ref{ln:stale} and ends in line~\ref{ln:sendUpdToSwitch}. 
	During that step $a_x$, the value of $msg_i$, which $p_i$ sends to $p_j \in P$ in line~\ref{ln:sendUpdToSwitch}, does not include the record $\langle \text{`}delMngr\text{'}, \bullet \rangle$ nor the record $\langle \text{`}delAllRules\text{'}, \bullet \rangle$, i.e., no deletions, whether they are illegitimate or not, of managers or rules. 
	(3) No controller $p_i$ takes a step in $R''$ during which the condition of line~\ref{ln:checkResponsesSize} holds, which implies that $p_i$ performs no C-reset during $R''$. 
\end{lemma}

\begin{proof}
	When comparing the conditions of Definition~\ref{def:legitimateState} and the conditions of Lemma~\ref{thm:propagation}, we see that Lemma~\ref{thm:propagation} guarantees that within $(\Delta_{comm}+\Delta_{synch})+2)D$ frames the system reaches a state $c_{almostSafe} \in R'$ in which all the conditions of Definition~\ref{def:legitimateState} hold except condition~\ref{def:legitimateState:manage} with respect to controllers $p_j \notin P_C$ that do not exist in the system (and their rules that are stored by the switches). 
	From condition~\ref{def:legitimateState:dis} of Definition~\ref{def:legitimateState}, we have that at each controller $p_i \in P_C$, it holds that $G(res(currTag_i))=G(\dis_i)=G_c$. 
	This implies that $p_i$ can identify correctly any stale information related to $p_j$ and remove it from configuration of every switch (see line~\ref{ln:forEachSwitch} to~\ref{ln:msgL}) that is in the system during the round that follows $c_{almostSafe}$, which takes one frame because condition~\ref{def:legitimateState:dis} of Definition~\ref{def:legitimateState} holds. 
	This means that within $(\Delta_{comm}+\Delta_{synch})+2)D+1$ frames the system reaches a legitimate state in which all the conditions of Definition~\ref{def:legitimateState} hold and thus $R''$ is a legal execution, i.e., the first part of the lemma holds. 
	Part (2) of this lemma is implied by the fact that there is no controller $p_j \notin P_C$ that the controller $p_i \in P_C$ needs to remove from the configuration of any switch during the legal execution $R''$. 
	Part (3) is implied by Part (3) of Lemma~\ref{thm:boundedControllerMemory} and the fact that $R''$ is a legal execution.
\end{proof}

\begin{theorem}[Self-Stabilization]
	\label{lem:staleResetIllegalDelete}
	Within $((\Delta_{comm}+\Delta_{synch})+2)D+1)[((\Delta_{comm}$ $+\Delta_{synch})D+1)\cdot N_S + N_C + 1]$ frames in $R$, the system reaches a state $c_{safe} \in R$ that is legitimate (Definition~\ref{def:legitimateState}). 
	Moreover, no execution that starts from $c_{safe} \in R$ includes a C-reset nor illegitimate deletion of managers or rules.
\end{theorem}

\begin{proof}
	In this proof, we say that an execution $R_{adm}$ is admissible when it includes at least $((\Delta_{comm}+\Delta_{synch})+2)D+1$ frames and no C-reset nor an illegitimate deletion. 
	Let $R$ be an execution of Algorithm~\ref{alg:selfStabCode}.
	Let us consider $R$'s longest possible prefix $R'$, such that $R'$ does not include any sub-execution that is admissible, i.e., $R=R'\circ R''$.
	Recall that by Corollary~\ref{thm:boundDelReset} the prefix $R'$ has no more than $((\Delta_{comm}+\Delta_{synch})D+1)\cdot N_S + N_C$ C-resets or illegitimate deletions.
	By the pigeonhole principle, the prefix $R'$ has no more than
	$((\Delta_{comm}+\Delta_{synch})+2)D+1)[((\Delta_{comm}+\Delta_{synch})D+1)\cdot N_S + N_C+1]$ frames. 
	By Lemma~\ref{thm:noDelMnger}, $R''$ does not include C-resets nor deletions of managers or rules, and the system has reached a safe state, which is $c_{safe}$. 
\end{proof}

\Subsection{Returning to a legitimate state after topology changes}
\label{sec:post}
This part of the proof considers executions in which the system starts in a state $c'$, that is obtained by taking a system state $c_{safe}$ that satisfies the requirements for a legitimate system state (Definition~\ref{def:legitimateState}), and then applying a bounded number of failures and recoveries. 
We discuss the conditions under which no packet loss occurs when starting from $c'$, which is obtained from $c_{safe}$ and  (i) the events of up to $r$ link failures and up to $\ell$ link additions (Lemma~\ref{thm:kappaLink}), as well as, (ii) the events of up to $r$ controller failures and up to $\ell$ controller additions (Lemma~\ref{lem:nodeAdditionDeletion}).

\begin{lemma}
	\label{thm:kappaLink}
	Suppose that $c'$ is obtained from a legitimate system state $c_{safe}$ by the removal of at most $r$ links and the addition of at most $\ell$ links (and no further failures), and $R$ is an execution of Algorithm~\ref{alg:selfStabCode} that starts in $c'$. 
	It holds that no packet loss occurs in $R$ as long as $r \leq \kappa$ and $\ell \geq 0$. 
	For the case of $r\leq \kappa \land \ell \geq 0$ recovery occurs within $\bigO(D)$ frames, while for the case of $r>\kappa$ bounded communication delays can no longer be guaranteed.
\end{lemma}
\begin{proof}
	We consider the following cases.\\
	\noindent \textbf{The case of $r \leq \kappa$ and $\ell=0$.~~~}
	Suppose that a single link $e$ has failed, i.e., it has been permanently removed from $G_c$, in a state $c'$ that follows a legitimate system state $c_{safe}$.
	Say that $e$ is included either in a primary path $\Pi_0$ in $G_o(0)$ or in one of the alternative paths of $\Pi_0$, $\Pi_k$ in $G_o(k)$, where $k>0$, for a controller $p_i$ (cf. definitions of the function $myRules()$ and the graphs $G_o(k)$ in Section~\ref{sec:kappaFlowsAboveZeroConstructing}).
	For every such case, since $e$'s failure occurs after a legitimate state, communication is maintained when at most $\kappa-1$ links (other than $e$) are non-operational.
	Let $s$ be the index in $\{0,1,\ldots,\kappa\}$ for which $e \in \Pi_s$.
	Due to the construction of the paths $\Pi_k$, $k\in \{0,1,\ldots,\kappa\}$, in the computation of the function $myRules()$ in $p_i$, if $s=0$, then each alternative path $\Pi_k$ before $e$'s failure is now considered as path $\Pi_{k-1}$, for $k\in \{1,\ldots,\kappa\}$. 
	Otherwise, if $s\neq 0$, the paths $\Pi_k$ remain the same for $k\in \{0,\ldots, s-1\}$ and each path $\Pi_k$ is now considered as the alternative path $\Pi_{k-1}$ for $k\in \{s+1,\ldots, \kappa\}$.
	In both cases, a new path $\Pi_{\kappa}$ is computed and installed in the switches if that is possible due to the edge-connectivity of $G_c$, and if that is not the case, the rules installed in the network's switches facilitate $(\kappa-1)$-fault-resilient flows between every controller and every other node in the network.
	The recovery time is at most 1 frame (if $e$ belongs to some path $\Pi_k$), since the removal of link $e$ occurs after a legitimate state and all nodes in the network can be reached by every controller $p_i\in P_C$.
	
	Note that if $e$ is not part of any flow, then its failure has no effect in maintaining bounded communication delays.
	By extension of the argument above, bounded communication delays can be maintained when at most $\kappa$ link failures occur.
	That is, in the worst case when exactly $\kappa$ link failures occur, bounded communication delays are maintained due to the existence of the $\kappa^\text{th}$ alternative paths and the assumption that no further failures occur in the network.
	
	\noindent \textbf{The case of $r=0 \land \ell>0$.~~~}
	A link addition can violate the first shortest path optimality, thus in this case all paths should be constructed from scratch. 
	Since, the link addition occurs after a legitimate state, no stale information exist in the system,  and no resets or illegitimate deletions occur.
	Hence, by Lemma~\ref{thm:propagation} (for $k=D$) within $2D$ frames it is possible to (re-)build the $\kappa$-fault containing flows throughout all nodes in the network and reach a legitimate system state (since the edge-connectivity cannot decrease with link additions).
	
	\noindent \textbf{The case of $r\leq \kappa$ and $\ell> 0$.~~~} 
	Note that by the first case, bounded communication delays are maintained, since $r\leq \kappa$.
	Since $\ell$ links are added in $G_c$, the controllers require $O(D)$ frames to install new paths (by Lemma~\ref{thm:propagation}), even though the connectivity of $G_c$ might be less than $\kappa+1$ (but for sure at least 1).
	Hence, bounded communication delays are guaranteed in this case, given that no more failures occur.
	
	\noindent \textbf{The case of $r > \kappa$.~~~} 
	In this case, we do not guarantee bounded communication delays.
	This holds, due to the fact that the removal of more than $\kappa$ edges might break connectivity in $G_c$, which makes the existence of alternative paths for $r>\kappa$ link failures impossible.
\end{proof}

\begin{lemma}
	\label{lem:nodeAdditionDeletion}
	Suppose that $c'$ is obtained from a legitimate system state $c_{safe}$ by the removal of at most $r$ nodes and the addition of at most $\ell$ nodes (and no further failures), and $R$ is an execution of Algorithm~\ref{alg:selfStabCode} that starts in $c'$. 
	It holds that no packet loss occurs in $R$ if, and only if, $G_c$ remains connected (and $N_C\geq 1 \land N_S\geq 1$), and in this case the network recovers within $\bigO(D)$ frames. 
	For the case of $r>0 \land \ell = 0$ bounded communication delays can no longer be guaranteed. 
\end{lemma}

\begin{proof}
	We study the following cases.
	
	\noindent \textbf{The case of $r > 0$ and $\ell = 0$.~~~} 
	The removal of a switch $p_j$ is equivalent to the removal of all the links that are adjacent to $p_j$.
	Since the edge-connectivity is at least $\kappa+1$, the minimum degree of every node in $G_c$ is at least $\kappa+1$.
	Thus, a switch removal (equiv. removal of at least $\kappa+1$ links) would violate the assumption of at most $\kappa$ link failures, possibly violating connectivity or affecting all the alternative paths between two endpoints in the network.
	In this case, Algorithm~\ref{alg:selfStabCode} can only guarantee that the controllers will install $\tilde{\kappa}$-fault-resilient flows, where $0\leq \tilde{\kappa} \leq \kappa$.
	
	The case of removing a controller $p_i$ can be handled by Algorithm~\ref{alg:selfStabCode} if we assume that the communication graph $G_c$ stays (at least) $(\kappa+1)$-edge-connected after removing $p_i$.
	In that case, each controller $p_{i'}$ can discover the removal of $p_i$ and delete it from $replyDB_{i'}$ in 1 frame, and then, in the subsequent frame, $p_{i'}$ can delete $p_i$'s rules from $rules_j(j)$ and $p_i$ from $manager_j(j)$, for every switch $p_j$.
	Hence, within 2 frames the system recovers to a legitimate state, since the existing rules of the other controllers stay intact.
	
	\noindent \textbf{The case of $r=0$ and $\ell > 0$.~~~} 
	We assume that if controller or switch additions occur (including their adjacent links) after a legitimate state, the new node is initialized with empty memory.
	That is, $replyDB_i$ is empty if a new controller $p_i$ is added, and $manager_j(j) = rules_j(j) = \emptyset$ if a new switch $p_j$ is added.
	Note that the new node should not violate the assumption of $G_c$'s edge-connectivity being at least $\kappa+1$.
	In both cases, and similarly to link additions, the first shortest path optimality might be violated and hence (as in the case of link additions) a period of $2D$ frames is needed (Lemma~\ref{thm:propagation}) to (re-)build the $\kappa$-fault-resilient flows (since no stale information exist, and no resets or illegal deletions occur).
	
	\noindent \textbf{The case of $r>0$ and $\ell > 0$.~~~} 
	Let $G_c'$ be $G_c$ after the removal of at most $r$ nodes and the addition of at most $\ell$ nodes.
	If $G_c'$ is $\tilde{\kappa}$-edge-connected, where $1< \tilde{\kappa} \leq \kappa$, then bounded communication delays in the occurrence of at most $\tilde{\kappa}$ link failures can be guaranteed by following the arguments of Section~\ref{sec:recover} for $\kappa = \tilde{\kappa}$.
\end{proof}

\section{Evaluation}
\label{sec:eval}

\begin{figure*}[t]
	\centering
	\includegraphics[width=\figSizeJournal\textwidth]{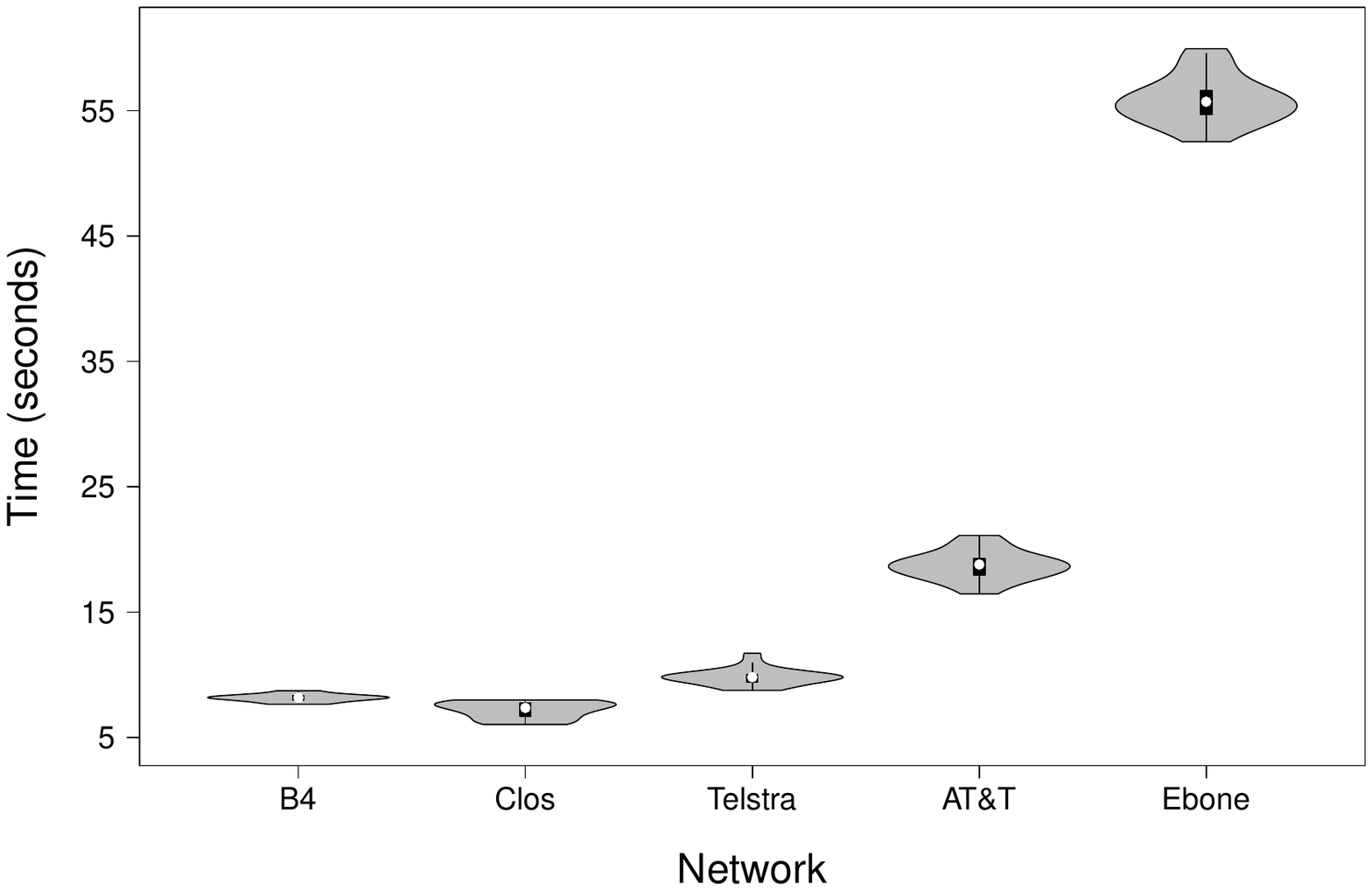}
	\caption{Bootstrap time for the networks using 3 controllers. The network  diameters are 4, 5, 8, 10 and 11 (left to right order).}
	\label{a}
\end{figure*}

In order to evaluate our approach,  and in particular, to complement our theoretical worst-case analysis as well as study the performance in different settings, we implemented a prototype using Open vSwitch (OVS) and Floodlight. To ensure reproducibility and to facilitate research on improved and alternative algorithms, the source code and evaluation data are accessible via~\cite{RenaissanceWeb}. In the following, we first explain our expectations with respect to the performance (Section~\ref{sec:limitations}) and discuss details related to the implementation of the proposed solution (Section~\ref{sec:implementation}) before presenting the setup of our experiments (Section~\ref{sec:setup}). In particular, we empirically evaluate the time to bootstrap an SDN (after the occurrence of different kinds of transient failures), the recovery time (after the occurrence of  different kinds of benign failures), as well as the throughput during a recovery period that follows a single link failure (Section~\ref{sec:results}). For the reproducibility sake, the source code and evaluation data can be access via~\cite{RenaissanceWeb}.

\subsection{Limitations and expectations}
\label{sec:limitations}
We study $\system$'s ability to recover from failures in a wide range of topologies and settings. We note that the scope of our work does not include an empirical demonstration of recovery after the occurrence of \emph{arbitrary} transient faults, because such a result would need to consider all possible starting system states. Nevertheless, we do consider recovery after changes in the topology, which Section~\ref{sec:faultModel} models as transient faults. However, in these cases, we mostly consider a single change to the topology, i.e., node or link failure (after the recovery from any other transient fault). 

The basis for our performance expectation is the analysis presented in Section~\ref{sec:proof}. Specifically, we use lemmas~\ref{thm:propagation},~\ref{thm:kappaLink} and~\ref{lem:nodeAdditionDeletion} to anticipate an $\bigO(D)$ bootstrap time and recovery period after the occurrence of benign failures. Recall that, for the sake of simple presentation, our theoretical analysis does not consider the number of messages sent and received (Section~\ref{sec:timeComplexity}), which depends on the number of nodes in the case of $\system$. Thus, we do not expect the asymptotic bounds of lemmas~\ref{thm:propagation},~\ref{thm:kappaLink} and~\ref{lem:nodeAdditionDeletion} to offer an exact prediction of the system performance since our aim in Section~\ref{sec:proof} is merely to demonstrate bounded recovery time. The measurements presented in this section show that $\system$'s performance is in the ballpark of the estimation presented in Section~\ref{sec:proof}.

\subsection{Implementation}
\label{sec:implementation}
In this evaluation section, we demonstrate $\system$'s ability to recover from failures without distinguishing between transient and permanent faults, as our model does (Figure~\ref{fig:self-stab-SDN}), because there is no definitive distinction between transient and permanent faults in real-world systems. Moreover, our implementation uses a variation on Algorithm~\ref{alg:selfStabCode}. The reason that we need this variation is that this evaluation section considers changes to the network topology during legal executions, whereas our model considers such changes as transient faults that can occur before the system starts running.

In detail, Algorithm~\ref{alg:selfStabCode} installs rules on the switches using two tags, which are $currTag$ and $prevTag$ (line~\ref{ln:tags}). That is, as the new rules for $currTag$ are being installed, the ones for $prevTag$ are being removed. Our variation uses a third tag, $\mathit{beforePrevTag}$, which tags the rules in the synchronization round that preceded the one that $prevTag$ refers to. When $\system$ installs new rules that are tagged with $currTag$, it does not remove the rules tagged with $prevTag$ but instead, it removes the rules that are tagged with $\mathit{beforePrevTag}$. This one extra round in which the switches hold on to the rules installed for $prevTag$'s synchronization round allows $\system$ to use the $\kappa$-fault-resilient flows that are associated with $prevTag$ for dealing with link failures (without having them removed, as Algorithm~\ref{alg:selfStabCode} does). The above variation allows us to observe the beneficial and complementary existence of the mechanisms for tolerating transient and permanent link failures, i.e., $\system$'s construction of $\kappa$-fault-resilient flows, and respectively, update of such flows according to changes reported by $\system$'s topology discovery.

\begin{figure*}[t]
	\centering
	\includegraphics[width=\figSizeJournal\textwidth]{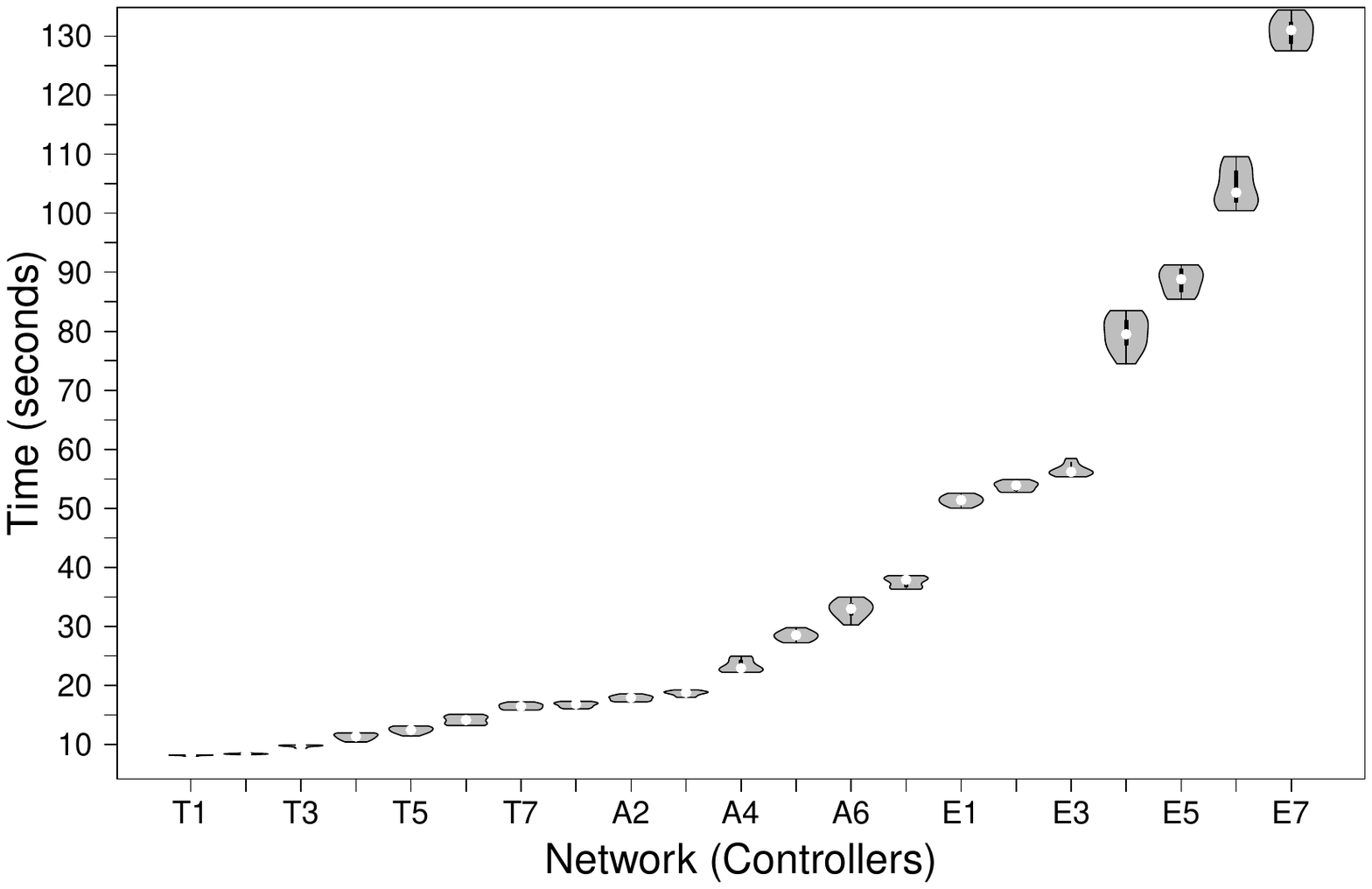}
	\caption{Bootstrap time for Telstra (T), AT\&T (A) and EBONE (E) for 1 to 7 controllers.}
	\label{b}
\end{figure*}

\begin{figure*}[t!]
	\centering
	\includegraphics[width=\figSizeJournal\textwidth]{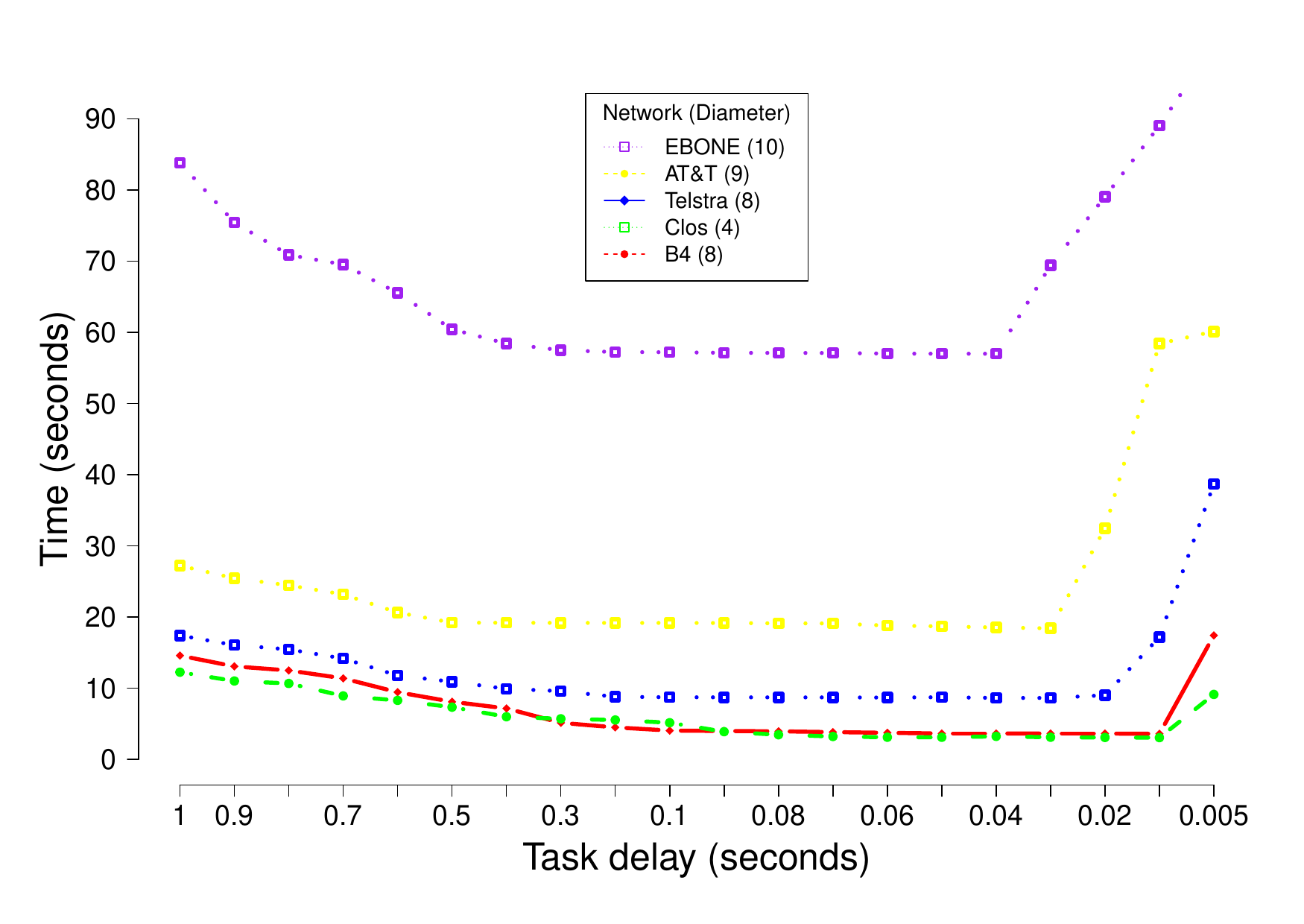}
	\caption{\label{c2}Bootstrap time for B4, Clos, Telstra, AT\&T and EBONE  using seven controllers, as a function of query intervals. Recall that the task delay in the added time between any repetition of the algorithm's do forever loop as well as each interval in which the abstract switch discovers its neighborhood.}
\end{figure*}

~\\
\subsection{Setup}
\label{sec:setup}
We consider a spectrum of different topologies (varying in size and diameter), including B4 (Google's inter-datacenter WAN based on SDN), Clos datacenter networks and Rocketfuel networks (namely Telstra, AT\&T and EBONE). The relevant statistics of these networks can be found in Table~\ref{nodeDiameterDegree}. The hosts for traffic and round-trip time (RTT) evaluation are placed such that the distance between them is as large as the network diameter. 
%
%
The evaluation was conducted on a PC running Ubuntu 16.04 LTS OS, with the Intel(R) Core(TM) i5-457OS CPU @ 2.9 GHz (4x CPU) processor and 32 GB RAM. The maximum transmission unit (MTU) for each link in the Mininet networks were set to 65536 bytes.

Paths are computed according to Breadth First Search (BFS) and we use OpenFlow fast-failover groups for backup paths. We introduce a delay before every repetition of the algorithm's do forever loop as well as between each interval in which the abstract switch discovers its neighborhood. In our experiments, the default delay value was $500$ ms. However, in an experiment related to the bootstrap time (Figure~\ref{c2}), we have varied the delay values.

\begin{wrapfigure}{r}{0.415\textwidth}
	\begin{center}
		\begin{\VCalgSize}
			\vspace*{-2.5em}
			\begin{tabular}{|l|l|l|l|}
				\hline
				\textbf{Network} & \textbf{Nodes} & \textbf{Diameter}  \\ \hline
				B4                    & 12                       & 5                                              \\ \hline
				Clos                  & 20                       & 4                                             \\ \hline
				Telstra               & 57                       & 8                                           \\ \hline
				AT\&T                  & 172                      & 10                                     \\ \hline
				EBONE                & 208                      & 11                                      \\ \hline
			\end{tabular}
			\vspace*{-1em}
		\end{\VCalgSize}
	\end{center}
	\caption{\label{nodeDiameterDegree}The number of nodes and diameter of the studied networks}
\end{wrapfigure}

The link status detector (for switches and controllers) has a parameter called $\Theta$, similar to the one used in~\cite[Section 6]{DBLP:conf/netys/BlanchardDBD14}. This threshold parameter refers to scenarios in which the abstract switch queries a non-failing neighboring node without receiving a query reply while receiving $\Theta$ replies from all other neighbors. The parameter $\Theta$ can  balance a trade-off between the certainty that node is indeed failing and the time it takes to detect a failure, which affects the recovery time. We have selected $\Theta$ to be $10$ for B4 and Clos, and $30$ for Telstra, AT\&T and Ebone. We observed that when using these settings the discovery of the entire network topology always occurred and yet had the ability to provide a rapid fault detection.

\begin{figure*}[t!]
	\centering
	\includegraphics[width=\figSizeJournal\textwidth]{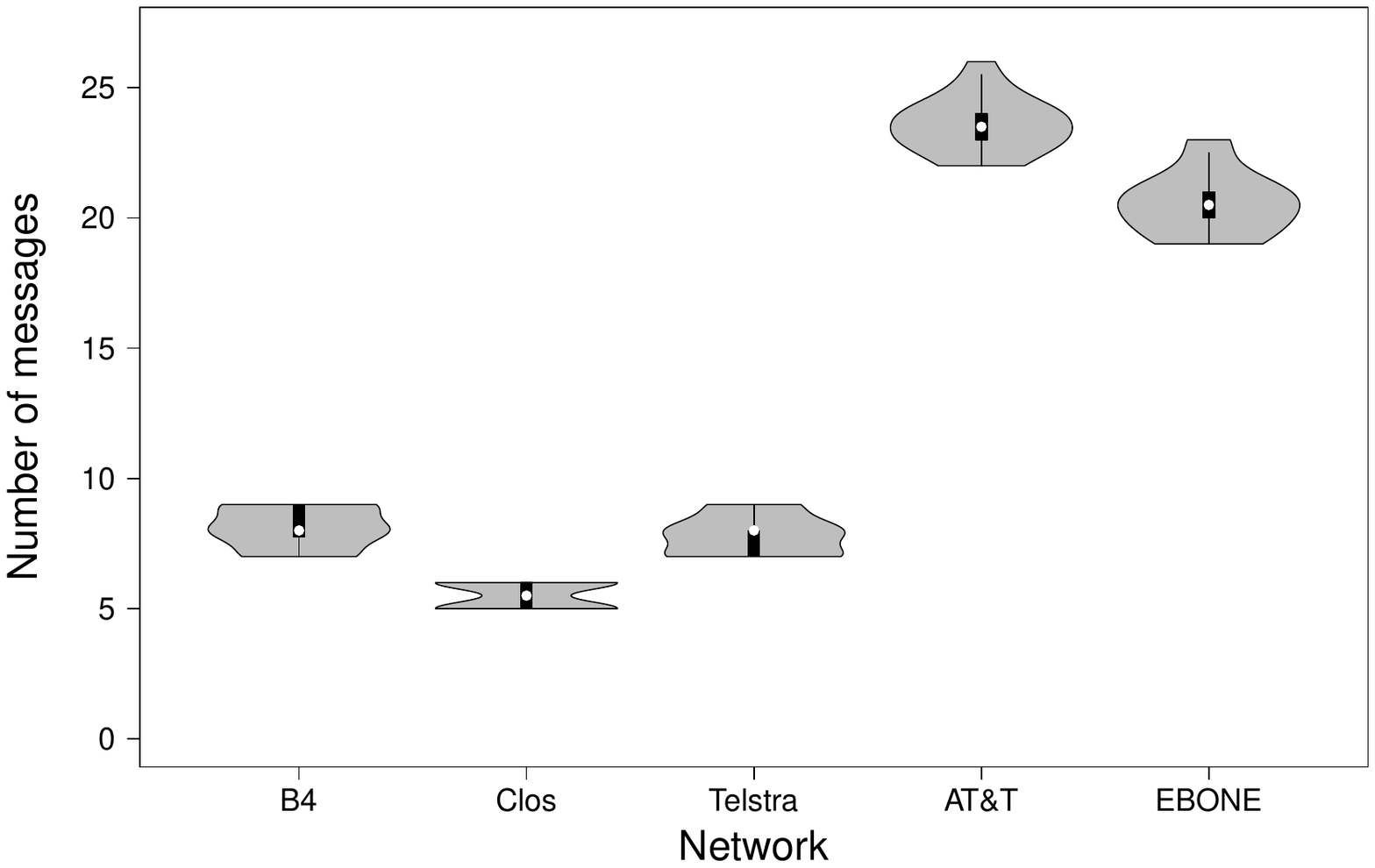}
	\caption{Communication cost per node needed from a maximum loaded global controller to reach a stable network. Note that we divide the number of messages by the number of iterations it takes to converge.}
	\label{d2}
\end{figure*}

\begin{figure*}[t]
	\centering
	\includegraphics[width=\figSizeJournal\textwidth]{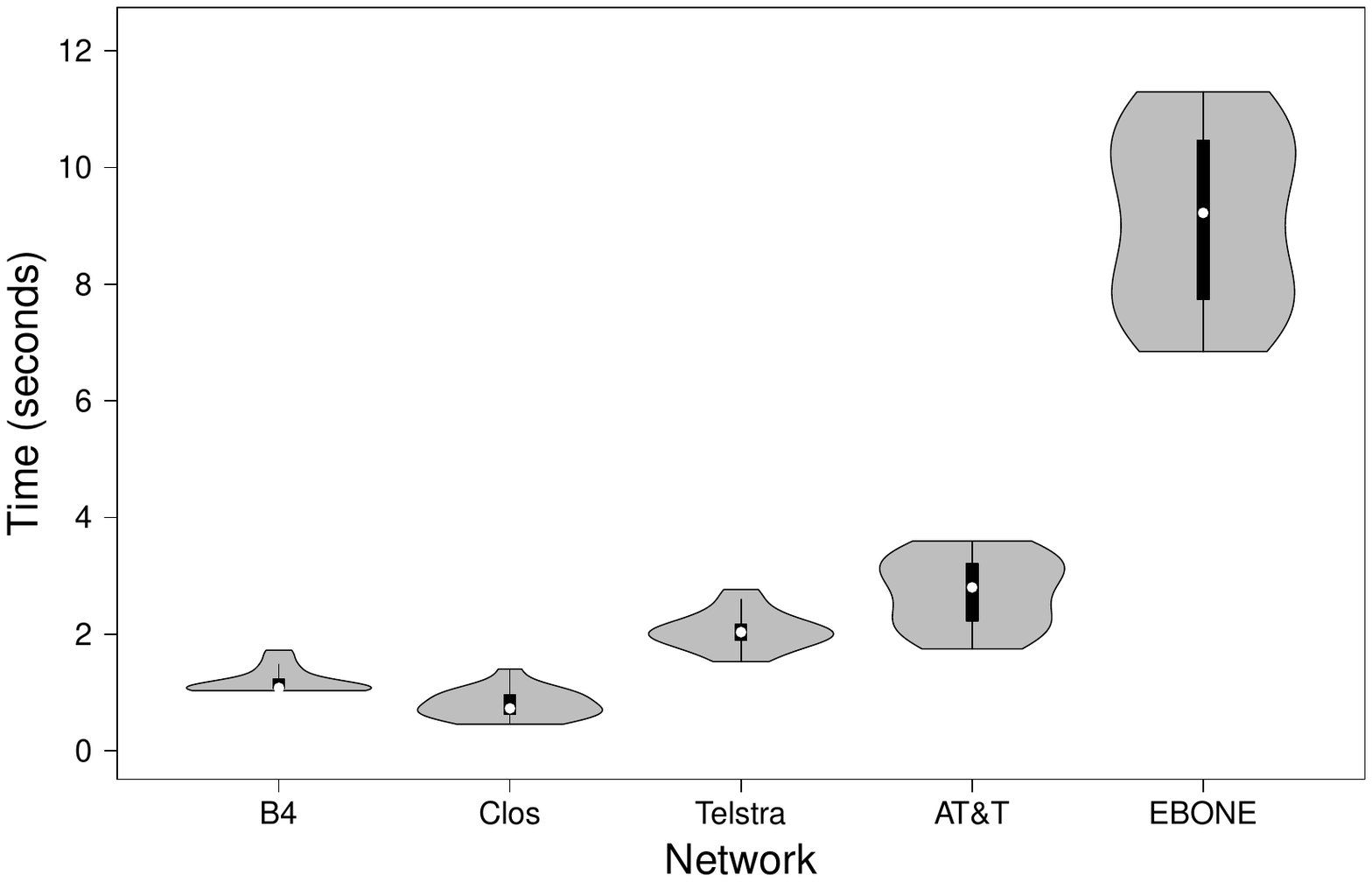}
	\caption{Recovery time after fail-stop failure for a controller.}
	\label{f}
\end{figure*}

\subsection{Results}
\label{sec:results}
We structure our evaluation of $\system$ around the main questions related to the SDN bootstrap, recovery times, and overhead, as well as regarding the throughput during failures. 

For illustrating our data in figures~\ref{a}--\ref{b} and \ref{d2}--\ref{j}, we use violin plots~\cite{hintze1998violin}. In these plots, we indicate the median with a white dot. The first and third quartiles are the endpoints of a thick black line (hence the white dot representing the median is a point on the black line). The thick black line is extended with thin black lines to denote the two extrema of all the data (as the whiskers of box plots). Finally, the vertical boundary of each surface denotes the kernel density estimation (same on both sides) and the horizontal boundary only closes the surface. We ran each experiment $20$ times. For the case of violin plots, we used all measurements except the two extrema. For the case of the other plots, we dismissed from the $20$ measurements the two extrema. Then, we calculated average values and used them in the plots.

\subsubsection{How efficiently $\system$ bootstraps an SDN?}
We first study how fast we can establish a stable network starting from empty switch configurations. Towards this end, we measure how long it takes until all controllers in the network reach a legitimate state in which each controller can communicate with any other node in the network (by installing packet-forwarding rules on the switches). For the smaller networks (B4~\cite{b4} and Clos~\cite{clos}), we use three controllers, and for the Rocketfuel networks~\cite{rocketfuel, rocketfuel2} (Telstra, AT\&T and EBONE), we use up to seven controllers. 


\begin{figure*}[t]
	\centering
	\includegraphics[width=\figSizeJournal\textwidth]{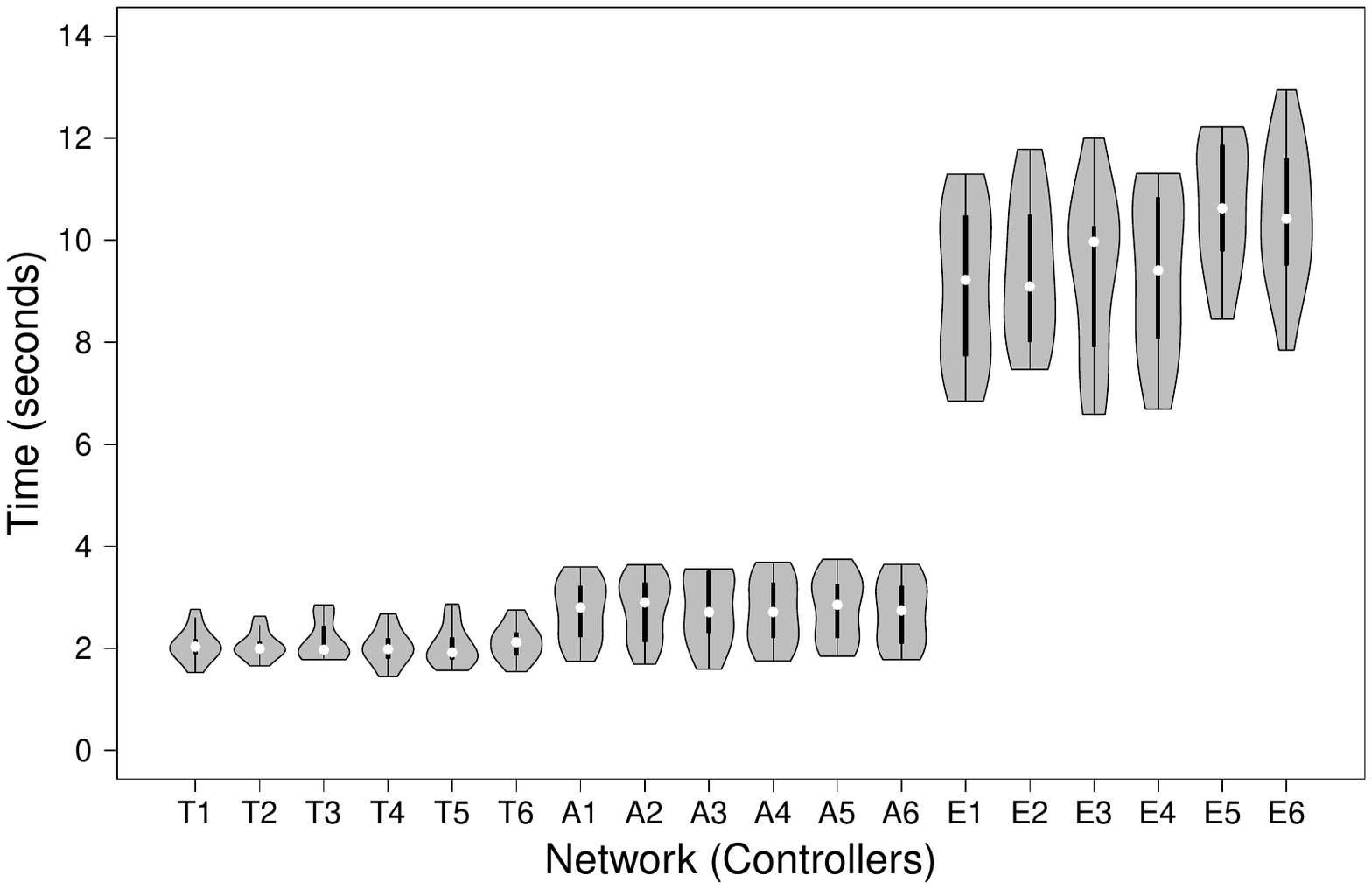}
	\caption{Recovery time after fail-stop failure of 1-6 controllers in Telstra (T), AT\&T (A) and EBONE (E).}
	\label{g}
\end{figure*}

\noindent \textbf{Bootstrapping time.~~}
We are indeed able to bootstrap in \emph{any} of the configurations studied in our experiments. Lemma~\ref{thm:propagation} predicts an $\bigO(D)$  bootstrap time when starting from an all empty switch configuration; that prediction does not consider the number of nodes, as explained above. Note that in such executions, no controller sends commands that perform (illegitimate) deletions before it discovers the entire network topology and thus no illegitimate deletion is ever performed by any controller. In terms of performance, we observe that the recovery time grows (Figure~\ref{a}) as the network dimensions increase (diameter and number of nodes). It also somewhat depends on the number of controllers when experimented with the larger networks (Figure~\ref{b}): more controllers result in slightly longer bootstrap times. We note that the recovery process over a growing number of controllers follows trades that appear when considering the maximum value over a growing number of random variables. Specifically, when an abstract switch updates its rules, the time it takes to update all of the rules that were sent by many controllers can appear as a brief bottleneck.

\begin{figure*}[t]
	\centering
	\includegraphics[width=\figSizeJournal\textwidth]{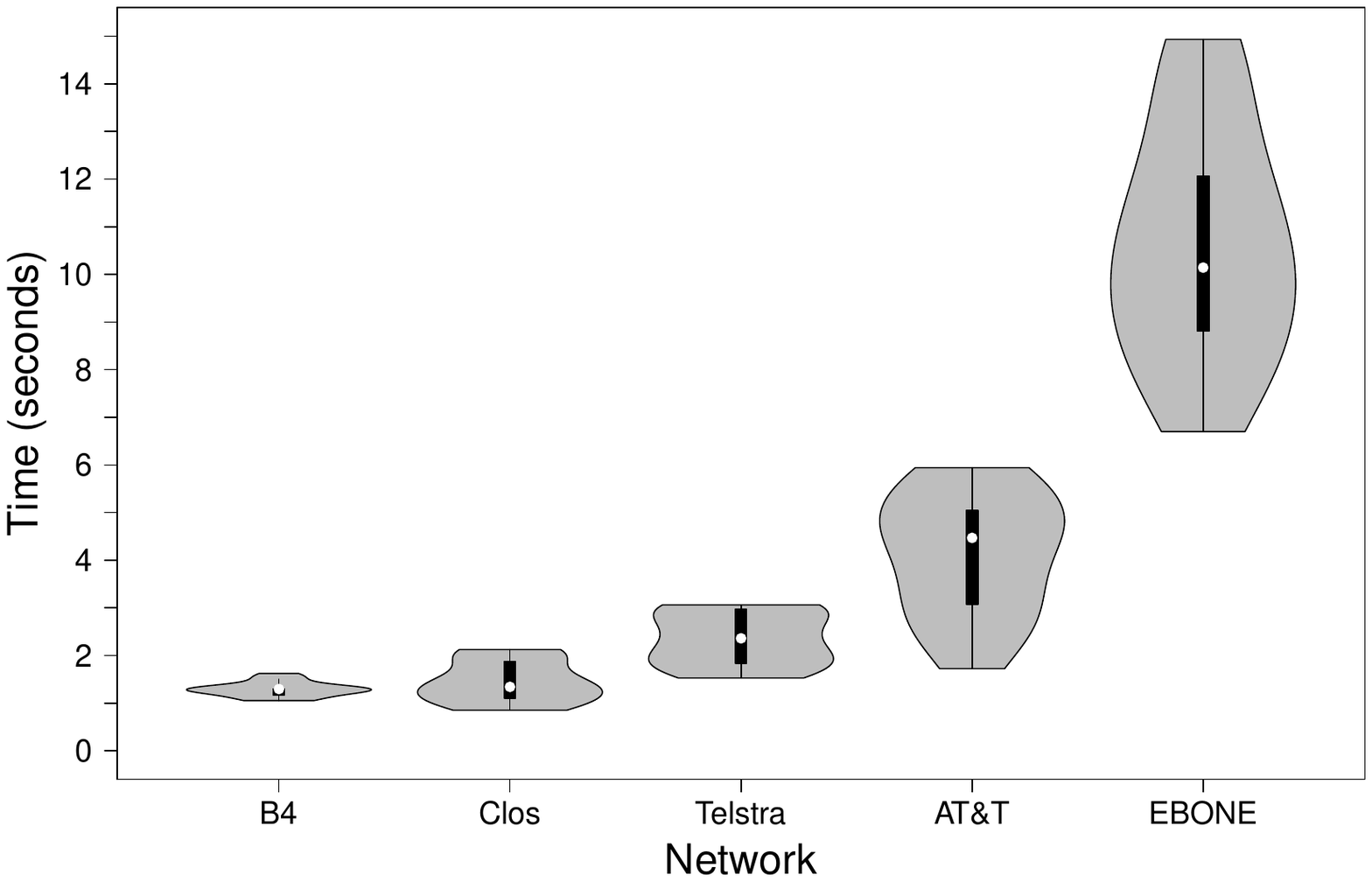}
	\caption{Recovery time after permanent switch-failure.}
	\label{h}
\end{figure*}

Note that the shown bootstrap times only provide qualitative insights: they are, up to a certain point,  proportional to the frequency at which controllers request configurations and install flows (Figure~\ref{c2}). Specifically, the rightmost peaks in the charts are due to the congestion caused by having task delays that overwhelm the networks. These peaks rise earlier for networks with an increasing number of switches. This is not a surprise because the proposed algorithm establishes more and longer flows in larger networks and thus use higher values of network traffic as the number of nodes becomes larger. 

\noindent \textbf{Communication overhead.~~}
The study of bootstrap time thus raises interesting questions regarding the \emph{communication overhead} during the network bootstrap period. Concretely, we measure the maximum number of controller messages, taking three controllers for the smaller networks B4 and Clos, and seven controllers for the Rocketfuel networks Telstra, AT\&T and EBONE in these experiments. While the communication overhead naturally depends on the network size, Figure~\ref{d2} suggests that when normalized, i.e., dividing by the number of iterations it takes to recover, the overhead is similar for different networks (and slightly higher for the case of the two largest networks).  

\subsubsection{How efficiently $\system$ recovers from link and node failures?~~}
In order to study the recovery from benign failures, we distinguish between their different types: (i) fail-stop failures of controllers, (ii) permanent switch-failures, and (iii) permanent link-failures. The experiments start from a legitimate system state, to which we inject the above failures.

\noindent \textbf{(i) Recovery after the occurrence of controller's fail-stop failure.~~}

We injected a \emph{fail-stop} failure by disconnecting a single controller chosen uniformly at random (Figure~\ref{f}). We have also conducted an experiment in which we have disconnected many-but-not-all controllers (Figure~\ref{g}). That is, we disconnected a single controller that is initially chosen at random and measured the recovery time. The procedure was repeated for the same controller while recording the measurements until only one controller was left. Lemma~\ref{lem:nodeAdditionDeletion}, which does not take into consideration the time it takes to send or receive messages, suggests that after the removal of at most $N_C -1$ controllers, the system reaches a legitimate system state within $\bigO(D)$. We observe in Figure~\ref{f} results that are in the ballpark of that prediction. Moreover, we also measure disconnecting one to six random controllers simultaneously for the Rocketfuel networks (Telstra, AT\&T, and EBONE), while running controller number 7. Note that we could not observe a relation between the number of failing controllers and the recovery time, see Figure~\ref{g}.

\begin{figure*}[t]
	\centering
	\includegraphics[width=\figSizeJournal\textwidth]{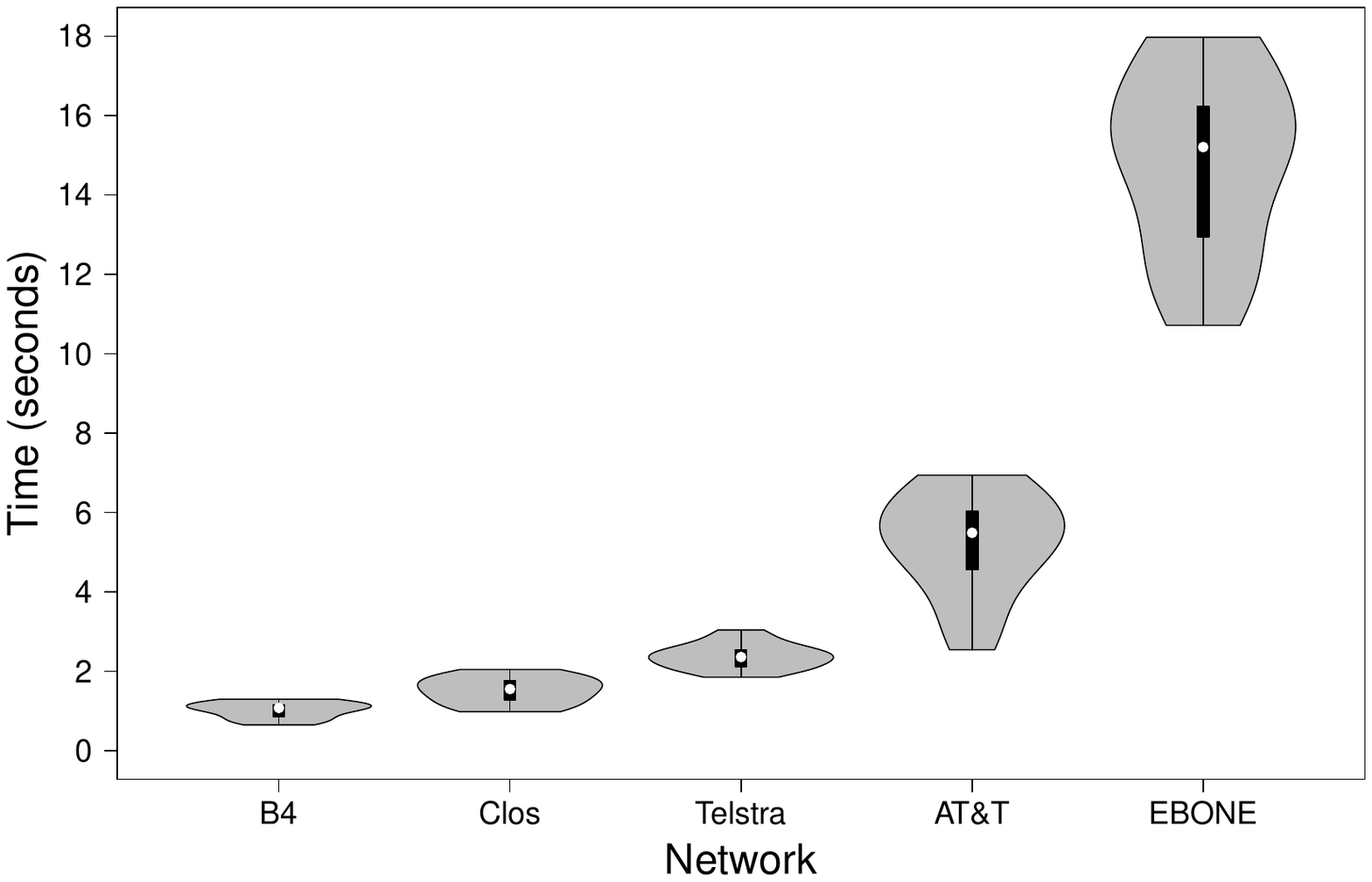}
	\caption{Recovery time after permanent link-failure.}
	\label{i}
\end{figure*}

\noindent \textbf{(ii) Recovery after the occurrence of switch's fail-stop failure.~~}
We have experimented with recovery after \emph{permanent switch-failures}. These experiments started by allowing the network to reach a legitimate (stale) state. Once in a legitimate (stale) state, a switch (selected uniformly at random) was disconnected from the network. We have then measured the time it takes the system to regain legitimacy (stability). We know that by Lemma~\ref{lem:nodeAdditionDeletion}, the recovery time here should be at most in the order of the network diameter. Figure~\ref{h} presents the measurements that are in the ballpark of that prediction. That is, the longest recovery time for each of the studied networks grows as the network diameter does. We also observe a rather large variance in the recovery time, especially for the larger networks. This is not a surprise since the selection of the disconnected switch is random.

\noindent \textbf{(iii) Recovery after the occurrence of permanent link-failures.~~}
During the experiments, we waited until the system reached a legitimate state, and then disconnected a link and waited for the system to recover. Lemma~\ref{lem:nodeAdditionDeletion} predicts recovery within $\bigO(D)$. Figure~\ref{i} presents results that are in the ballpark of that prediction. We also investigated the case of multiple and simultaneous permanent link failures that were selected randomly. Figure~\ref{j} suggests that the number of simultaneous failures does not play a significant role with respect to the recovery time.

\begin{figure*}[t]
	\centering
	\includegraphics[width=\figSizeJournal\textwidth]{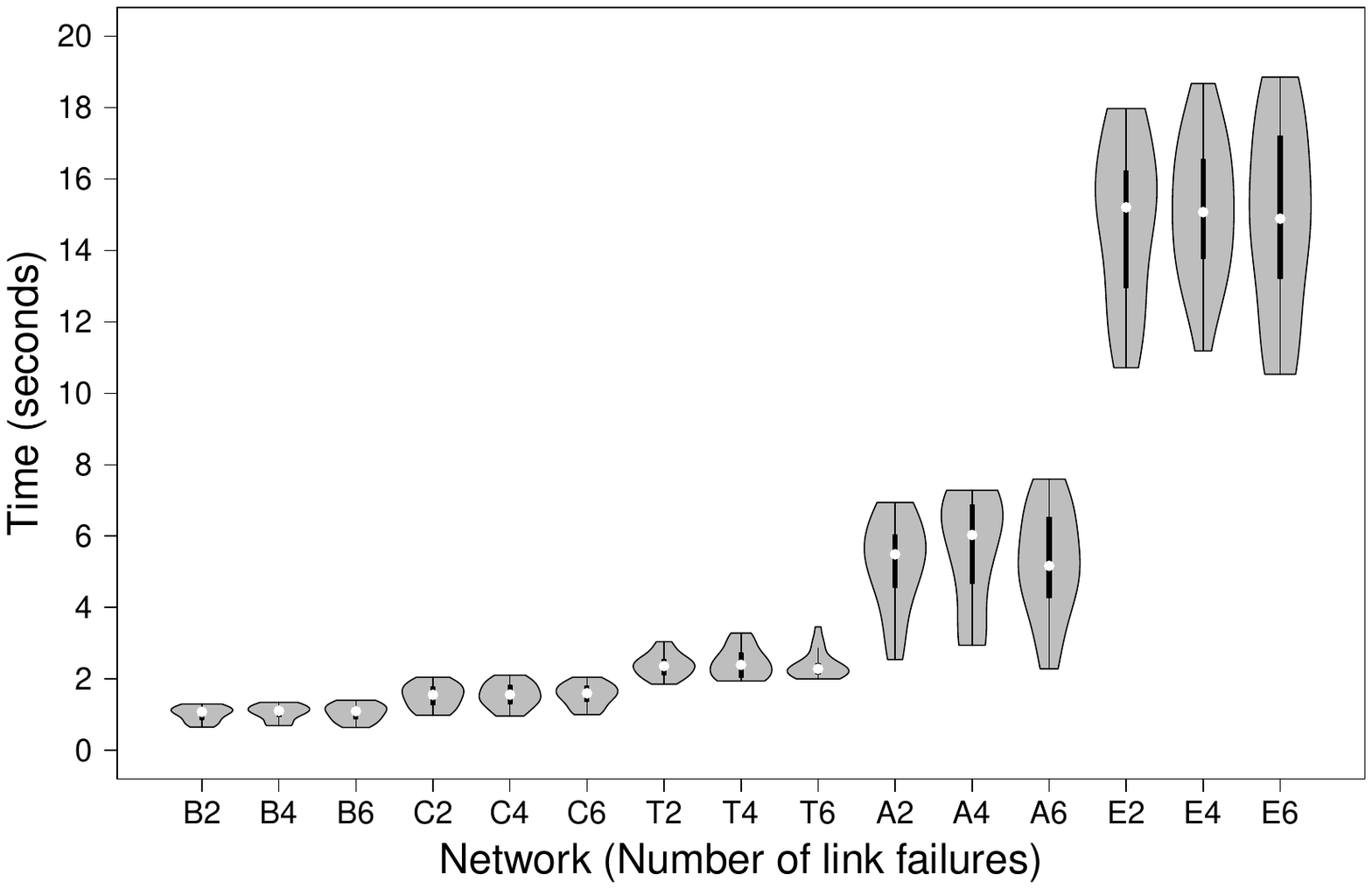}
	\caption{Recovery time after multiple (2,4 or 6) permanent link-failures at random for B4 (B), Clos (C), Telstra (T), AT\&T (A) and EBONE (E).}
	\label{j}
\end{figure*}

\subsubsection{Performance during failure recovery~~}
Besides connectivity, we are also interested in performance metrics such as throughput and message loss \emph{during recovery period that occur after a single link failure}. Recall that we model such failures as transient faults and therefore there is a need to investigate empirically the system's behavior during such recovery periods since the mechanism for fault-resilient flows (Section~\ref{sec:kappaFlowsAboveZeroConstructing}) is always active. Our experiments show that the combination between the proposed algorithm and the mechanism for fault-resilient flows performs rather well. That is, the recovery period from a single permanent link failure is brief and it has a limited impact on the throughput. 

\begin{figure*}[t!]
	\centering
	\includegraphics[width=\figSizeJournal\textwidth]{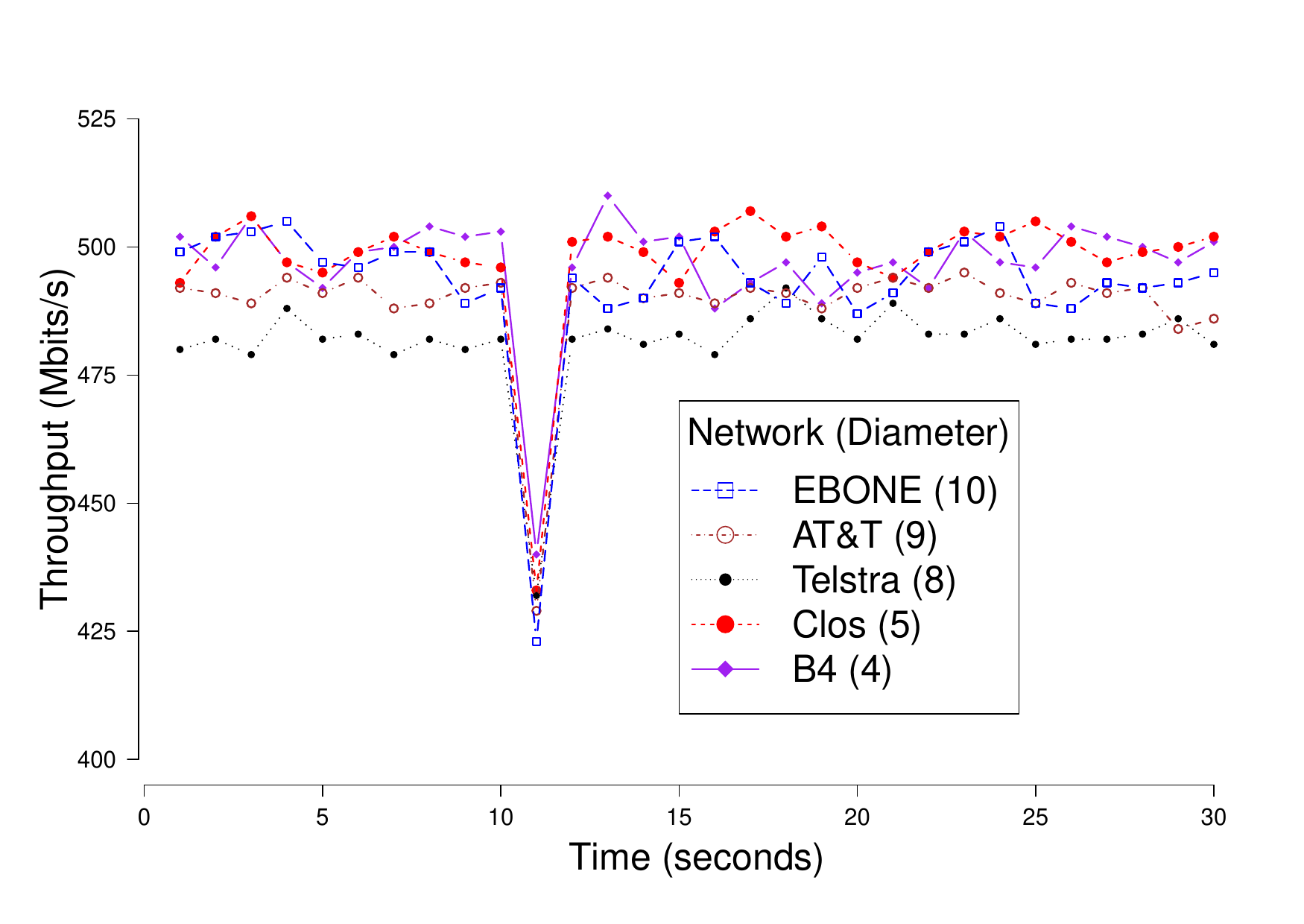}
	\caption{Throughput for the different networks using network updates with tags. Here, a single link failure causes the drop after the $10^{th}$ second.}
	\label{m2}
\end{figure*}

In the following, we measure the TCP throughput between two hosts (placed at a maximal distance from each other), in the presence of a link-failure located as close to the middle of the primary path as possible. To generate traffic, we use Iperf. A specific link to fail is chosen, such that it enables a backup path between the hosts. 

The maximum link bandwidth is set to 1000 Mbits/s. During the experiments, we conduct throughput measurements during a period of 30 seconds. The link-failure occurs after 10 seconds, and we expect a throughput drop due to the traffic being rerouted to a backup path. We note that our prototype utilizes packet tagging for consistent updates~\cite{DBLP:conf/hotnets/ReitblattFRW11} using the variation of Algorithm~\ref{alg:selfStabCode} (presented in Section~\ref{sec:implementation}), which allows the controllers to repair the $\kappa$-fault-resilient flows without the removal of the ones tagged with the previous tag.

\begin{figure*}[t!]
	\centering
	\includegraphics[width=\figSizeJournal\textwidth]{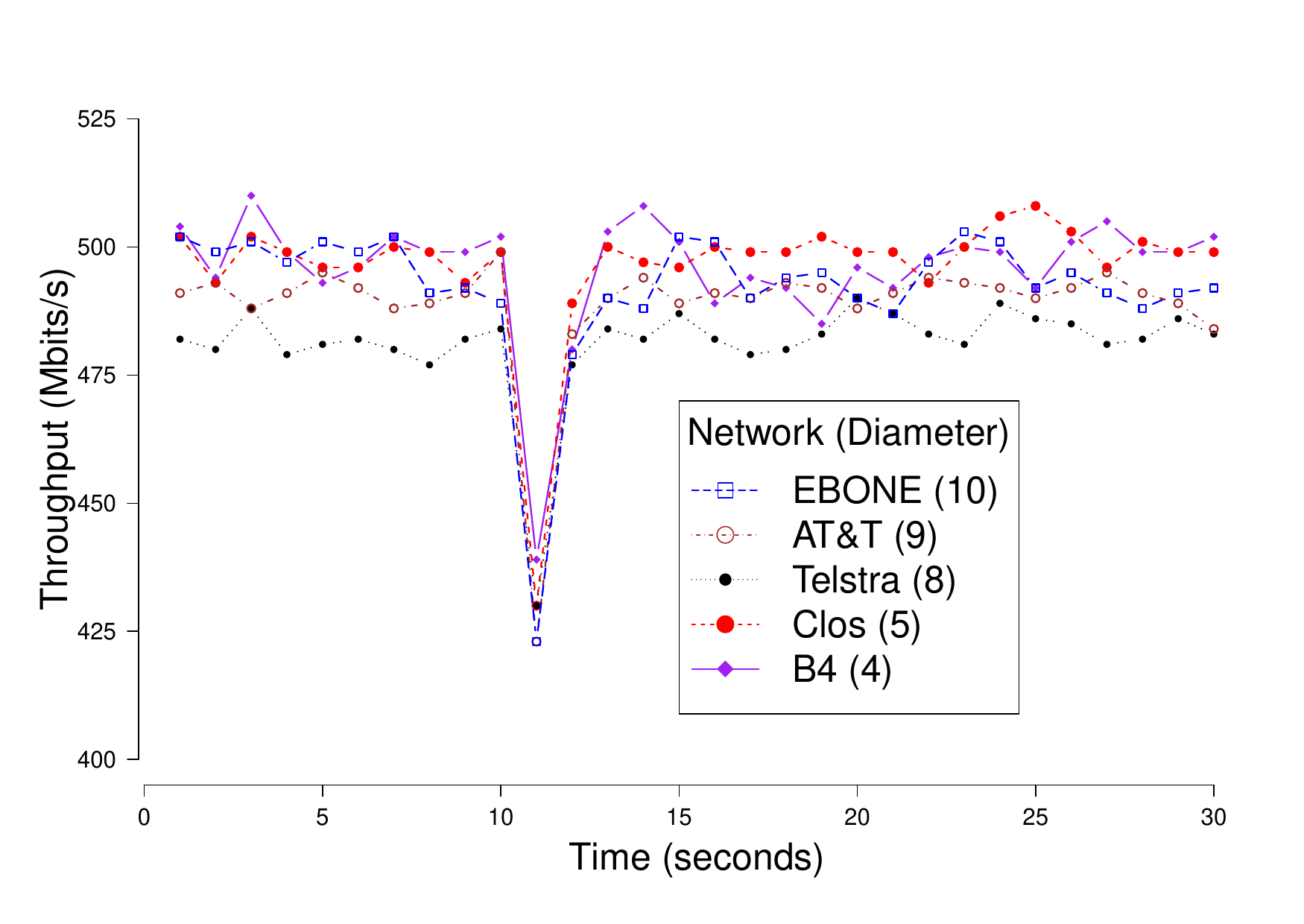}
	\caption{Throughput for the different networks using no recovery after link-failure. Here, a single link failure causes the drop after the $10^{th}$ second.}
	\label{m3}
\end{figure*}

We can see in Figure~\ref{m2} that \emph{one} throughput valley occurs indeed (to around 480 - 510 Mbits/s). For comparison, Figure~\ref{m3} shows the throughput over time without recovery that includes consistent updates~\cite{DBLP:conf/hotnets/ReitblattFRW11}: only the backup paths are used in these experiments, and no new primary paths are calculated or used after the link-failure at the 10th second. The results in figures~\ref{m2} and~\ref{m3} are very similar: there is a strong correlation between these two methods in terms of performance, see Table~\ref{t1}.
\begin{wrapfigure}{r}{0.325\textwidth}
	\centering
	\begin{\VCalgSize}
		\vspace*{-1.2em}
		\begin{tabular}{|c|c|}
			\hline
			Network & Correlation \\
			\hline\hline
			Clos & 0.94 \\
			\hline
			B4 & 0.95 \\
			\hline
			Telstra & 0.92 \\
			\hline
			EBONE & 0.96 \\
			\hline
			Exodus & 0.94 \\
			\hline
		\end{tabular}
		\caption{\label{t1}Correlation coefficient of the average throughput for the experiments in Figure~\ref{m2} and Figure~\ref{m3}.}
		\vspace*{-1.35em}
	\end{\VCalgSize}
\end{wrapfigure}


In order to gain more insights, we used Wireshark~\cite{Orebaugh2007} for investigating the number of re-transmissions (after the link-failure) for Telstra, AT\&T and EBONE network topologies. We observed an increase in the packets sent at the 11th second (after the link-failure) are re-transmissions (Figure~\ref{m5}) and ``BAD TCP'' flags (Figure~\ref{fig:badtcp}). This increase was from levels of below 1\% to levels of between 10\% and 15\% and it quickly deescalated. We have also observed a much smaller presence of out-of-order packets (Figure \ref{fig:order}). We observe that these phenomena (and the slight irregularity in the throughput) are related to TCP congestion control protocol, which is TCP Reno~\cite{DBLP:journals/ton/PadhyeFTK00}. Indeed, whenever congestion is suspected, Reno's fast recovery mechanism divides the current congestion window by half (when skipping the slow start mechanism).

\begin{figure*}[t!]
	\centering
	\includegraphics[width=\figSizeJournal\textwidth]{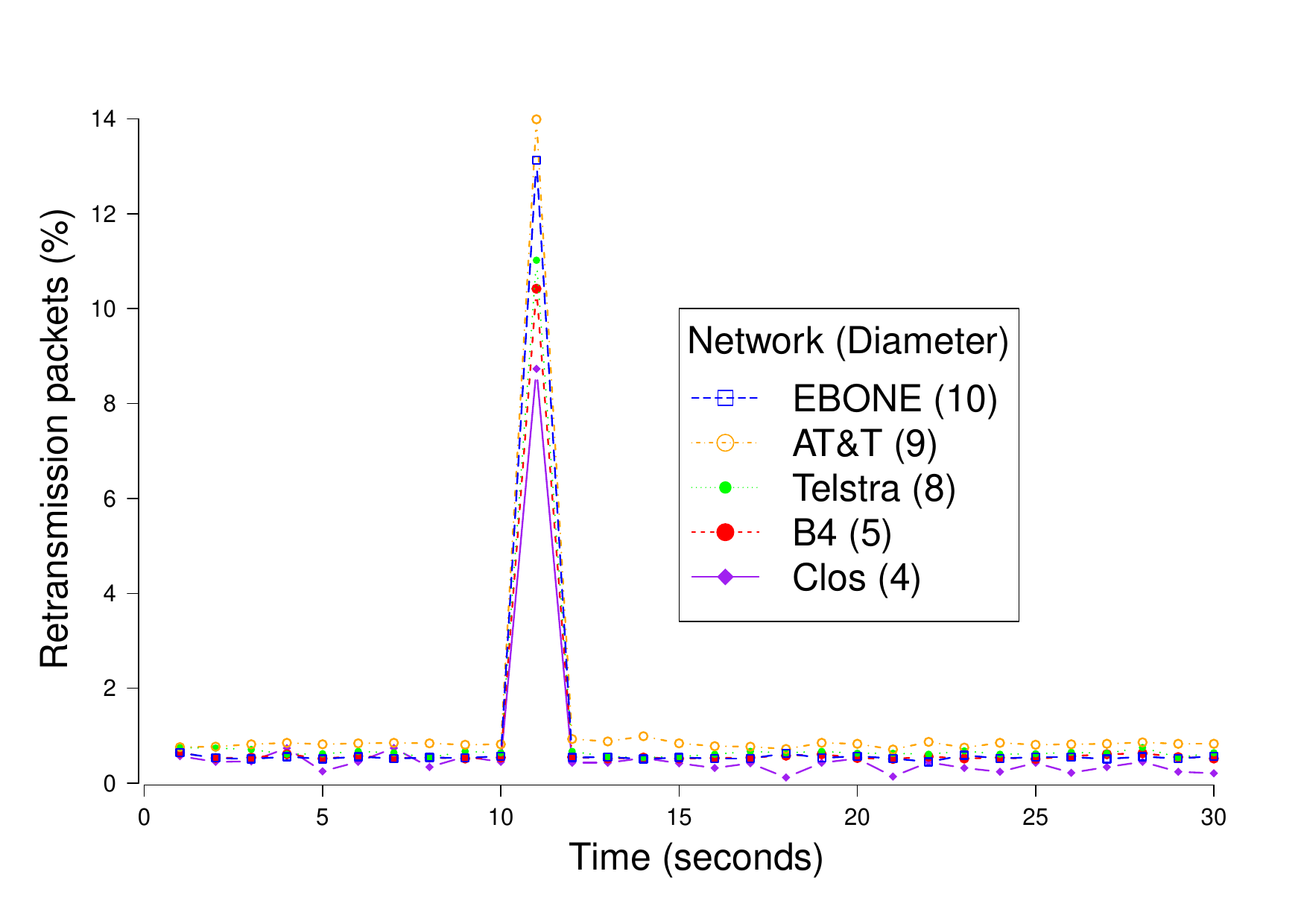}
	\caption{Retransmission percentage rate for packets sent at each second.}
	\label{m5}
\end{figure*}

\begin{figure}[t!]
	\includegraphics[width=\figSizeJournal\textwidth]{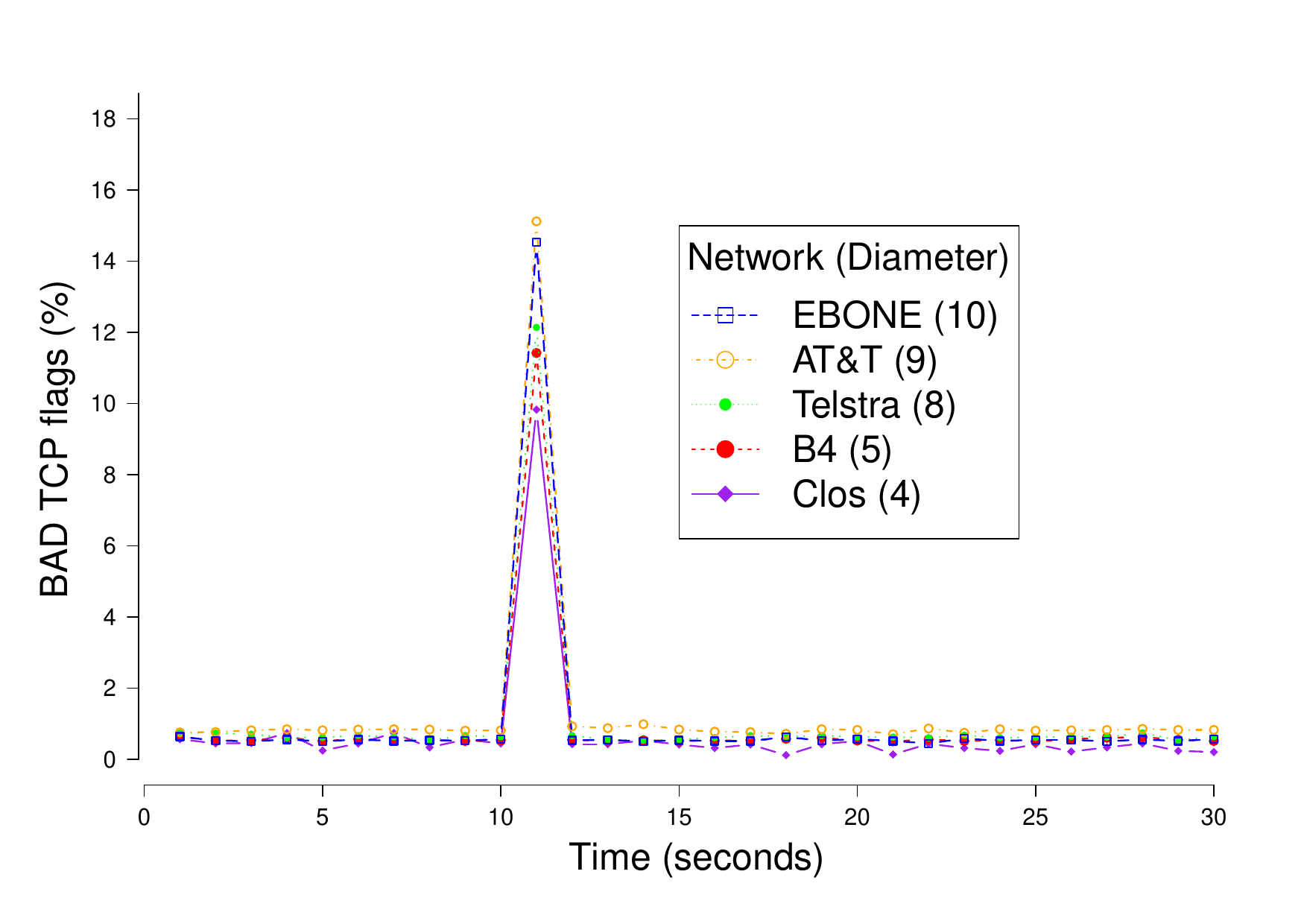}
	\centering
	\captionsetup{justification=centering,margin=2cm}
	\caption{Percentage of ``BAD TCP'' flags during the recovery period that follows a single link failure}
	\label{fig:badtcp}
\end{figure}	

\begin{figure}[t!]
		\includegraphics[width=\figSizeJournal\textwidth]{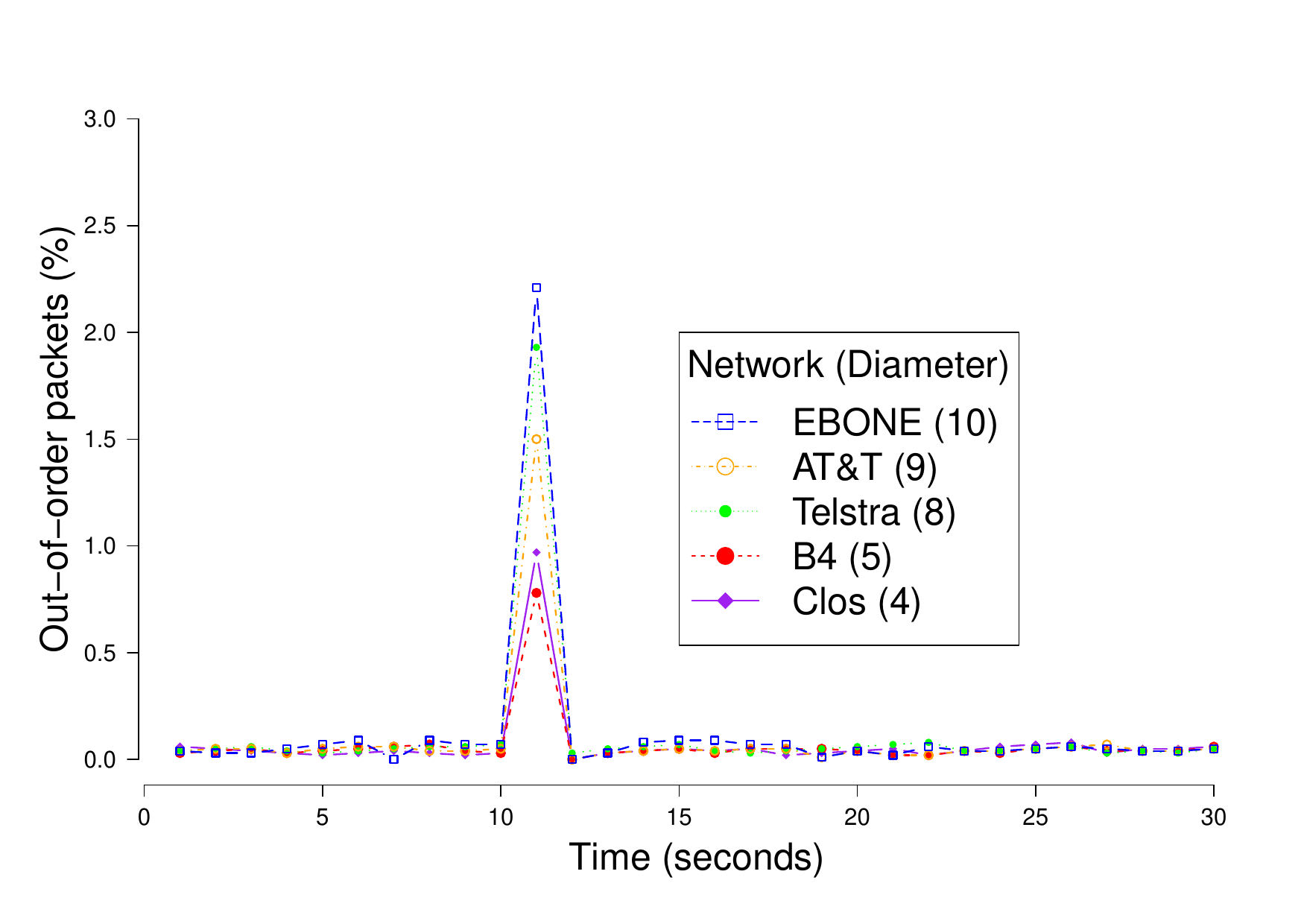}
	\centering
	\caption{Percentage of out-of-order packets during the recovery period that follows a single link failure}
	\label{fig:order}
\end{figure}

\section{Related Work}
\label{sec:relwork}
The design of distributed SDN control planes has been studied intensively in the last few years~\cite{DBLP:conf/sigcomm/BerdeGHHKKLORSP14,onix,elasticon,beehive,kandoo,Canini.INFOCOM15,hotsdn13loc}; both for performance and robustness reasons. While we are not aware of any existing solution for our problem (supporting an in-band and distributed network control), there exists interesting work on bootstrapping connectivity in an OpenFlow network~\cite{Sharma.IEEE16,DBLP:conf/aintec/KatiyarPGK15} that does not consider self-stabilization. In contrast to our paper, Sharma~et al.~\cite{Sharma.IEEE16} do not consider how to support multiple controllers nor how to establish the control network. Moreover, their approach relies on switch support for traditional STP and requires modifying DHCP on the switches. We do consider multiple controllers and establish an in-band control network in a self-stabilizing manner. Katiyar et al.~\cite{DBLP:conf/aintec/KatiyarPGK15} suggest bootstrapping a control plane of SDN networks, supporting multiple controller associations and also non-SDN switches. However, the authors do not consider fault-tolerance. We provide a very strong notion of fault-tolerance, which is self-stabilization.

To the best of our knowledge, our paper is the first to present a comprehensive model and rigorous approach to the design of in-band distributed control planes providing self-stabilizing properties. As such, our approach complements much ongoing, often more applied, related research. In particular, our control plane can be used together with and support distributed systems such as ONOS~\cite{DBLP:conf/sigcomm/BerdeGHHKKLORSP14}, ONIX~\cite{onix}, ElastiCon~\cite{elasticon}, Beehive~\cite{beehive}, Kandoo~\cite{kandoo}, STN~\cite{Canini.INFOCOM15} to name a few. Our paper also provides missing links for the interesting work by Akella and Krishnamurthy~\cite{DBLP:conf/hotnets/AkellaK14}, whose switch-to-controller and controller-to-controller communication mechanisms rely on strong primitives, such as consensus protocols, consistent snapshot and reliable flooding, which are not currently available in OpenFlow switches. \ems{Our proposal does not utilize consensus or flooding, as in~\cite{DBLP:conf/hotnets/AkellaK14}. In other words, the proposed solution requires less than that of~\cite{DBLP:conf/hotnets/AkellaK14} from the underlying system, \eg, we do not assume synchrony. Also, unlike the proposal in~\cite{DBLP:conf/hotnets/AkellaK14}, our proposal is self-stabilizing and includes both algorithmic and empirical analysis.}

We also note that our approach is not limited to a specific technology, but offers flexibilities and can be configured with additional robustness mechanisms, such as warm backups, local fast failover~\cite{fattire}, or alternatives spanning trees~\cite{hotsdn14failover, merav}.  

Our paper also contributes to the active discussion of which functionality can and should be implemented in OpenFlow. DevoFlow~\cite{devoflow} was one of the first works proposing a modification of the OpenFlow model, namely to push responsibility for most flows to switches and adding efficient statistics collection mechanisms. SmartSouth~\cite{hotnets14inband} shows that in recent OpenFlow versions, interesting network functions (such as anycast or network traversals) can readily be implemented in-band. More closely related to our paper,~\cite{ccr16sync} shows that it is possible to implement atomic read-modify-write operations on an OpenFlow switch, which can serve as a powerful synchronization and coordination primitive also for distributed control planes; however, such an atomic operation is not required in our system: a controller can claim a switch with a simple write operation. In this paper, we presented a first discussion of how to implement a strong notion of fault-tolerance, namely a self-stabilizing SDN~\cite{D2K,DBLP:journals/cacm/Dijkstra74}.

We are not the first to consider self-stabilization in the presence of faults that are not just transient faults (see~\cite{D2K}, Chapter 6 and references therein). Thus far, these self-stabilizing algorithms consider networks in which all nodes can compute and communicate. In the context of the studied problem, some nodes, i.e., the switches, can merely forward packets according to rules that are decided by other nodes, i.e., the controllers. To the best of our knowledge, we are the first to demonstrate a rigorous proof for the existence of self-stabilizing algorithms for an SDN control plane. This proof uses a number of techniques, such as the one for assuring a bounded number of resets and illegitimate rule deletions, that were not used in the context of self-stabilizing bootstrapping of communication (to the best of our knowledge).

~\\
\noindent \textbf{Bibliographic note.~~}
We reported on preliminary insights on the design of in-band control planes in two short papers on \emph{Medieval}~\cite{ccr16sync,disn16medieval}. However, Medieval is not self-stabilizing, because its design depends on the presence of non-corrupted configuration data, e.g., related to the controllers' IP addresses, which goes against the idea self-stabilization. A self-organizing version of Medieval appeared in~\cite{DBLP:conf/icdcs/CaniniSSSS17}. $\system$ provides a rigorous algorithm and proof of self-stabilization; it appeared as an extended abstract~\cite{DBLP:conf/icdcs/CaniniSSSS18} and as a technical report~\cite{DBLP:journals/corr/abs-1712-07697}.

\section{Discussion}
\label{sec:conclusion}
While the benefits of the separation between control and data planes have been studied intensively in the SDN literature, the important question of how to connect these planes has received less attention. This paper presented a first model and an algorithm, as well as a detailed analysis and proof-of-concept implementation of a self-stabilizing SDN control plane called  $\system$. 

\subsection{A $\Theta(D)$ stabilization time variation (without memory adaptiveness)}
Before concluding the paper, we would like to point out the existence of a straightforward $\Omega(D)$ lower bound to the studied task to which we match an $\bigO(D)$ upper bound. Indeed, consider the case of a single controller that needs to construct at least one flow to every switch in the network. Starting from a system state in which no switch encodes any rule and the controller is unaware of the network topology, an in-band bootstrapping of this network cannot be achieved within less than $\bigO(D)$ frames, where $D$ is the network diameter (even in the absence of any kind of failure).

We also present a variation of the proposed algorithm that provides no memory adaptiveness. In this variation, no controller ever removes rules installed by another controller (line~\ref{ln:delAllRules}). This variation of the algorithm simply relies on the memory management mechanism of the abstract switches (Section~\ref{sec:memory}) to eventually remove stale rules (that were either installed by failing controllers or appeared in the starting system state). Recall that, since the switches have sufficient memory to store the rules of all controllers in $P_C$, this mechanism never removes any rule of controller $p_i \in P_C$ after the first time that $p_i$ has refreshed its rules on that switch. Similarly, this variation of the algorithm does not remove managers (line~\ref{ln:rm}) nor performs C-resets (line~\ref{ln:resetResponses}). Instead, these sets are implemented as constant size queues and similar memory management mechanisms eventually remove stale set items. We note the existence of bounds for these queues that make sure that they have sufficient memory to store the needed non-failing managers and replies, i.e., $maxManagers$, and respectively, $3 \cdot maxRules$.

Recall the conditions of Lemma~\ref{thm:propagation} that assume no C-resets and illegitimate deletions to occur during the system execution. It implies that the system reaches a legitimate state within  $((\Delta_{comm}+\Delta_{synch})+2)D+1$ frames from the beginning of the system execution. However, the cost of memory use \emph{after stabilization} can be $N_C/n_C$ times higher than the proposed algorithm. We note that Lemma~\ref{thm:propagation}'s bound is asymptotically the same as the recovery time from benign faults (lemmas~\ref{thm:kappaLink} and~\ref{lem:nodeAdditionDeletion}). Theorem~\ref{lem:staleResetIllegalDelete} brings an upper-bound for the proposed algorithm that is  $(((\Delta_{comm}+\Delta_{synch})D+1)\cdot N_S + N_C + 1)$ times larger than the one of the above variance with respect to the period that it takes the system to reach a legitimate state. However, Theorem~\ref{lem:staleResetIllegalDelete} considers arbitrary transient faults, which are rare. Thus, the fact that the recovery time of the proposed memory adaptive solution is longer is relevant only in the presence of these rare faults.

\subsection{Possible extensions}
\label{sec:ext}
We note that the proposed algorithm can serve as the basis for more even advanced solutions. In particular, while we have deliberately focused on the more challenging in-band control scenario only, we anticipate that our approach can also be used in networks which combine both in-band and out-of-band control, e.g., depending on the network sub-regions. Another possible extension can consider the use of a self-stabilizing reconfigurable replicated state machine~\cite{DBLP:journals/jcss/DolevGMS18,DBLP:conf/cscml/DolevGMS18,DBLP:conf/netys/DolevGMS17} for coordinating the actions of the different controllers, similar to ONOS~\cite{DBLP:conf/sigcomm/BerdeGHHKKLORSP14}.

\ems{This work showed how to construct a distributed control plane by connecting every controller to any node in the network. That is, the algorithm defines rules for forwarding control packets between every controller and every node. Note that, once the proposed distributed control plane is up and running, the controllers can collectively define rules for forwarding data packets.} This, for example, can be built using self-stabilizing (Byzantine fault-tolerant) consensus and state-machine replication~\cite{DBLP:conf/netys/DolevLS13,DBLP:conf/edcc/LundstromRS21,DBLP:conf/icdcn/LundstromRS21,DBLP:journals/corr/abs-2201-12880,DBLP:conf/netys/LundstromRS20}.  

~\\
\noindent \textbf{Acknowledgments.~~} Part of this research was supported by 
Vienna Science and Technology Fund (WWTF) project, Fast and Quantitative What-if Analysis for Dependable Communication Networks (WHATIF), ICT19-045, 2020-2024.
We are grateful to Michael Tran, Ivan Tannerud and Anton Lundgren for developing the prototype. 
We are also thankful to Emelie Ekenstedt for helping to improve the presentation. 
Last but not least, we also thank the anonymous reviewers whose comments greatly helped to improve the presentation of the paper.


\bibliographystyle{elsarticle-num}
\bibliography{main}



\end{document}

%% file: imports.tex
\usepackage{graphicx}
\usepackage{balance}

\usepackage[usenames,dvipsnames]{color}
\usepackage{xcolor}
\definecolor{pennblue}{cmyk}{1,.65,0,.3}
\definecolor{pennred}{cmyk}{0,1,.65,.34}

\usepackage{amsmath,amssymb}
\usepackage{multirow}
\usepackage{soul}

\usepackage[colorlinks=true,linkcolor={pennblue},citecolor={pennred},urlcolor={pennblue}]{hyperref}
\usepackage{url}

\usepackage{subcaption}
\usepackage{xspace}
\usepackage{framed}
\usepackage{amsmath}
\usepackage{amsthm}
\usepackage[linesnumbered,ruled,vlined]{algorithm2e}

\newtheorem{theorem}{Theorem}
\newtheorem{lemma}{Lemma}

\newtheorem{corollary}{Corollary}
\newtheorem{definition}{Definition}

\newcommand{\system}{\emph{Renaissance}}

\SetAlCapNameFnt{\small}
\SetAlCapFnt{\small}
\SetAlCapNameSty{textbf}

\newcommand{\bigO}{O}

\newcommand{\Subsection}[1]{\subsection{#1}}
\newcommand{\Subsubsection}[1]{\subsubsection{#1}}

%% file: macros.tex
\def\NOTES{0}
\newcommand{\smartparagraph}[1]{\noindent{\bf #1}\ }
\newcommand{\eg}{{\it e.g.}}
\newcommand{\ie}{{\it i.e.}}
\newcommand{\etc}{{\it etc.}}
\newcommand{\etal}{{\it et al.}\xspace}

\newcommand{\Mng}{\textsc{Mng}\xspace}
\newcommand{\MngDist}{\textsc{MDist}\xspace}
\newcommand{\MngPath}{\textsc{{MPath}\xspace}}
\newcommand{\timeout}{\tau\xspace}
\newcommand{\Rules}{\text{R}}
\newcommand{\TCP}{\text{TCP}}
\newcommand{\OF}{\text{OF}}
\newcommand{\MFE}{\tau}
\newcommand{\echo}{\text{echo}}

\newcommand{\superscript}[1]{\ensuremath{^{\textrm{#1}}}}
\def\wu{\superscript{*}}
\def\wg{\superscript{\dag}}
\def\wb{\superscript{\ddag}}

\newcommand{\reaches}{\rightarrow_G}
\newcommand{\cO}{\mathcal{O}}
\newcommand{\dis}{fusion}

\definecolor{shadecolor}{rgb}{0.9,0.9,0.9}
\definecolor{heraldBlue}{rgb}{0.0,0.0,0.8}
\definecolor{heraldRed}{rgb}{0.8,0.0,0.0}
\definecolor{heraldGray}{rgb}{0.4,0.4,0.4}
\definecolor{heraldBlack}{rgb}{0.0,0.0,0.0} 
\definecolor{heraldGreen}{rgb}{0.0,0.4,0.0} 
\def\r#1{\textcolor{heraldBlue}{\em #1}}
\def\q#1{\textcolor{heraldRed}{\em #1}}
\def\d#1{\textcolor{heraldBlue}{#1}}
\def\R#1{\textcolor{heraldBlue}{#1}}
\def\D#1{\textcolor{heraldBlue}{#1}}

\newcommand{\figureSizeSwitch}[2]{#1}
 
\newcommand{\removed}[1]{}
\newcommand{\blue}{\textcolor{blue}}

\if \NOTES 1
  \newcommand{\mcnote}[1]{\textcolor{heraldBlue}{\small \textbf{[MC: #1]}}}
  \newcommand{\stefan}[1]{\textcolor{heraldGreen}{\small \textbf{[Stefan: #1]}}}
  \newcommand{\nlnote}[1]{\textcolor{Purple}{\small \textbf{[NL: #1]}}}
  \newcommand{\liron}[1]{\textcolor{heraldRed}{\small \textbf{[LS: #1]}}}
	\newcommand{\IS}[1]{\textcolor{red}{[\textbf{IS: #1}]}}
	\newcommand{\EMS}[1]{#1}
	\newcommand{\ems}[1]{\textcolor{blue}{#1}}
	\newcommand{\modified}[2]{\blue{#2}} 
\else
  \newcommand{\mcnote}[1]{}
  \newcommand{\stefan}[1]{}
  \newcommand{\nlnote}[1]{}
    \newcommand{\liron}[1]{}
	\newcommand{\IS}[1]{}
	\newcommand{\EMS}[1]{}
	\newcommand{\ems}[1]{#1}
	\newcommand{\modified}[2]{#2} 
\fi
\newcommand{\specialcell}[2][c]{%
  \begin{tabular}[#1]{@{}l@{}}#2\end{tabular}}

\newtheorem{scenarios}{Scenarios}[section]
\newtheorem{observation}{Observation}
\newtheorem{takeaway}{Takeaway}[section]
\newtheorem{claim}{Claim}[section]

\newcommand{\defconip}{{CIP}}
\newcommand{\contag}{{CTAG}}

\let\oldnl\nl
\newcommand{\nonl}{\renewcommand{\nl}{\let\nl\oldnl}}

\newcommand{\remove}[1]{}

\newcommand{\new}[1]{\blue{#1}}